\documentclass[11pt,a4paper,reqno]{amsart}
\usepackage[hmargin={2.7cm,2.7cm},vmargin={2.5cm,2.5cm},centering]{geometry}
\usepackage{amssymb,amsmath,amsthm,amsfonts,amscd}
\usepackage{graphicx}
\usepackage[dvipsnames]{xcolor}
\usepackage[unicode,psdextra]{hyperref}
\usepackage[utf8]{inputenc}
\usepackage{enumerate}
\hypersetup{pdfborder={0 0 0.06},linkbordercolor=BrickRed,citebordercolor=Fuchsia,urlbordercolor=CornflowerBlue}
\usepackage{tikz-cd}
\usepackage[english]{babel}
\usepackage{dsfont}
\usepackage{url}
\usepackage{setspace}
\usepackage[font={scriptsize,bf}]{caption}
\usepackage{changepage}
\usepackage{booktabs}
\usepackage[sort&compress,capitalise,nameinlink]{cleveref}
\crefname{section}{\textsection}{\textsection}
\crefname{subsection}{\textsection}{\textsection}
\crefname{subsubsection}{\textsection}{\textsection}
\crefname{paragraph}{\textparagraph}{\textparagraph}
\crefname{thm}{Theorem}{Theorems}
\crefname{minprob}{Minimization Problem}{Minimization Problems}
\usepackage{mathrsfs}
\usepackage{bm}
\renewcommand{\vec}[1]{\mathbf{#1}}
\DeclareMathOperator{\tr}{tr}

\renewcommand{\Re}{\mathrm{Re}}
          \newtheorem{thm}{Theorem}[section]
          \newtheorem{proposition}[thm]{Proposition}
          \newtheorem{lemma}[thm]{Lemma}
          \newtheorem{corollary}[thm]{Corollary}
          \newtheorem{definition}[thm]{Definition}
          
          \theoremstyle{definition}

          \newtheorem{remark}[thm]{Remark}

\renewcommand{\emptyset}{\varnothing}

\usepackage{csquotes}
\usepackage[authormarkup=none,
deletedmarkup=xout
]{changes}

\newcommand{\one}{\mathds{1}}
\newcommand{\av}{\mathbf{a}}

\newcommand{\xv}{\mathbf{x}}

\newcommand{\kv}{\mathbf{k}}

\newcommand{\kvp}{\mathbf{k}^{\prime}}

\newcommand{\yv}{\mathbf{y}}

\newcommand{\xxv}{\mathbf{X}}
\newcommand{\yyv}{\mathbf{Y}}

\newcommand{\mm}{\mathfrak{m}}
\newcommand{\nn}{\mathfrak{n}}
\newcommand{\dom}{\mathscr{D}}
\newcommand{\diff}{\mathrm{d}}
\newcommand{\eps}{\varepsilon}

\newcommand{\aav}{\mathbf{A}}

\newcommand{\hfree}{H_{\mathrm{free}}}
\newcommand{\psigs}{\Psi_{\mathrm{gs}}}
\newcommand{\fqc}{\mathcal{E}_{\mathrm{qc}}}
\newcommand{\ffqc}{\mathcal{F}_{\mathrm{qc}}}
\newcommand{\fgqc}{\mathcal{E}_{\mathrm{gqc}}}
\newcommand{\eqc}{E_{\mathrm{qc}}}
\newcommand{\egqc}{E_{\mathrm{gqc}}}

\newcommand{\psiqc}{\psi_{\mathrm{qc}}}
\newcommand{\zqc}{z_{\mathrm{qc}}}
\newcommand{\rhogqc}{\rho_{\mathrm{gqc}}}
\newcommand{\zgqc}{z_{\mathrm{gqc}}}
\newcommand{\zpek}{\eta_{\mathrm{Pekar}}}
\newcommand{\vpek}{\mathcal{V}_{\mathrm{Pekar}}}
\newcommand{\fpek}{\mathcal{E}_{\mathrm{Pekar}}}
\newcommand{\epek}{E_{\mathrm{Pekar}}}
\newcommand{\psipek}{\psi_{\mathrm{Pekar}}}

\newcommand{\domqc}{\dom_{\mathrm{qc}}}
\newcommand{\domgqc}{\dom_{\mathrm{gqc}}}
\newcommand{\domp}{\dom_{\mathrm{Pekar}}}
\newcommand{\domsv}{\dom_{\mathrm{svm}}}
\newcommand{\domgsv}{\dom_{\mathrm{gsvm}}}
\newcommand{\dompm}{\dom_{\mathrm{pm}}}
\newcommand{\fsv}{\mathcal{E}_{\mathrm{svm}}}
\newcommand{\fpm}{\mathcal{E}_{\mathrm{pm}}}
\newcommand{\esv}{E_{\mathrm{svm}}}
\newcommand{\egsv}{E_{\mathrm{gsvm}}}
\newcommand{\epm}{E_{\mathrm{pm}}}
\newcommand{\msv}{\mm_{\mathrm{svm}}}
\newcommand{\nsv}{\nn_{\mathrm{gsvm}}}
\newcommand{\mupm}{\mu_{\mathrm{pm}}}
\newcommand{\psipm}{\psi_{\mathrm{pm}}}

\newcommand{\beq}{\begin{equation}}
\newcommand{\eeq}{\end{equation}}

\newcommand{\OO}{\mathcal{O}}

\newcommand{\HH}{\mathcal{H}}

\newcommand{\GG}{\mathcal{G}}

\newcommand{\CC}{\mathbb{C}}
\newcommand{\R}{\mathbb{R}}

\newcommand{\EE}{\mathcal{E}}
\newcommand{\WW}{\mathscr{W}}
\newcommand{\GW}{\mathscr{GW}}

\newcommand{\bdm}{\begin{displaymath}}
\newcommand{\edm}{\end{displaymath}}
\newcommand{\bdn}{\begin{eqnarray}}
\newcommand{\edn}{\end{eqnarray}}
\newcommand{\bay}{\begin{array}{c}}
\newcommand{\eay}{\end{array}}
\newcommand{\ben}{\begin{enumerate}}
\newcommand{\een}{\end{enumerate}}
\newcommand{\beqn}{\begin{eqnarray}}
\newcommand{\eeqn}{\end{eqnarray}}

\newcommand{\lf}{\left}
\newcommand{\ri}{\right}
\newcommand{\disp}{\displaystyle}
\newcommand{\tx}{\textstyle}

\newcommand{\bra}[1]{\lf\langle #1\ri|}
\newcommand{\ket}[1]{\lf|#1 \ri\rangle}
\newcommand{\braket}[2]{\lf\langle #1|#2 \ri\rangle}
\newcommand{\braketr}[2]{\lf\langle #1\lf|#2\ri. \ri\rangle}
\newcommand{\braketl}[2]{\lf.\lf\langle #1\ri|#2 \ri\rangle}

\newcommand{\meanlrlr}[3]{\lf\langle #1\lf|#2\ri|#3\ri\rangle}

\renewcommand{\leq}{\leqslant}
\renewcommand{\geq}{\geqslant}

\numberwithin{equation}{section}

\begin{document}
\title{Ground State Properties in the Quasi-Classical Regime}

\author[M.\ Correggi]{Michele Correggi}

\address{Dipartimento di Matematica, Politecnico di Milano, Piazza Leonardo da Vinci, 32,
  20133, Milano, Italy.}

\email{michele.correggi@gmail.com}

\urladdr{https://sites.google.com/view/michele-correggi}

\author[M.\ Falconi]{Marco Falconi}

\address{Dipartimento di Matematica e Fisica, Universit\`{a} degli Studi Roma
  Tre, L.go S. Leonardo Murialdo, 1, 00146, Roma, Italy.}

\email{mfalconi@mat.uniroma3.it}

\urladdr{http://ricerca.mat.uniroma3.it/users/mfalconi/}


\author[M.\ Olivieri]{Marco Olivieri}

\address{ Fakult\"{a}t f\"{u}r Mathematik, Karlsruher Institut f\"{u}r
  Technologie, D-76128, Karlsruhe, Germany.}

\email{marco.olivieri@kit.edu}

\urladdr{}

\keywords{Quasi-classical limit; Interaction of matter and light;
  Semiclassical analysis.}


\date{\today}

\begin{abstract}
  We study the ground state energy and ground states of systems coupling
  non-relativistic quantum particles and force-carrying Bose fields, such as
  radiation, in the quasi-classical approximation. The latter is very useful
  whenever the force-carrying field has a very large number of excitations,
  and thus behaves in a semiclassical way, while the non-relativistic
  particles, on the other hand, retain their microscopic features. We prove
  that the ground state energy of the fully microscopic model converges to
  the one of a nonlinear quasi-classical functional depending on both the
  particles' wave function and the classical configuration of the
  field. Equivalently, this energy can be interpreted as the lowest energy of
  a Pekar-like functional with an effective nonlinear interaction for the
  particles only. If the particles are confined, the ground state of the
  microscopic system converges as well, to a probability measure concentrated
  on the set of minimizers of the quasi-classical energy.
\end{abstract}

\maketitle

\onehalfspacing{}

\tableofcontents

\section{Introduction and Main Results}
\label{sec: intro}

The description and rigorous derivation of effective models for complex
quantum systems is a flourishing line of research in modern mathematical
physics. Typically, in suitable regimes, the fundamental quantum description
can be approximated in terms of some simpler model retaining the salient
physical features, but also allowing a more manageable computational or
numerical treatment. The questions addressed in this work naturally belong to
such a wide class of problems.

We consider indeed a quantum system composed of $ N $ non-relativistic
particles interacting with a quantized bosonic field, in the {\it
  quasi-classical regime}. We refer to the series of works
\cite{correggi2017ahp,correggi2017arxiv,carlone2019arxiv,correggi2019arxiv}
for a detailed discussion of such a regime: in extreme synthesis, we plan to
study field configurations with a suitable semiclassical behavior. We require
indeed that there is a large number of field excitations, although each one
of the latter is carrying a very small amount of energy, in such a way that
the field's degrees of freedom are almost classical. More precisely, we
assume that the average number of force carriers $ \langle \mathcal{N} \rangle $ is of
order $ \frac{1}{\eps} $, for some $ 0 < \eps \ll 1 $, and thus much larger
than the commutator between $ a^{\dagger} $ and $ a $, which is of order $ 1 $ (we
use units in which $ \hbar = 1 $). Concretely, this can be realized by rescaling
the canonical variables $ a^{\dagger}, a $ by $ \sqrt{\eps} $, {\it i.e.}, setting
$ a^{\sharp}_{\eps} : = \sqrt{\eps} a^{\sharp} $, which leads to
\beq \lf[a_{\eps}(\kv), a^{\dagger}_{\eps}(\kvp) \ri] = \eps \delta(\kv- \kvp), \qquad \eps \ll 1.  \eeq
On the other hand, the degrees of freedom associated with the particles are
not affected by the scaling limit $ \eps \to 0 $ and the particles remain
quantum. Our goal is precisely to set up and rigorously derive an effective
quantum model for the lowest energy state of the system in the
quasi-classical regime $ \eps \to 0$, when the field becomes classical.

Let us now describe in more detail the type of microscopic models we plan to
address. The space of states of the full system is\footnote{We do not take
  into account the spin degrees of freedom nor the symmetry constraints
  induced by the presence of identical particles, but such features can be
  included in the discussion without any effort and the results trivially
  apply to the corresponding models. In fact, we may even allow for a
  coupling term between the radiation field and the particle spins
  \cite{correggi2017arxiv}, as the one often included in the Pauli-Fierz
  model.}
\beq
\label{eq: hilbert}
\mathscr{H}_{\eps} : = L^{2}(\R^{dN}) \otimes \GG_{\eps}(\mathfrak{h}),
\eeq
where $ d \in \{1,2,3\} $, $ \mathfrak{h} $ is the single one-excitation space
of the field and $ \GG_{\eps} $ stands for the second quantization map, so
that $ \GG_{\eps}(\mathfrak{h}) $ is the bosonic Fock space constructed over
$ \mathfrak{h} $, with canonical commutation relations
\beq
\label{eq: ccr}
\lf[ a_{\eps}(\xi), a^{\dagger}_{\eps}(\eta) \ri] = \eps \braketl{\xi}{\eta}_{\mathfrak{h}},
\eeq
for any $ \xi,\eta \in \mathfrak{h} $.

The energy of the microscopic system and thus its Hamiltonian is given by the
non-relativistic energy of the particles, the field energy and the
interaction between the particles and the field, in such a way that
\begin{itemize}
\item the particle and field energies are a priori of the same order $ \OO(1)
  $;
\item the interaction is weak, {\it i.e.}, a priori subleading w.r.t.\ the
  unperturbed energies.
\end{itemize}
This is concretely realized by considering Hamiltonians of the form
\beq
\label{eq: hamiltonian}
H_{\varepsilon}=\mathcal{K}_0 \otimes 1 + 1 \otimes \diff \GG_{\eps}(\omega) + H_I,
\eeq
where:
\begin{itemize}	
\item $ \mathcal{K}_0 $ is the ($ \eps$-independent) free particle Hamiltonian
\beq
	\label{eq: K0}
  	\mathcal{K}_0 = \sum_{j = 1}^N \lf( - \Delta_j \ri) + \mathcal{W}(\xv_1, \ldots, \xv_N)
\eeq
 which is assumed to be self-adjoint and bounded from below (we specify in
  \cref{sec:concrete-models} the working assumptions on $ \mathcal{W} $); 
 
\item $ \diff \GG_{\eps}(\omega) $ is the free field energy and is the second
  quantization of the positive operator $ \omega $ on $ \mathfrak{h} $, admitting
  possibly unbounded inverse $ \omega^{-1} $;
\item the interaction $ H_I $ is the only non-factorized term of the
  Hamiltonian, it depends on $ \eps $ only through the creation and
  annihilation operators $ a^{\sharp}_{\eps} $ and it is a polynomial of such
  operators of order between one and two.
\end{itemize}
Such requests meet the scaling conditions mentioned above. Indeed, assuming
that the average number $ \langle \mathcal{N} \rangle $ of bare excitation of the field
is $ \OO(\eps^{-1}) $, the field energy is of order $ \eps \langle
\mathcal{N} \rangle =  \OO(1) $, due to the rescaling of $ a_{\eps}^{\dagger}$ and $a_{\eps} $. For the same reason
and since the interaction is at least of order one in the creation and
annihilation operators, we have that $
H_I $ is of order $ \OO(\sqrt{\eps}) $, {\it i.e.}, a
priori subleading w.r.t. the rest of $ H_{\eps} $.

The specific models we are considering in the following are:
\begin{enumerate}[\;\;\;\;\;\;\;(a)]

\item the {\it Nelson model} \cite{nelson1964jmp}: the coupling in $ H_I $ is
  simply linear, {\it i.e.},
  \beq
  \label{eq: nelson coupling}
  H_I = \sum_{j=1}^N A_{\eps}(\xv_j),
  \eeq
  where
  \beq
  \label{eq: field operator}
  A_{\varepsilon}(\xv) := a^{\dagger}_{\varepsilon}\bigl(\vec{\lambda}(\xv)\bigr) +
  a_{\varepsilon}\bigl(\vec{\lambda}(\xv)\bigl)
  \eeq
  is the field operator and
  \beq
  \label{eq: lambda N}
  \lambda, \: \omega^{-1/2} \lambda \in L^{\infty}(\R^{3}; \mathfrak{h})
  \eeq
  (a typical choice is $
  \mathfrak{h} = L^2(\R^d) $, $ \omega $ the multiplication operator by $ \omega(\kv) \geq
  0 $
  and $ \lambda(\xv; \kv) = \lambda_0(\kv) e^{- i \kv \cdot \xv} $, with $ \lambda_0, \omega^{-1/2}
  \lambda_0 \in \mathfrak{h} $);
\item the {\it Fr\"{o}hlich polaron} \cite{frohlich1937prslA}: it is a
  variant of the Nelson model where $ \mathfrak{h} = L^2(\R^d) $, $ \omega = 1 $
  and
  \beq
  \label{eq: lambda polaron}
  \lambda(\xv; \kv) = \sqrt{\alpha} \frac{e^{-i \kv\cdot \xv}}{\lvert \kv \rvert_{}^{\frac{d-1}{2}}},
  \eeq
  for some $ \alpha > 0 $;
\item the {\it Pauli-Fierz model} \cite{pauli1938nc}: it is the most
  elaborate model and we consider only its three-dimensional realization,
  namely $ d = 3 $; the interaction is provided by the minimal coupling
  \beq
  \label{eq: hamiltonian PF}
  H_{\eps} = \sum_{j=1}^N \tx\frac{1}{2m_j}{\lf(-i\nabla_j + e \aav_{\varepsilon,j}(\xv_j)
    \ri)}^2 + \mathcal{W}(\xv_1,\dotsc, \xv_N) + 1 \otimes \diff \GG_{\eps}(\omega),
  \eeq
  where $ \omega \geq 0 $, $ m_j > 0 $, $ j = 1, \ldots, N $ and $e$ are the particles' masses and charge, respectively, and
  the field operators $ \aav_{\eps,j} $, $ j =1, \ldots, N $, have here the same
  formal expression as in \eqref{eq: field operator} but $ \bm{\lambda}_j =
  (\lambda_{j,1}, \lambda_{j,2}, \lambda_{j,3}) $, with
  \beq
  \label{eq: lambda PF}
  \lambda_{j,\ell}, \: \omega^{\pm1/2} \lambda_{j,\ell} \in L^{\infty}(\mathbb{R}^3; \mathfrak{h})\;,
  \eeq
  is a vector function to account for the electromagnetic polarizations and
  the charge distributions of the particles (the standard choice is, indeed, $\mathfrak{h} = L^2(\mathbb{R}^3;\mathbb{C}^2 )$) and we fix for convenience the
  gauge to be the Coulomb's one, {\it i.e.}, $ \nabla_j \cdot \bm{\lambda}_j = 0 $.
\end{enumerate}
The physical meaning of the three models above is quite different and we
refer, {\it e.g.}, to the monograph \cite{spohn2004dcp} for a detailed
discussion. The Nelson model is the simplest one and can be applied to model nucleons interacting with a meson
field or, in first
approximation, to model the interaction of particles with radiation fields,
although the case of the electromagnetic field is typically described through
the Pauli-Fierz model. The polaron, on the other hand, provides an effective description of
quantum particles in a phonon field, {\it e.g.}, generated by the vibrational
models of a crystal. Note also that the quasi-classical limit $ \eps \to 0 $
itself can have different interpretations in each model. For instance, in the
framework of the polaron model, it can be reformulated as a {\it strong
  coupling limit}, which has recently attracted a lot of attention (see, {\it
  e.g.}, \cite{griesemer2016arxiv2, frank2019nonlinear, lieb2019arxiv,
  leopold2020derivation, mitrouskas2020arxiv} and references
therein).
	
In the Nelson and Pauli-Fierz Hamiltonians, there is an ultraviolet
regularization, made apparent in the assumptions on $ {\lambda} $; we do not
consider here the renormalization procedure to remove such ultraviolet
cut-off, even if for the Nelson model it is possible to perform it
rigorously. We plan to address such a problem in a future work. We also skip
at this stage the discussion of the well-posedness of such models (see
\cref{sec:nelson-model}, \cref{sec:polaron-model} and \cref{sec:pauli-fierz-model} for further details), but we point out that,
with the assumptions made, the operator \eqref{eq: hamiltonian} is self-adjoint and bounded from
below in each model.
	
The main problem we study concerns the behavior of the ground state of the
microscopic Hamiltonian $ H_{\eps} $ in the quasi-classical limit $ \eps \to 0
$ and, more precisely, we investigate the convergence in the same limit of
the bottom of the spectrum
\beq
\label{eq: eeps}
E_{\eps} : = \inf \sigma(H_{\eps}) = \inf_{\Gamma_{\eps} \in
  \mathscr{L}^1(\mathscr{H}_{\eps}), \lf\| \Gamma_{\eps} \ri\|_1 = 1} \tr \lf(H_{\eps} \Gamma_{\eps} \ri)
\eeq
of $ H_{\eps} $ as well as the approximation of any corresponding {\it approximate ground state} or {\it  minimizing
sequence} $ \Psi_{\varepsilon, \delta} \in \mathscr{D}(H_{\varepsilon}) $ satisfying
\begin{equation}
  \label{eq:1}
  \langle  \Psi_{\varepsilon, \delta} \vert H_{\varepsilon}\vert \Psi_{\varepsilon, \delta} \rangle_{\mathscr{H}_{\varepsilon}}< E_{\varepsilon}+\delta\; ,
\end{equation}
for some small $ \delta > 0 $.
	
The quasi-classical counterparts of such quantities are determined via the
minimization of a suitable coupled problem, where the particle's degrees of
freedom are driven by a classical field. Such a problem is described in
detail in \cref{sec: quasi-classical var} below. The quasi-classical energy
is given by a functional $ \fqc[\psi,z] $ (see \eqref{eq: fqc} below), depending
on the particle's wave function $ \psi \in L^2(\R^{dN}) $ and on the classical
field configuration\footnote{The space $\mathfrak{h}_{\omega}$ is constructed
  starting from $\mathfrak{h}$ and the dispersion relation $\omega$ of the
  semiclassical field; see \eqref{eq:3} for a precise
  definition. It is necessary to use $\mathfrak{h}_{\omega}$ in place of
  $\mathfrak{h}$ as the field's configuration space whenever the field is massless,
  such as in the Pauli-Fierz model or in the massless Nelson model. For
  massive fields, $\mathfrak{h}_{\omega}\subseteq \mathfrak{h}$.} $ z \in
\mathfrak{h}_{\omega} $. Denoting by $ \eqc $ and $ \lf( \psiqc, \zqc \ri) \in
L^2(\R^{dN}) \oplus \mathfrak{h}_{\omega} $ the infimum of a such a quasi-classical
energy and the relative minimizing configuration (if any), respectively, our
main results are:
\begin{enumerate}[i)]
\item {\bf Energy convergence}. Both the quantum and the quasi-classical
  problems are stable, {\it i.e.}, $ E_{\eps}, \eqc > -\infty $ and
  \beq
  E_{\eps} \xrightarrow[\eps \to 0]{} \eqc.
  \eeq	
\item {\bf Convergence of ground states and approximate ground states}. Assuming that the operator $ \mathcal{K}_0 $ has compact resolvent, then any limit point
  of an approximate ground state $\Psi_{\varepsilon, \delta}$ in the sense of {\it quasi-classical Wigner
    measures} is a an approximate ground state of the quasi-classical functional $
  \fqc $,
  in a sense to be clarified in \cref{thm:2} below. Furthermore, any
  limit point of the family of approximate ground states $\Psi_{\varepsilon, o_{\varepsilon}(1)}$ is concentrated on the set
  of minimizers of the quasi-classical functional $\fqc$; since the set of
  limit points is never empty, the latter admits {\it at least one minimizer} in
  $
  L^2(\R^{dN}) \oplus \mathfrak{h}_{\omega} $. If $H_{\varepsilon}$ has a ground state $\psigs$, then any of
  its limit points in the sense of quasi-classical Wigner measures is
  concentrated on the set of minimizers of $\fqc$.
\item {\bf Generalized convergence of  ground states and approximate ground states}. If the
  operator $ \mathcal{K}_0 $ does not have compact resolvent, then any limit
  point of $ \Psi_{\varepsilon,\delta} $ is a {\it generalized quasi-classical Wigner measure},
  and it is a minimizing sequence for a suitable generalization $ \fgqc $ of
  the energy $ \fqc $. Furthermore, any limit point of $\Psi_{\varepsilon,o_{\varepsilon}(1)}$ is a
  minimizer for $\fgqc$. Let us remark that this does not imply the existence
  of a minimizing configuration $ \lf( \psiqc, \zqc \ri) \in L^2(\R^{dN}) \oplus
  \mathfrak{h}_{\omega} $. If $H_{\varepsilon}$ has a ground state $\psigs$, then any of its limit points in the
  sense of {\it generalized quasi-classical Wigner measures} is concentrated
  on the set of minimizers of $\fgqc$.
\end{enumerate}

We state the above results in all details in \cref{sec: ground state} together with a precise
definition of the notions of quasi-classical Wigner measure and generalized
quasi-classical Wigner measure and the relative topologies. In the next
\cref{sec: quasi-classical var}, we first introduce and discuss the quasi-classical variational problems. 
In the rest of the paper, we present the proofs of the above results. We
stress that the main techniques we are going to use belong to the framework
of semiclassical analysis in infinite dimensional spaces, which was
introduced in the series of works \cite{ammari2008ahp, ammari2009jmp,
  ammari2011jmpa, ammari2015asns} and further discussed in
\cite{falconi2017ccm, falconi2017arxiv}. Apart from the aforementioned works
on quasi-classical analysis, semiclassical techniques have already been used
in the study of variational problems, both for systems with creation and
annihilation of particles \cite{ammari2014jsp}, and for systems with many
bosons, using a slightly different approach called quantum de Finetti theorem
(see \cite{lewin2014am,lewin2015amrx,lewin2016tams}, and references therein
contained). We also point out that partially classical regimes have already
been explored in the literature in \cite{ginibre2006ahp, amour2015arxiv,
  amour2017arxiv, amour2017jmp}, although in other contexts and with
different purposes.

\subsection{Quasi-classical variational problems}
\label{sec: quasi-classical var}
		
As discussed in detail in the series of works \cite{correggi2017ahp,
  correggi2017arxiv, correggi2019arxiv}, each of the microscopic models
introduced so far admits a quasi-classical counterpart in the limit $ \eps \to
0 $.
More precisely, both their stationary \cite{correggi2017ahp,
  correggi2017arxiv} and dynamical \cite{correggi2019arxiv} properties can be
approximated in such a regime in terms of effective models, where the quantum
particle system is driven by a classical field, which in turn is the
classical counterpart of the quantized field. In extreme synthesis, the
quantum field operator gets replaced by a classical field, which is just a
function on $ \R^d $, and the interaction term $ H_I $ in $ H_{\eps} $ gives
rise to a potential $ \mathcal{V}_z $ depending on the classical field
configuration $ z \in \mathfrak{h} $. Concretely, the quasi-classical effective
Hamiltonian reads
\beq
\label{eq: classical hamiltonian}
\mathcal{H}_z = \mathcal{K}_0 + \sum_{j = 1}^N \mathcal{V}_z(\xv_j) + \meanlrlr{z}{\omega}{z}_{\mathfrak{h}},
\eeq
and it is self-adjoint on some dense $ \dom \subset
L^2(\R^{dN}) $ for any $ z \in \mathfrak{h} $ (see \cite[Thms.\
2.1--2.3]{correggi2017ahp} and \cite[Thm.\ 1.1]{correggi2017arxiv}). In each
model the explicit expression of such an effective potential can be
identified explicitly:
\begin{enumerate}[\;\;\;\;\;\;\;(a)]	
\item in the Nelson model, each particle feels a potential of the form
  \beq
  \label{eq: potential N}
  \mathcal{V}_z(\xv) = 2\Re \braket{z}{\lambda(\xv)}_{\mathfrak{h}} \in \mathscr{B}(L^2(\R^d));
  \eeq
\item for the polaron, the formal expression of the potential $ \mathcal{V}_z
  $ is the same as
  in \eqref{eq: potential N} above, although, since \eqref{eq: lambda polaron} does not belong to $ L^{\infty}(\R^d; \mathfrak{h})
  $, the
  expression on the r.h.s.\ must be interpreted in the proper way (see
  \cref{sec:polaron-model}); in addition, the obtained potential is no longer
  bounded but it is infinitesimally form-bounded w.r.t.\ $ - \Delta $;
\item in the Pauli-Fierz model, the effective operator is obtained via the
  replacement of the field $ \aav_{\eps} $ by its classical counterpart $
  \av_z(\xv) = 2 \Re \braket{z}{\bm{\lambda} (\xv)}_{\mathfrak{h}} $,
  which is continuous and vanishing at $ \infty $ (see
  \cite[Rmk. 1.5]{correggi2017arxiv}), and thus, in order to recover the
  expression \eqref{eq: classical hamiltonian}, $ \mathcal{V}_z $ must be the
  operator
  \beq
  \label{eq: potential PF}
  \mathcal{V}_z(\xv) = 2 \sum_{j = 1}^N \frac{1}{m_j} \lf[ - i e \Re
  \braket{z}{\bm{\lambda}_j(\xv)}_{\mathfrak{h}} \cdot \nabla_j + e^2 \lf( \Re
  \braket{z}{\bm{\lambda}_j(\xv)}_{\mathfrak{h}} \ri)^2 \ri].
  \eeq		
\end{enumerate}
Note that in case (c) the effective operator can in fact be simply rewritten
as\footnote{We use the compact notation $ \xxv : = \lf(\xv_1, \ldots, \xv_N \ri) \in
  \R^{dN} $.}
\beq
\label{eq: classical hamiltonian PF}
\mathcal{H}_z = \sum_{j = 1}^N \tx\frac{1}{2m_j} \lf( - i \nabla_j + e \av_z(\xv_j) \ri)^2 + \mathcal{W}(\xxv) + \meanlrlr{z}{\omega}{z}_{\mathfrak{h}}.
\eeq
	
We can now define the effective quasi-classical ground state energy in terms
of the energy functional
\beq
\label{eq: fqc}
\fqc[\psi, z] : = \meanlrlr{\psi}{\HH_z}{\psi}_{L^2(\R^{dN})}, \qquad \lf(\psi, z\ri) \in L^{2}(\R^{dN}) \oplus \mathfrak{h}_{\omega};
\eeq
as
\beq
\label{eq: eqc}
\eqc : = \inf_{\lf( \psi, z \ri) \in \domqc} \EE[\psi, z],
\eeq
where
\beq
\label{eq: domqc}
\domqc : = \lf\{ \lf(\psi, z\ri) \in L^{2}(\R^{dN}) \oplus \mathfrak{h}_{\omega} \: \Big| \: \lf\|\psi \ri\|_2 = 1, \lf| \fqc[\psi,z] \ri| < + \infty  \ri\}.
\eeq
Here, $\mathfrak{h}_{\omega}$ is the Hilbert completion of $\bigcap_{k\in \mathbb{N}} \mathscr{D}(\omega^k)$ with respect to the
scalar product $\langle \,\cdot\, \vert \,\cdot\, \rangle_{\mathfrak{h}_{\omega}} :=\langle \,\cdot \, \vert \omega \vert\,\cdot \,
\rangle_{\mathfrak{h}}$, {\it i.e.},
\begin{equation}
  \label{eq:3}
  \mathfrak{h}_{\omega} := \overline{\bigcap_{k\in \mathbb{N}} \mathscr{D}(\omega^k)}\,\!^{\langle  \cdot  \vert \cdot  \rangle_{\mathfrak{h}_{\omega}}}.
\end{equation}
We denote by $ (\psiqc, \zqc) \in \domqc $ a corresponding
minimizing configuration (if any), {\it i.e.}, such that
\beq
\eqc = \fqc\lf[\psiqc, \zqc\ri].
\eeq

Concretely, the functional $ \fqc $ plays the role of the
\emph{quasi-classical energy} of the system under consideration. However, the
reader should be careful and be aware that $ \HH_z $ \emph{is not} the
Hamiltonian energy of the whole system: the complete environment + small
system's evolution is indeed not of Hamiltonian type. For each fixed $z\in
\mathfrak{h}_{\omega}$, the Hamilton-Jacobi equations of $\fqc[\psi,z]$, w.r.t.\ the
(complex) $\psi$ variable, yield the dynamics of the small system; the
environment on the other hand is stationary in the problems under
consideration in this paper (see \cite{correggi2019arxiv} for a detailed
analysis of quasi-classical dynamical systems).

The preliminary questions to address towards the derivation of the above
quasi-classical effective models are whether such models are stable and, if
this is the case, whether a minimizing configuration does exist: explicitly,
if
\beq
\eqc \overset{?}{>} -\infty \qquad \mbox{(stability)}, \tag{VP1}\label{eq: vp1}
\eeq
\beq
\overset{?}{\exists} \lf( \psiqc, \zqc \ri) \in \domqc \qquad \mbox{(existence of a ground state)}. \tag{VP2}\label{eq: vp2}
\eeq
Note that any critical point $ \lf( \psi, z \ri) \in \domqc $ of the functional $
\fqc[\psi,z] $ must satisfy
the condition $ \delta_{(\psi,z)} [ \fqc[\psi,z] - \epsilon \lf\| \psi
\ri\|_2^2 ] = 0 $, which yields the Euler-Lagrange
equations
\beq
\label{eq: el}
\begin{cases}
  \HH_z \psi = \epsilon \psi,	\\
  \omega z + \meanlrlr{\psi}{ \partial_{\bar{z}} \sum_{j}
    \mathcal{V}_z(\xv_j)}{\psi}_{L^2(\R^{dN})} = 0,
\end{cases}
\eeq
where the Lagrange multiplier $ \epsilon  = \meanlrlr{\psi}{H_z}{\psi} \in \R $ takes into account the normalization
constraint on $ \psi $. We anticipate
that a consequence of the convergence of the microscopic ground state, stated
in \cref{teo: minimizers 1}  below, is that, under suitable assumptions on $
\mathcal{K}_0 $ (for instance if $ \mathcal{W} $
is trapping), the answer to both questions in \eqref{eq: vp1} and \eqref{eq: vp2} is positive and, in
particular, the set of minimizers is not empty.

The variational problem above is strictly related to the more general issue
of rigorous derivation of effective theories, since, at least for the polaron
model, it is known that the minimization of the microscopic energy can be
approximated in the limit $ \eps \to 0 $ in terms of a nonlinear problem on $ \psi
$ alone. Indeed, focusing on the particle system, one can naturally approach
\eqref{eq: eqc} in a different and \emph{a priori} inequivalent way, {\it
  i.e.}, {\it first} one gets rid of the classical field by minimizing over $
z \in \mathfrak{h}_{\omega} $ and \emph{then} investigates the minimization of the
remaining functional on $ \psi $, which is obviously nonlinear, since the
minimizing $ z $ depends on $ \psi $ itself. As anticipated, this strategy has
been already followed in the literature in the case of the polaron in the
strong coupling regime, leading to the {\it Pekar functional} and the
corresponding variational problem \cite{pekar1955ac, donsker1983cpam,
  lieb1997cmp}. Such a feature is however not exclusive of the polaron and
can be observed in all the models mentioned above: we present below a formal
derivation of a Pekar-like functional $ \fpek[\psi] $, for both the Nelson and
polaron model. The Pauli-Fierz case is also discussed below, let us remark
however that in this case such a procedure does not yield an explicit
nonlinear functional of $ \psi $ (see \eqref{eq: fpek PF} below), because it is
in general not possible to solve explicitly the variational equation
expressing the minimizing $ z $ in terms of $ \psi $.

The formal procedure goes as follows: solving the critical point condition
$\delta_{z} \fqc = 0 $ w.r.t.\ the variable $z $ for fixed $ \psi $, we find some $
z_{\psi} $, that we can plug
in $ \fqc $, thus obtaining the Pekar energy $
\fpek[\psi] := \fqc [\psi, z_{\psi}] $. Such a scheme can be
made to work rigorously for the polaron (case (b)) with some care, but the
variable $ z $ is not the right one to consider in cases (a) and (c). Under
the assumptions we have made (recall in particular \eqref{eq: lambda N} and
\eqref{eq: lambda PF}), it is indeed more natural to set, since $z\in \mathfrak{h}_{\omega}$,
\beq
\eta : = \omega^{1/2} z,
\eeq
(note however that in case (b) $ \eta = z $) and consider the functional
$ \ffqc[\psi, \eta] : = \fqc[\psi, \omega^{-1/2} \eta] $, which in case (a) reads
\begin{multline}
  \label{eq: qc energy computed}
  \ffqc[\psi, \eta] = \meanlrlr{\psi}{\mathcal{K}_0 +  2\Re \tx\sum_{j} \braket{\eta}{\omega^{-1/2}\lambda(\xv_j)}_{\mathfrak{h}}}{\psi}_{L^2(\R^{dN})} + \lf\| \eta \ri\|^2_{\mathfrak{h}}	\\
  = \meanlrlr{\psi}{\mathcal{K}_0}{\psi}_{L^2(\R^{dN})} + 2 \Re \braketr{\eta}{\meanlrlr{\psi}{\Lambda}{\psi}_{L^2(\R^{dN})}}_{\mathfrak{h}} + \lf\| \eta \ri\|^2_{\mathfrak{h}}
\end{multline}
where $ \Lambda \in L^{\infty}(\R^{dN}; \mathfrak{H}) $ is given by $ \Lambda(\xxv) := \sum_{j =
  1}^N \lf(\omega^{-1/2} \lambda \ri)(\xv_j) $ (recall again the assumption \eqref{eq:
  lambda N} on $ \lambda $) and we have exploited the linearity of the scalar
product. Taking now the functional derivative w.r.t.\ to $ \eta $, we get the
Euler-Lagrange equation for the minimization of the above energy w.r.t. $ \eta \in
\mathfrak{h} $, {\it i.e.},
\beq
\label{eq: el N}
\eta + \meanlrlr{\psi}{\Lambda(\: \cdot \:)}{\psi}_{L^2(\R^{dN})} = 0,
\eeq
yielding the minimizing $ \zpek $ as
\beq
\label{eq: zpek N}
\zpek[\psi] = - \sum_{j = 1}^N \int_{\R^{dN}} \diff \xv_1 \cdots \xv_N \: \lf(\omega^{-1/2} \lambda
\ri)( \xv_j) \lf| \psi(\xv_1, \ldots, \xv_N) \ri|^2
\eeq
which can be easily seen to belong to $ \mathfrak{h} $ under the assumptions
made. Plugging $ \zpek $ back into \eqref{eq: qc energy computed}, we get
\beq
\label{eq: fpek N}
\fpek[\psi] := \inf_{\eta \in \mathfrak{h}} \ffqc[\psi, \eta] = \ffqc\lf[\psi, \zpek[\psi] \ri] =\meanlrlr{\psi}{\mathcal{K}_0 + \vpek \star \lf| \psi \ri|^2}{\psi}.
\eeq
Here we have denoted by $ \star $ the action of the integral kernel $ \vpek(\xxv,
\yyv) $ on $ |\psi|^2 $,
{\it i.e.},
\beq
\lf( \vpek \star \lf| \psi \ri|^2 \ri) (\xxv) := \int_{\R^{dN}} \diff \yyv \: \vpek(\xxv, \yyv) \lf|\psi(\yyv)\ri|^2,
\eeq
and
\beq
\vpek(\xxv,\yyv) = - \Re \sum_{i,j = 1}^N \meanlrlr{\lambda(\xv_i)}{\omega^{-1}}{\lambda(\yv_j)}_{\mathfrak{h}} \in L^{\infty}(\R^{2dN}).
\eeq

Note that in case of identical particles -- either fermionic or bosonic --, the
above expressions may be conveniently rewritten using the one-particle
density $ \rho_{\psi} \in L^1(\R^d) $ associated with $ \psi $, {\it i.e.},
\beq
\rho_{\psi}(\xv) : = N \int_{\R^{d(N-1)}} \diff \xv_2 \cdots \diff \xv_N \: \lf| \Psi(\xv, \xv_2, \ldots, \xv_N) \ri|^2.
\eeq
Indeed, in this case, \eqref{eq: zpek N} reads
\bdm
\zpek[\psi] = - \braketr{\rho_{\psi}}{ \lf(\omega^{-1/2}\lambda \ri) \lf( \: \cdot \: \ri)}_{L^2(\R^d)},
\edm
and the Pekar energy becomes
\beq
\fpek[\psi] = \meanlrlr{\psi}{\mathcal{K}_0}{\psi}_{L^2(\R^{dN})} + \meanlrlr{\rho_{\psi}}{\mathcal{U}}{\rho_{\psi}}_{L^2(\R^d)},
\eeq
where
\beq
\mathcal{U} = U(\xv,\yv) : = \meanlrlr{\lambda(\xv)}{\omega^{-1}}{\lambda(\yv)}_{\mathfrak{h}},
\eeq
which is its typical form in the literature. For instance, in the polaron
case, one recovers the self-interacting potential generated by the kernel
$ \mathcal{U}(\xv - \yv) = - \alpha \lf| \xv - \yv \ri|^{-1} $.

The above derivation can be easily seen to be correct under the assumptions
made in case (a). In case (b), however, one can not apply such a derivation
straightforwardly because $ \lambda \notin L^{\infty}(\R^{dN}; \mathfrak{h}) $, but a simple
well-known trick (see \cref{sec:polaron-model}) allows to split it into two
terms, which can be handled separately as above. In case (c) on the other
hand the Pekar functional takes the implicit form
\beq
\label{eq: fpek PF}
\begin{cases}
  \zpek  + \sum_j \frac{1}{m_j} \meanlrlr{\psi}{e \omega^{-1/2} \bm{\lambda}_j \cdot \lf( - i \nabla_j \ri) + 2 e^2 \omega^{-1/2} \bm{\lambda}_j \cdot \Re \braketr{\zpek}{\omega^{-1/2} \bm{\lambda}_j}_{\mathfrak{h}}}{\psi}_{L^2(\R^{3N})} =0,	\\
  \fpek[\psi] = \meanlrlr{\psi}{\HH_{z_{\psi}}}{\psi}_{L^2(\R^{3N})},
\end{cases}
\eeq
where $ \HH_z $ is given by \eqref{eq: classical hamiltonian PF} and we set $
z_{\psi} : = \omega^{-1/2} \zpek[\psi] $ for short. As before, all the terms in the first equation belong to $ \mathfrak{h} $, thanks to the assumptions on $ \bm{\lambda}_j $ and the fact that any $ \lf(\psi, z\ri) \in \domqc $ is such that $ \psi \in H^1(\R^{3N}) $. Furthermore, the last term can be thought of as the action on $ \zpek $ of a linear operator $ T  $ on $ \mathfrak{h} $ whose norm is bounded by
\bdm
	2 e^2 \sum_{j = 1}^N \tx\frac{1}{m_j} \lf\| \omega^{-1/2} \bm{\lambda}_j \ri\|^2_{\mathfrak{h}},
\edm
which is smaller than $ 1 $, if $ e $ is small enough. In this case, $ 1  + T $ is invertible and there exists a unique solution $ \zpek[\psi] \in \mathfrak{h} $ of the first equation. More in general, existence and uniqueness of $ \zpek[\psi] $ for any value of $ e $ follows from the strict convexity of the energy in $ \eta $ (see next \cref{rem: uniqueness pekar} and \cref{lemma: convex}). Note however that unfortunately  it is not possible to write explicitly $ \fpek $ as a functional of $ \psi $ alone, since, due to the presence of an operator -- the gradient --, one can not exchange the scalar product  in $ L^2(\R^{3N}) $ with the one in $ \mathfrak{h} $, as it was done in \eqref{eq: qc energy computed}. In particular, even for identical particles, the second term in the first equation in \eqref{eq: fpek PF} depends on the reduced density matrix, while the last one is a function of the density alone.

We now define
\beq
\epek : = \inf_{\psi \in \domp} \fpek[\psi],
\eeq
with $ \domp : = \lf\{ \psi \in L^2(\R^{dN}) \: | \: \lf\| \psi \ri\|_2 = 1, \lf|
\fpek[\psi] \ri| < + \infty \ri\} $, as the ground state energy of
the Pekar functionals \eqref{eq: fpek N} and \eqref{eq: fpek PF}, and denote by $
\psipek \in \domp $ any corresponding
minimizer. It is then natural to wonder whether there is any connection
between the questions \eqref{eq: vp1} and \eqref{eq: vp2} and the analogous problems for $ \fpek $,
{\it i.e.},
\beq
\epek \overset{?}{>} -\infty \qquad \mbox{(stability)},
\tag{VP${}^{\prime}$1}\label{eq:
  vpp1}
\eeq
\beq
\overset{?}{\exists} \psipek \in L^2(\R^{dN}) \qquad \mbox{(existence of a ground state)}.
\tag{VP${}^{\prime}$2}\label{eq: vpp2}
\eeq
This is of particular interest for physical applications, since the
minimization of the nonlinear functional $ \fpek $ may be easier to address
also in numerical experiments. A priori however it is not at all obvious that
such a relation exists, but in the next \cref{pro: equivalence} we are going
to state that the two variational problems are actually equivalent, which is
particularly interesting in case (c) since the explicit form of $ \fpek $ is
not available.

\begin{proposition}[Equivalence of variational problems]
  \label{pro: equivalence}
  \mbox{}	\\
  Under the assumptions made above,
  \beq \epek = \eqc > - \infty.
  \eeq
  Furthermore, if $(\psiqc,\zqc) \in \domqc $ is a minimizer of $ \fqc[\psi,z] $, then
  \beq
  \label{eq: equivalence 1}
  \fpek\lf[\psiqc \ri] = \epek.
  \eeq
  Conversely, if $ \psipek $ is a minimizer of $ \fpek[\psi] $, then $
  \zpek[\psipek] \in \mathfrak{h} $ (given by
  \eqref{eq: el N} and \eqref{eq:
    fpek PF} with $ \psi = \psipek $, respectively) and
  \beq
  \label{eq: equivalence 2}
  \EE[\psipek, \zpek] = \eqc.
  \eeq
\end{proposition}

\begin{remark}[Uniqueness of $ \zpek $]
  \label{rem: uniqueness pekar}
  \mbox{}	\\
  We prove in \cref{lemma: convex} that the quasi-classical functional $
  \ffqc[\psi,\eta] $ (or,
  equivalently, $ \fqc[\psi,z] $) is strictly convex in $ \eta \in
  \mathfrak{h} $ for given $ \psi \in L^2(\R^{dN}) $.
  Hence, $ \zpek[\psi] $ is unique (for fixed $ \psi $). Note however that the functional $ \ffqc $
  is not jointly convex in $ ( \lf|\psi \ri|^2, \eta ) $.
\end{remark}

\subsection{Ground state in the quasi-classical regime}
\label{sec: ground state}

We can now state in detail our main results. We first consider the
microscopic ground state energy $ E_{\eps} $ defined in \eqref{eq: eeps} and
its quasi-classical limit. Recall the definition of the quasi-classical
energy $ \eqc $ in \eqref{eq: eqc}.

\begin{thm}[Ground state energy]
  \label{teo: ground state energy}
  \mbox{}	\\
  Under the assumptions made above, $ \exists C < + \infty $ such that $ E_{\eps} > - C
  $ and
  \beq
  \label{eq: ground state energy}
  E_{\eps} \xrightarrow[\eps \to 0]{} \eqc,
  \eeq
  which in particular implies that \eqref{eq: vp1} holds true.
\end{thm}

\begin{remark}[Assumptions]
  \mbox{}	\\
  The above result requires only a minimal set of assumptions on the
  microscopic models, those listed in their definitions, which are the
  weakest ones guaranteeing the self-adjointness and boundedness from below
  of the microscopic Hamiltonians. In particular, the quantum potential $
  \mathcal{W} $
  may not be trapping, so that there might be no ground state
  for both the microscopic and the macroscopic problems.
\end{remark}

The above \cref{teo: ground state energy} completes and extends analogous
results proven in \cite[Thm. 2.4]{correggi2017ahp} and
\cite[Thm. 1.9]{correggi2017arxiv}, relaxing the assumptions on the
microscopic models and taking into account more general settings. We also
point out that the proof of the above result provided in
\cref{sec:ground-states-quasi} is quite different and much more general than
the ones contained in the above references and involves the new mathematical
structure of {\it quasi-classical Wigner measures} first introduced in
\cite{correggi2019arxiv}. In fact, the argument in the proofs given in
\cite{correggi2017ahp,correggi2017arxiv} is not complete, since it relies on
the assumption that one can find a minimizing sequence which can be
decomposed into a linear combination of finitely-many product states, whose
number is {\it uniformly bounded in} $ \eps $. This is \emph{a posteriori}
right (as it follows from the proof of \cref{teo: ground state energy}), up
to errors vanishing in the limit $ \eps \to 0 $, but it is unproven there.

Once the energy convergence has been stated, it is natural to ask whether, in
presence of a microscopic approximate ground state $\Psi_{\varepsilon,\delta}$ or ground state $ \psigs
$, one can prove a suitable convergence respectively to quasi-classical
minimizing sequences or configurations $ (\psiqc, \zqc) \in \domqc $. Let us
stress that the question of existence of a ground state of the microscopic
energy has been widely studied in the literature and there are more
restrictive conditions on the models guaranteeing that $ E_{\eps} \in
\sigma_{\mathrm{pp}}(H_{\eps}) $ (see
\cref{sec:nelson-model,sec:polaron-model,sec:pauli-fierz-model}); our results
about approximate ground states apply even if the microscopic ground state do not
exist, and whenever it exists we are able to provide its quasi-classical
characterization.

In order to properly formulate the convergence, we first need to introduce a
key structure in quasi-classical analysis: the \emph{quasi-classical Wigner
  measures} and their relative topologies. We preliminarily recall the
definition of the space $ \mathscr{P}(\mathfrak{h}_{\omega}; L^2(\R^{dN})) $ of
{\it state-valued probability measures} (see \cite[Def. 2.1]{correggi2019arxiv}), given
by measures $ \mm $ on $ \mathfrak{h}_{\omega} $ taking values in $
\mathscr{L}^1_+(L^2(\R^{dN})) $ -- the space of positive trace
class operators on $ L^2(\R^{dN}) $ -- such that $ \mm(\emptyset) = 0 $, the measure is
unconditionally $ \sigma-$additive in the trace class norm and $\lVert \mathfrak{m}(\mathfrak{h}_{\omega})  \rVert_{L^2}^{}=1$. Starting from such a
notion, it is possible to construct a theory of integration of functions with
values in the space of bounded operators on $ L^2(\R^{dN}) $ w.r.t.\
state-valued measures, so that, for any measurable $ \mathcal{B}(z) \in
\mathscr{B}(L^2(\R^{dN})) $,
\beq
\label{eq: integration 1}
\int_{\mathfrak{h}_{\omega}} \diff \mm(z) \: \mathcal{B}(z) \in \mathscr{L}^1(L^2(\R^{dN})).  
\eeq
We refer to \cref{sec:gener-state-valu}, or to the existing literature
(\emph{e.g.}, \cite{ba,pg,gms,fg,teu}) for further details. In particular, we
point out that any such state-valued measure $ \mm $ admits a
Radon-Nikod\'{y}m decomposition, {\it i.e.}, there exist a scalar Borel
measure $\mu_{\mm} $ and a $\mu_{\mm}$-integrable function $ \gamma_{\mm}(z) \in
\mathscr{L}^1_{+,1}(L^2(\R^{dN}))$ defined a.e.\ and
with values in normalized density matrices, such that 
\beq 
\diff \mm(z) = \gamma_{\mm}(z) \diff \mu_{\mm}(z).  
\eeq 
Hence, \eqref{eq: integration 1} can be rewritten 
\beq
\label{eq: integration 2}
\int_{\mathfrak{h}_{\omega}} \diff \mm(z) \: \mathcal{B}(z) = \int_{\mathfrak{h}_{\omega}} \diff \mu_{\mm}(z) \: \gamma_{\mm}(z) \mathcal{B}(z).  
\eeq
Finally, let us denote by $ W_{\varepsilon}(z) $, $ z \in \mathfrak{h} $ the Weyl operator constructed over the
creation and annihilation operators $
a^{\sharp}_{\eps} $, {\it i.e.},
\begin{equation}
  \label{eq: weyl}
  W_{\varepsilon}(z) :=e^{i(a^{\dagger}_{\varepsilon}(z)+a_{\varepsilon}(z))}\; .
\end{equation}

\begin{definition}[Quasi-classical Wigner measures]
  \label{def: wigner}
  \mbox{}	\\
  For any family of normalized microscopic states $ \lf\{ \Psi_{\varepsilon} \ri\}_{\varepsilon\in
    (0,1)} \subset \mathscr{H}_{\eps} $, the
  associated set of quasi-classical Wigner measures
  $ \WW(\Psi_{\eps}, \varepsilon\in (0,1)) \subset \mathscr{P}(\mathfrak{h}_{\omega};
  \mathscr{L}^1_+(L^2(\R^{dN}))$ is the subset of all probability measures $ \mm $, such that
  \beq
  \Psi_{\varepsilon_n} \xrightarrow[\eps_{n} \to 0]{\mathrm{qc}} \mathfrak{m},
  \eeq
  where the above convergence means that, for all $\eta\in
  \mathscr{D}(\omega^{-1/2})$ and all compact
  operators $\mathcal{K}\in
  \mathscr{L}^{\infty}(L^2(\R^{dN}))$,
  \begin{multline}
    \label{eq: convergence}
    \lim_{n\to + \infty} \braketr{\Psi_{\varepsilon_n}}{\mathcal{K} \otimes W_{\varepsilon_n}(\eta) \Psi_{\varepsilon_n}}_{\mathscr{H}_{\eps_n}} = \int_{\mathfrak{h}_{\omega}}^{} \mathrm{d}\mu_{\mathfrak{m}}(z) \: e^{2i\Re \braket{\eta}{ z}_{\mathfrak{h}}} \tr_{L^2(\R^{dN})} \: \lf[ \gamma_{\mathfrak{m}}(z) \mathcal{K} \ri]\\=\int_{\mathfrak{h}_{\omega}}^{} \mathrm{d}\mu_{\mathfrak{m}}(z) \: e^{2i\Re \braketr{\omega^{-1/2} \eta}{\omega^{1/2} z}_{\mathfrak{h}}} \tr_{L^2(\R^{dN})} \: \lf[ \gamma_{\mathfrak{m}}(z) \mathcal{K} \ri] \; .
  \end{multline}
\end{definition}
\begin{remark}[Measures on $\mathfrak{h}_{\omega}$ and test functions]
  \label{rem:8}
  \mbox{}\\
  A reader familiar with infinite dimensional semiclassical analysis or
  quasi-classical analysis will find the definition of Wigner measures given
  here slightly different to the usual one
  \cite{ammari2008ahp,correggi2019arxiv}. Typically, one considers microscopic
  states that satisfy a number operator estimate, namely for which the
  expectation of $\mathrm{d}\mathcal{G}_{\varepsilon}(1)^{c}$ is $\varepsilon$-uniformly bounded
  for some $c >0$. The corresponding Wigner measures are concentrated on
  $\mathfrak{h}$ \cite{ammari2008ahp}, and it is natural to test the
  convergence with Weyl operators having arguments $\eta\in
  \mathfrak{h}$. However, in studying
  variational problems the number operator estimate may not always be
  available, in particular whenever the field is massless, such as in
  electromagnetism (Pauli-Fierz model). In that case, only  energy
  estimates, \emph{i.e.},\, involving $\mathrm{d}\mathcal{G}_{\varepsilon}(\omega)$, are
  available. The Wigner measures of states satisfying such an energy estimate
  are concentrated in $\mathfrak{h}_{\omega}$, and it is natural to test
  convergence with Weyl operators having arguments $\eta\in \mathscr{D}(\omega^{-1/2}) $ belonging to a dense subset of the
  continuous dual space \cite{falconi2017ccm}. If
  both the number estimate and the free energy estimate are available, then
  the measure is concentrated in $\mathfrak{h}\cap \mathfrak{h}_{\omega}$; this
  happens for massive fields, where in addition $\mathfrak{h}\cap
  \mathfrak{h}_{\omega}=\mathfrak{h}_{\omega}$. Finally, let us
  remark as well that in all concrete applications $\mathfrak{h}_{\omega}$ is in fact
  the natural domain of definition of the quasi-classical energy $\mathcal{E}_{\mathrm{qc}}$.
\end{remark}

The above notion of quasi-classical convergence, defined in \eqref{eq:
  convergence}, is however not the only meaningful topology one can consider
for sequences of microscopic states. More precisely, the test in \eqref{eq:
  convergence} may be extended to bounded operators, which means that one is
considering the weak-* topology on $\mathscr{B}(L^2(\R^{dN}))'$, instead of
$\mathscr{L}^1(L^2(\R^{dN})) =\mathscr{L}^{\infty}(L^2(\R^{dN}))'$. In this case,
the cluster points belong to a larger space than $
\mathscr{P}(\mathfrak{h}_{\omega}; \mathscr{L}^1_+(L^2(\R^{dN})) $, namely the
space of {\it generalized state-valued measures} (see
\cite{falconi2017arxiv} for a detailed and more general discussion). We thus
introduce the set of positive states $
\overline{\mathscr{L}^1}_+(L^2(\R^{dN})) $ in the closure w.r.t.\ the
weak-* topology of the space of trace class operators on $ L^2(\R^{dN}) $: we
denote the action of a functional $ F \in
\overline{\mathscr{L}^1}_+(L^2(\R^{dN})) $ on a bounded operator
$
\mathcal{B} \in \mathscr{B}(L^2(\R^{dN})) $ as $ F[\mathcal{B}] \in \CC $ and its
norm as
\beq
\label{eq: norm B}
\lf\| F \ri\|_{\mathscr{B}'} : = \sup_{\mathcal{B} \in
  \mathscr{B}(L^2(\R^{dN})), \lf\| \mathcal{B} \ri\| = 1} \lf| F\lf[
\mathcal{B} \ri] \ri|.
\eeq

\begin{definition}[Generalized quasi-classical Wigner measures]
  \label{def: g wigner}
  \mbox{}	\\
  For any family of normalized microscopic states $ \lf\{ \Psi_{\varepsilon} \ri\}_{\varepsilon\in
    (0,1)} \subset L^2(\R^{dN})_{\varepsilon} $, the
  associated set of generalized quasi-classical Wigner measures
  $ \GW(\Psi_{\eps}, \varepsilon\in (0,1)) \subset
  \mathscr{P}(\mathfrak{h}_{\omega}; \overline{\mathscr{L}^1}_+(L^2(\R^{dN}))$ is
  the subset of all probability measures $ \nn $, such that
  \beq
  \Psi_{\varepsilon_{n}}\xrightarrow[\eps_{n} \to 0]{\mathrm{gqc}} \mathfrak{n},
  \eeq
  where the above convergence means that, for all $\eta\in \mathscr{D}(\omega^{-1/2})$ and all
  bounded operators $\mathcal{B} \in \mathscr{B}(L^2(\R^{dN}))$,
  \begin{equation}
    \label{eq: g convergence}
    \lim_{n\to + \infty} \braketr{\Psi_{\varepsilon_n}}{\mathcal{B} \otimes W_{\varepsilon_n}(\eta) \Psi_{\varepsilon_n}}_{\mathscr{H}_{\eps_n}} = \int_{\mathfrak{h}_{\omega}}^{} \mathrm{d}\mathfrak{n}(z)[\mathcal{B}] \: e^{2i\Re \braketr{\omega^{-1/2} \eta}{\omega^{1/2} z}_{\mathfrak{h}}}  \;.
  \end{equation}
\end{definition}

We can now formulate the results about the convergence of microscopic
minimizing sequences $\Psi_{\varepsilon,\delta}$ and microscopic minimizers $ \psigs $. We
start by stating a stronger result with some additional assumptions on the
microscopic models. Without such assumptions we are still able to prove a
weaker convergence, but it requires to introduce a generalized variational
problem.

\begin{thm}[Convergence of approximate ground states (I)]
  \label{thm:2}
  \mbox{} \\
  If $\mathcal{K}_0$ has compact resolvent, then, for any $ \delta > 0 $ and for any family of approximate ground states $\Psi_{\varepsilon,\delta}$ satisfying \eqref{eq:1}, $ \WW(\Psi_{\varepsilon,\delta}, \varepsilon\in (0,1)) \neq \emptyset $. Moreover, any family of
  quasi-classical Wigner measures $\{\mathfrak{m}_{\delta}\}_{\delta>0} \in
  \bigcup_{\delta>0}\WW(\Psi_{\varepsilon,\delta}, \varepsilon\in (0,1)) $ is such that, for all $\delta>0$,
  $
  \tr_{L^2(\R^{dN})} \mm_{\delta}(\mathfrak{h}_{\omega}) = 1 $  and it is an approximate ground state of $\mathcal{E}_{\mathrm{qc}}[\psi,z]$, {\it i.e.},
  \begin{equation}
    \int_{\mathfrak{h}_{\omega}}^{}  \mathrm{d}\mu_{\mathfrak{m}_{\delta}}(z)\tr_{L^2(\mathbb{R}^{dN})}\bigl(\gamma_{\mathfrak{m}_{\delta}}(z)\mathcal{H}_z\bigr)<E_{\mathrm{qc}}+\delta\; .
  \end{equation}
  Consequently, there is small $\mathfrak{m}_{\delta}$-probability that
  $\mathcal{E}_{\mathrm{qc}}(\psi_{\delta},z_{\delta})$ is larger than $E_{\mathrm{qc}}+\delta$: for all $k\in \mathbb{N}_{*}$,
  \begin{equation}
    \mathbb{P}_{\mathfrak{m}_{\delta}}\Bigl\{\mathcal{E}_{\mathrm{qc}}(\psi_{\delta},z_{\delta})\geq E_{\mathrm{qc}}+k\delta\Bigr\}<\frac{1}{k}\; .
  \end{equation}
\end{thm}

\begin{corollary}[Convergence to ground states (I)]
  	\label{cor:2}
  	\mbox{}	\\
  If $\mathcal{K}_0$ has compact resolvent, then any quasi-classical Wigner
  measure $\mathfrak{m} \in \WW(\Psi_{\varepsilon,o_{\varepsilon}(1)}, \varepsilon\in (0,1)) $, corresponding to approximate ground states $\Psi_{\varepsilon,o_{\varepsilon}(1)}$ satisfying \eqref{eq:1} with
  $\delta=o_{\varepsilon}(1)$, is such that $ \tr_{L^2(\R^{dN})} \mm(\mathfrak{h}_{\omega}) = 1
  $ and it is concentrated on the
    set of ground states $(\psi_{\mathrm{qc}},z_{\mathrm{qc}})\in
  \mathscr{D}_{\mathrm{qc}}$ of $\mathcal{E}_{\mathrm{qc}}[\psi,z]$. Consequently,
  $\mathcal{E}_{\mathrm{qc}}[\psi,z]$ has at least one ground state and both \eqref{eq: vp2} and \eqref{eq: vpp2} hold true.
\end{corollary}

\begin{corollary}[Convergence of ground states (I)]
  \label{teo: minimizers 1}
  \mbox{}	\\
  If $\mathcal{K}_0$ has compact resolvent and $H_{\varepsilon}$ has a ground state $
  \psigs $, then any
  corresponding quasi-classical Wigner measure $\mathfrak{m} \in \WW(\psigs, \varepsilon\in
  (0,1)) $ is such that
  $ \tr_{L^2(\R^{dN})}
  \mm(\mathfrak{h}_{\omega}) = 1 $ and it is concentrated on the set of ground states 
  $(\psiqc,\zqc) \in \domqc $ of $ \fqc[\psi,z] $.
\end{corollary}
	
\begin{remark}[Uniqueness and gauge invariance]
  \mbox{}	\\
  Concerning uniqueness, we point out that both the microscopic and the
  quasi-classical variational problems are gauge invariant, namely the
  multiplication by a constant phase factor of $ \Psi $ or $ \psi $ does not change
  the energy. Hence, even if one could prove uniqueness of the
  quasi-classical minimizer $ (\psiqc, \zqc) $ up to gauge transformations,
  one could not conclude that the set of limit points $ \WW(\Psi_{\varepsilon,o_{\varepsilon}(1)},\varepsilon\in
  (0,1)) $ or
  $ \WW(\psigs,\varepsilon\in (0,1)) $ are just given by a Dirac delta
  measure centered at $ (\psiqc, \zqc) $. Indeed, because of gauge
  invariance, the quasi-classical Wigner measures would be supported over the
  unit one-dimensional sphere generated by the configurations $ \lf(e^{i \vartheta}
  \psiqc, \zqc \ri) $, $ \vartheta \in \R $.
\end{remark}
	
\begin{remark}[Condition on $ \mathcal{K}_0 $]
  \mbox{}	\\
  The assumption that $ \mathcal{K}_0 $ has compact resolvent is reasonable,
  since that is typically the case in which one can also prove the existence
  of a microscopic minimizer at least for massive systems (see \cref{rem:
    minimizer} below),
  {\it e.g.}, in presence of a trapping potential. However, it is also needed
  in a technical step in the proof to ensure that there is no loss of mass
  along the convergence \eqref{eq:
    convergence}, {\it i.e.}, $ \tr_{L^2(\R^{dN})} \mm(\mathfrak{h}_{\omega}) = 1
  $. Similar assumptions
  are present also in \cite{correggi2019arxiv} (see in
  particular the discussion in \cite[Rmks. 1.9 -- 1.10 \& \textsection 1.6]{correggi2019arxiv}).
\end{remark}
\begin{remark}[Existence of $ \psigs $]
  \label{rem: minimizer}
  \mbox{}	\\
  In all the three cases (a) -- (c), if the Bose field is {\it massive}, {\it
    i.e.}, $ \exists m > 0 $ such that $ \omega \geq m > 0 $ (which is always the case for
  the polaron), then it is known \cite[Thm. 4.1]{derezinski1999rmp} that the
  microscopic Hamiltonian $H_{\eps} $ admits a ground state $ \psigs \in
  \mathscr{H}_{\eps} $, if $ \mathcal{K}_0 $ has compact
  resolvent. Hence, in the massive case, one can remove the assumption on the
  existence of $
  \psigs $. When the field is {\it massless}, on the other hand, it 
  is also known that microscopic ground states might not exist or belong to a
  non-Fock representation of the algebra of observables
  \cite{pizzo2003ahp}. This second case is not covered by the above
  \cref{teo: minimizers 1}, but it may be treated with our techniques. We
  plan to come back to such a question in a future work.
\end{remark}
\begin{remark}[Existence of quasi-classical minimizers]
  \label{rem:2}
  Our analysis shows that the quasi-classical energy functionals
  $\mathcal{E}_{\mathrm{qc}}[\psi,z]$ \emph{always} have at least one minimizer,
  provided that $\mathcal{K}_0$ has compact resolvent, \emph{i.e.}\ provided
  that the quantum subsystem is trapped. This gives an additional evidence of
  the fact that nonexistence or non-Fock-representability (see \cref{rem:
    minimizer} above) of
  the microscopic ground-state is one of the many complications encountered in
  quantizing fields.
\end{remark}
	
As anticipated, if we drop the assumption on the operator $ \mathcal{K}_0 $,
there is still convergence, but the variational problem \eqref{eq: eqc} has
to be generalized: we thus set, for any pure state $ \rho \in
\overline{\mathscr{L}^1}_+ \lf(L^2(\R^{dN})\ri) $ and any
$ z \in
\mathfrak{h}_{\omega} $,
\beq
\label{eq: fqct}
\fgqc[\rho,z] : = \rho \lf[\mathcal{H}_z \ri].
\eeq
We consider the corresponding variational problem: setting (recall the
definition \eqref{eq: norm B})
\beq
\label{eq: domgqc}
\domgqc : = \lf\{ \lf(\rho,z\ri) \in \overline{\mathscr{L}^1}_+(L^2(\R^{dN})) \oplus \mathfrak{h}_{\omega} \: \Big| \: \lf\| \rho \ri\|_{\mathscr{B}'} = 1\, ,\, \lf| \rho\lf[\mathcal{H}_z\ri] \ri| < +\infty \ri\},
\eeq
we define
\beq
\label{eq: egqc}
\egqc : = \inf_{\lf(\rho,z\ri) \in \domgqc} \fgqc[\rho,z],
\eeq
and denote by $(\rho_{\delta},z_{\delta})\in \mathscr{D}_{\mathrm{gqc}}$ a minimizing sequence satisfying
\begin{equation*}
  \mathcal{E}_{\mathrm{gqc}}[\rho_{\delta},z_{\delta}]< E_{\mathrm{gqc}}+\delta\; ,
\end{equation*}
and by $ (\rhogqc, \zgqc) \in \domgqc $ any corresponding minimizing
configuration.

\begin{thm}[Convergence of approximate ground states (II)]
  \label{thm:5}
  \mbox{}\\
  If $\mathcal{K}_0$ does not have compact resolvent, then, for any $ \delta > 0 $ and for any family of approximate ground states $\Psi_{\varepsilon,\delta}$ satisfying \eqref{eq:1}, $ \GW(\Psi_{\varepsilon,\delta}, \varepsilon\in (0,1)) \neq \emptyset $. Moreover, any family of
  generalized quasi-classical Wigner measures $\{\mathfrak{n}_{\delta}\}_{\delta>0} \in
  \bigcup_{\delta>0}\GW(\Psi_{\varepsilon,\delta}, \varepsilon\in (0,1)) $ is such that,
  for all $\delta>0$, $\lf\|
  \nn_{\delta}(\mathfrak{h}_{\omega}) \ri\|_{\mathscr{B}'} = 1 $ and it is an approximate ground state of $\mathcal{E}_{\mathrm{gqc}}[\rho,z]$,
  {\it i.e.},
  \begin{equation}
    \int_{\mathfrak{h}_{\omega}}^{}  \mathrm{d}\mathfrak{n}_{\delta}(z)[\mathcal{H}_z]<E_{\mathrm{gqc}}+\delta\; .
  \end{equation}
\end{thm}

\begin{corollary}[Convergence to ground states (II)]
  	\label{cor:3}
  	\mbox{}	\\
 	If $\mathcal{K}_0$ does not have compact resolvent, then any generalized
  quasi-classical Wigner measure $\mathfrak{n} \in \GW(\Psi_{\varepsilon,o_{\varepsilon}(1)}, \varepsilon\in
  (0,1)) $, corresponding to approximate ground states $\Psi_{\varepsilon,o_{\varepsilon}(1)}$ satisfying \eqref{eq:1} with $\delta=o_{\varepsilon}(1)$, is such that
  $\lf\| \nn(\mathfrak{h}_{\omega})
  \ri\|_{\mathscr{B}'} = 1 $ and it is concentrated on the set of ground states
  $(\varrho_{\mathrm{gqc}},z_{\mathrm{gqc}})\in \mathscr{D}_{\mathrm{gqc}}$ of $\mathcal{E}_{\mathrm{gqc}}[\varrho,z]$. Consequently, the functional $
  \fgqc[\rho,z] $ admits at
  least one ground state in $ \domgqc $.
\end{corollary}

\begin{corollary}[Convergence of ground states (II)]
  \label{teo: minimizers 2}
  \mbox{}	\\
  If $\mathcal{K}_0$ does not have compact resolvent and $H_{\varepsilon}$ has a ground
  state $ \psigs $, then any generalized Wigner measure $\mathfrak{n} \in
  \GW(\psigs,\varepsilon\in (0,1)) $ is such
  that $ \lf\| \nn(\mathfrak{h}_{\omega})
  \ri\|_{\mathscr{B}'} = 1 $ and it is concentrated on the set of ground states
  $ (\rhogqc, \zgqc) \in \domgqc $ of $ \fgqc[\rho,z] $.
\end{corollary}

\begin{remark}[Quasi-classical energy and generalized quasi-classical energy]
  	\label{rem:7}
  	\mbox{}	\\
  	As proved in \cref{sec:quasi-class-minim} below (see \cref{prop:5}),
  	\begin{equation*}
    		E_{\mathrm{qc}}=E_{\mathrm{gqc}}\; ,
	\end{equation*}
	which is in fact crucial to prove convergence of the ground state energy
  for systems without trapping on the quantum particles.
\end{remark}

\subsection*{Acknowledgements}
\label{sec:acknowledgements}

\begin{footnotesize}
The authors would like to thank Z.\ Ammari, N.\ Rougerie, and F.\ Hiroshima
for many stimulating discussions and helpful insights during the redaction of
this paper. M.F.\ has been supported by the European Research Council (ERC)
under the European Union’s Horizon 2020 research and innovation programme
(ERC CoG UniCoSM, grant agreement n.724939). M.O.\ has been partially
supported by GNFM group of INdAM through the grant Progetto Giovani 2019
``Derivation of effective theories for large quantum systems''.
\end{footnotesize}

\section{Quasi-Classical Minimization Problems}
\label{sec:quasi-class-minim}

In this section we consider minimization problems in the quasi-classical
setting: we study the functionals introduced in \cref{sec: quasi-classical
  var} and the relative
minimizations, but also define and investigate more general problems.

\subsection{Quasi-classical functionals, states and related minimization
  problems}
\label{sec:quasi-class-funct}

A quasi-classical system behaves like an open system in which a {\it
  classical environment} (of infinite dimension) drives a quantum {\it small
  system}, described by an Hilbert space $ L^2(\R^{dN}) $. The classical
environment is described by a space of configurations $\mathfrak{h}_{\omega}$,
usually a complex Hilbert space identifiable with the complex phase space of
the environment's degrees of freedom. A probability distribution $\mu$ on
$\mathfrak{h}_{\omega}$ tells how probable each environment's configuration
is, while a state-valued function $\mathfrak{h}_{\omega}\ni z\mapsto
\gamma(z)\in \mathscr{L}^1_+(L^2(\R^{dN}))$ tells how each
environment's configuration drives the small system's quantum state.
Analogously, both the value of observables $\mathcal{F}(z)$ and the small
system's dynamics $\mathcal{U}_t(z)$ are driven by the environment.

A quasi-classical minimization problem is the problem of finding the lowest
energy and possibly the ground states of a suitable functional
$
\mathcal{E}[\psi,z]:L^2(\R^{dN}) \oplus \mathfrak{h}_{\omega}\to \mathbb{R} $ depending on the configuration of both the small system
and the environment. The first energy functional to consider is $ \fqc[\psi, z]
$, as
defined in \eqref{eq: fqc}:
\bdm
\fqc[\psi, z] : = \meanlrlr{\psi}{\HH_z}{\psi}_{L^2(\R^{dN})}, \qquad \lf(\psi, z\ri) \in \domqc,
\edm
where $ \HH_z $ and $ \domqc $ are given in \eqref{eq: classical hamiltonian}
and \eqref{eq: domqc}, respectively. We also recall that
the ground state energy and minimizer of $ \fqc $ are denoted by $ \eqc $ and $
(\psiqc, \zqc) $,
respectively.

Although the above is the foremost functional coming to mind in this context,
another minimization problem emerges naturally in studying the
quasi-classical limit. To this purpose, we recall the notion of state-valued
measure \cite{falconi2017arxiv,correggi2019arxiv}, already mentioned in
\cref{sec: ground state}: a state-valued {\it probability} measure $\mathfrak{m} \in
\mathscr{P}\lf(\mathfrak{h}_{\omega};\mathscr{L}^1_+(L^2(\R^{dN})) \ri) $ is a
vector Borel Radon measure on $\mathfrak{h}_{\omega}$, taking values in the
density matrices $ \mathscr{L}^1_+(L^2(\R^{dN}))$ of the small system, such
that
\beq
\lVert \mathfrak{m}(\mathfrak{h}_{\omega}) \rVert_{\mathscr{L}^1}^{}=1.
\eeq
Thanks to the Radon-Nikod\'{y}m property enjoyed by the separable dual space
$\mathscr{L}^1(L^2(\R^{dN}))$, it is possible to decompose $\mathfrak{m}$ in a scalar Borel Radon probability
measure $\mu_{\mathfrak{m}}\in
\mathscr{P}(\mathfrak{h}_{\omega})$, such that $ \mu_{\mathfrak{m}}(\mathfrak{h}) = 1 $, and in an a.e.-defined function (the Radon-Nikod\'{y}m
derivative)
\begin{equation*}
  \mathfrak{h}_{\omega}\ni z \mapsto \gamma_{\mathfrak{m}}(z)\in \mathscr{L}^1_{+,1}\lf(L^2(\R^{dN})\ri)
\end{equation*}
taking values in the normalized density matrices of the small system:
\begin{equation*}
  \mathrm{d}\mathfrak{m}(z)=\gamma_{\mathfrak{m}}(z)\mathrm{d}\mu_{\mathfrak{m}}(z)\; .
\end{equation*}
The quasi-classical energy $ \fqc $, constrained to $\lVert \psi
\rVert_{L^2(\R^{dN})}^{}=1$, is the expectation
of the quasi-classical Hamiltonian $ \HH_z $. Therefore, its generalization to state
valued measures obviously reads
\begin{equation}
  \label{eq: fsv}
  \fsv[\mm] : = \int_{\mathfrak{h}_{\omega}} \mathrm{d}\mu_{\mathfrak{m}}(z) \; \tr_{L^2(\R^{dN})} \lf[ \gamma_{\mathfrak{m}}(z) \HH_z \ri] .
\end{equation}
This leads to the following minimization problem: setting
\beq
\label{eq: domsv}
\domsv : = \lf\{ \mm \in \mathscr{P} \lf(\mathfrak{h}_{\omega};\mathscr{L}^1_+\lf(L^2(\R^{dN})\ri)\ri) \: \Big| \:\tr_{L^2(\R^{dN})} \mm(\mathfrak{h}_{\omega}) = 1, \lf| \fsv[\mm] \ri| < + \infty\ri\},
\eeq
we define
\beq
\esv : = \inf_{\mm \in \domsv} \fsv[\mm] \overset{?}{>} -\infty,
\tag{vp1}\label{eq: vpsv1}
\eeq
\beq
\overset{?}{\exists} \msv \in \domsv \mbox{ s.t. } \fsv[\msv] = \esv.
\tag{vp2}\label{eq: vpsv2}
\eeq

A variant of the above problem is obtained by assuming that $
\gamma_{\mathfrak{m}}(z) =\lvert \psi \rangle\langle \psi \rvert$ for
some $ \psi \in L^2(\R^{dN}) $ independent of $z$, in which case the functional depends only on
a wave function $ \psi $ and a probability measure $ \mu $ over $ \mathfrak{h}_{\omega} $. We thus set
\begin{equation}
  \label{eq: fpm}
  \fpm[\psi, \mu] : = \int_{\mathfrak{h}_{\omega}} \mathrm{d}\mu(z) \; \meanlrlr{\psi}{\HH_z}{\psi}_{L^2(\R^{dN})}.
\end{equation}
The variational problem reads
\beq
\epm : = \inf_{\lf(\psi, \mu \ri) \in \dompm} \fpm[\psi, \mu] \overset{?}{>} -\infty, \tag{vp${}^{\prime}$1}\label{eq: vppm1} \eeq
where\beq
\dompm : = \lf\{ \lf(\psi, \mu \ri) \in L^2(\R^{dN}) \oplus \mathscr{P}\lf(\mathfrak{h}_{\omega}\ri), \lf\| \psi \ri\|_2 = 1, \mu(\mathfrak{h}_{\omega})= 1, \lf| \fpm\lf[ \psi, \mu \ri] \ri| < + \infty \ri\},
\eeq
and
\beq
\overset{?}{\exists} \lf(\psipm, \mupm\ri) \in \dompm \mbox{ s.t. } \fpm \lf[\lf( \psipm, \mupm \ri)\ri] = \epm.
\tag{vp${}^{\prime}$2}\label{eq: vppm2}
\eeq
Note that the functional $ \fqc $ and the corresponding variational problems
\eqref{eq: vp1} and \eqref{eq: vp2} are recovered by simply imposing in $
\fpm $ above that $ \mu $ is a Dirac
delta, {\it i.e.}, $ \exists z_0 \in
\mathfrak{h}_{\omega} $ such that $ \mu = \delta_{z_0} $. Yet another minimization
problem can be formulated by substituting the minimization over
$\mathscr{P}\lf(\mathfrak{h}_{\omega};\mathscr{L}^1_+(L^2(\R^{dN})) \ri) $ and $
\mathscr{P}\lf(\mathfrak{h}_{\omega}\ri) $ in \eqref{eq: fsv} and \eqref{eq: fpm}
with the one over atomic measures $
\mathscr{P}_{\mathrm{atom}}\lf(\mathfrak{h}_{\omega};\mathscr{L}^1_+(L^2(\R^{dN}))
\ri) $ and $ \mathscr{P}_{\mathrm{atom}}\lf(\mathfrak{h}_{\omega}\ri) $,
respectively.

Finally, in the spirit of derivation of effective functionals of $ \psi $ or $ z
$
alone, as the Pekar-like functionals defined in \eqref{eq: fpek N}  and
\eqref{eq: fpek PF}, we can also define the following effective energy
\beq
\mathcal{I}[z] : = \inf_{\psi \in L^2(\R^{dN}), \lf\| \psi \ri\|_2 = 1} \fqc[\psi, z].
\eeq


The rest of this section is devoted to prove equivalences between the
minimization problems defined above. In fact, we are interested mostly in
deriving information concerning \eqref{eq: vp1} and \eqref{eq: vp2}, obtained
by studying the quasi-classical limit. The latter, however, yields naturally
information about \eqref{eq: vpsv1} and \eqref{eq: vpsv2}, and thus the link
between the two quasi-classical minimization problems will be very
useful. Firstly, the infima of all the aforementioned functionals coincide.

\begin{proposition}[Quasi-classical energies]
  \label{prop:2}
  \mbox{}	\\
  Under the assumptions made,
  \begin{multline}
    \eqc = \esv = \inf_{\mm \in \domsv \cap \mathscr{P}_{\mathrm{atom}}(\mathfrak{h}_{\omega};\mathscr{L}^1_+(L^2(\R^{dN})))} \fsv[\mm] = \epm  \\
    = \inf_{\lf(\psi, \mu \ri) \in \dompm, \mu \in\mathscr{P}_{\mathrm{atom}}\lf(\mathfrak{h}_{\omega}\ri)} \fpm[\psi, \mu] = \epek= \inf_{z \in \mathfrak{h}_{\omega}} \mathcal{I}[z].
  \end{multline}
\end{proposition}
\begin{proof}
  We use the weak density of atomic scalar measures, supported on a finite
  number of points, in the space of all finite measures, that holds for
  $\mathfrak{h}_{\omega}$ separable \cite{parthasarathy1967pms}. Thanks to that it
  is possible to prove the following (see \cite[Lemma 3.20]{correggi2017ahp}
  for a detailed proof):
  \beqn
  \esv &=& \disp\inf_{\mm \in \domsv} \fsv[\mm] = \inf_{\mm \in \domsv \cap \mathscr{P}_{\mathrm{atom}}(\mathfrak{h}_{\omega};\mathscr{L}^1_+(L^2(\R^{dN})))} \fsv[\mm]	\; ;	\nonumber \\
  \epm &=& \disp\inf_{\lf(\psi, \mu \ri) \in \dompm} \fpm[\psi, \mu] = \inf_{\lf(\psi, \mu\ri)\in \dompm, \mu \in \mathscr{P}_{\mathrm{atom}}\lf(\mathfrak{h}_{\omega}\ri)} \fpm[\psi, \mu]. \nonumber
  \eeqn
    	
  Now, let us prove that
  \begin{equation}
    \inf_{\mm \in \domsv \cap \mathscr{P}_{\mathrm{atom}}(\mathfrak{h}_{\omega};\mathscr{L}^1_+(L^2(\R^{dN})))} \fsv[\mm] = \inf_{\lf(\psi, \mu \ri) \in \dompm, \mu \in \mathscr{P}_{\mathrm{atom}}\lf(\mathfrak{h}_{\omega}\ri)} \fpm[\psi, \mu].
  \end{equation}
  Let $\delta>0$ and let $\mathfrak{m}_{\delta}=\sum_{k=1}^{K}\lambda_k\gamma_k\delta_{z_k}$, with $ \gamma_k \in
  \mathscr{L}^1_{+,1}(L^2(\R^{dN})) $, $\lambda_k\geq 0$ (recall that $ \mm_{\delta} $
  takes values in positive operators) and $\sum_{k=1}^{K}\lambda_k=1$, be an atomic
  state-valued measure, such that
  \begin{equation*}
    \fsv\lf[\mathfrak{m}_{\delta}\ri] = \sum_{k=1}^{K} \lambda_k\tr_{L^2(\R^{dN})} \lf[\gamma_k \mathcal{H}_{z_k} \ri] < \inf_{\mm \in \domsv \cap \mathscr{P}_{\mathrm{atom}}(\mathfrak{h}_{\omega};\mathscr{L}^1_+(L^2(\R^{dN})))} \fsv[\mm] +\delta\; .
  \end{equation*}
  For fixed $k$, since $\gamma_k $ is a normalized density matrix,
  \begin{equation*}
    \inf_{\psi \in L^2(\R^{dN}), \lf\| \psi \ri\|_2 = 1}  \meanlrlr{\psi}{\mathcal{H}_{z_k}}{\psi}_{L^2(\R^{dN})} \leq \tr_{L^2(\R^{dN})} \lf[\gamma_k \mathcal{H}_{z_k} \ri].
  \end{equation*}
  Therefore,
  \begin{multline}
    \inf_{\lf(\psi, \mu \ri) \in \dompm, \mu \in \mathscr{P}_{\mathrm{atom}}\lf(\mathfrak{h}_{\omega}\ri)} \fpm[\psi, \mu] = \inf_{\lf(\psi, \mu \ri) \in \dompm, \mu \in \mathscr{P}_{\mathrm{atom}}\lf(\mathfrak{h}_{\omega}\ri)} \int_{\mathfrak{h}_{\omega}} \diff \mu \: \meanlrlr{\psi}{\HH_z}{\psi}_{L^2(\R^{dN})} \\
    \leq \sum_{k = 1}^K \lambda_k \inf_{\psi \in L^2(\R^{dN}), \lf\| \psi \ri\|_2 = 1}  \meanlrlr{\psi}{\mathcal{H}_{z_k}}{\psi}_{L^2(\R^{dN})} \leq \sum_{k=1}^{K} \lambda_k \tr_{L^2(\R^{dN})} \lf[\gamma_k \mathcal{H}_{z_k} \ri]\\
    < \inf_{\mm \in \domsv \cap\mathscr{P}_{\mathrm{atom}}(\mathfrak{h}_{\omega};\mathscr{L}^1_+(L^2(\R^{dN})))}\fsv[\mm] +\delta\; .  
  \end{multline}
  Since $\delta>0$ is arbitrary, we conclude that
  \begin{equation}
    \inf_{\lf(\psi, \mu \ri) \in \dompm, \mu \in \mathscr{P}_{\mathrm{atom}}\lf(\mathfrak{h}_{\omega}\ri)} \fpm[\psi, \mu] \leq \inf_{\mm \in \domsv \cap \mathscr{P}_{\mathrm{atom}}(\mathfrak{h}_{\omega};\mathscr{L}^1_+(L^2(\R^{dN})))} \fsv[\mm].
  \end{equation}
  To prove the opposite inequality, we follow a similar reasoning. Let $ \delta >
  0 $ and $\mu_{\delta}=\sum_{k=1}^{K}\lambda_k\delta_{z_k}$ be a scalar atomic measure and
  $\psi_{\delta,z_k} \in L^2(\R^{dN})$ a family of normalized wave functions, such that $
  \mu_{\delta}(\mathfrak{h}_{\omega}) = 1 $ and
  \begin{equation*}
    \sum_{k=1}^{K}\lambda_k \meanlrlr{\psi_{\delta,z_k}}{\mathcal{H}_{z_k}}{\psi_{\delta,z_k}}_{L^2(\R^{dN})} < \inf_{\lf(\psi, \mu \ri) \in \dompm, \mu \in \mathscr{P}_{\mathrm{atom}}\lf(\mathfrak{h}_{\omega}\ri)} \fpm[\psi, \mu]+\delta\; .
  \end{equation*}
  Now, $\mathfrak{m}_{\delta} := \sum_{k=1}^K\lambda_k \ket{\psi_{\delta,z_k}}\bra{\psi_{\delta,z_k}}
  \delta_{z_k}$ is an atomic state-valued measure belonging to
  $ \domsv
  $. Therefore,
  \begin{multline}
    \inf_{\mm \in \domsv \cap \mathscr{P}_{\mathrm{atom}}(\mathfrak{h}_{\omega};\mathscr{L}^1_+(L^2(\R^{dN})))} \fsv[\mm] \leq \fsv\lf[\mm_{\delta}\ri] = \sum_{k=1}^{K}\lambda_k \meanlrlr{\psi_{\delta,z_k}}{\mathcal{H}_{z_k}}{\psi_{\delta,z_k}}_{L^2(\R^{dN})} \\
    < \inf_{\lf(\psi, \mu \ri) \in \dompm, \mu \in \mathscr{P}_{\mathrm{atom}}\lf(\mathfrak{h}_{\omega}\ri)} \fpm[\psi, \mu] +\delta\;,
  \end{multline}  
  which yields the desired inequality.
  
  To complete the proof, we show that
  \begin{equation}
    \label{eq: prop2 proof 1}
    \inf_{\lf(\psi, \mu \ri) \in \dompm, \mu \in \mathscr{P}_{\mathrm{atom}}\lf(\mathfrak{h}_{\omega}\ri)} \fpm[\psi, \mu] = \inf_{z\in \mathfrak{h}_{\omega}}\mathcal{I}[z] = \epek = \eqc \; .
  \end{equation}
  Let us prove the first equality beforehand. Let $\mu_{\delta}=\sum_{k=1}^K\lambda_k
  \delta_{z_k}$ be the atomic 
  minimizing family of measures defined before and $\psi_{\delta,z_k}$ the corresponding
  minimizing vectors. Then,
  \begin{multline}
        \sum_{k = 1}^K \lambda_k \inf_{\psi \in L^2(\R^{dN}), \lf\| \psi \ri\|_2 = 1}  \meanlrlr{\psi}{\mathcal{H}_{z_k}}{\psi}_{L^2(\R^{dN})} \leq \sum_{k=1}^{K}\lambda_k \meanlrlr{\psi_{\delta,z_k}}{\mathcal{H}_{z_k}}{\psi_{\delta,z_k}}_{L^2(\R^{dN})} \\
    < \inf_{\lf(\psi, \mu \ri) \in \dompm, \mu \in \mathscr{P}_{\mathrm{atom}}\lf(\mathfrak{h}_{\omega}\ri)} \fpm[\psi, \mu]+\delta\; .  
  \end{multline}
  Since the l.h.s. is a convex combination and $\delta$ is arbitrary, we
  immediately deduce that
  \begin{equation}
    \inf_{z\in \mathfrak{h}_{\omega}}\mathcal{I}[z] \leq \inf_{\lf(\psi, \mu \ri) \in \dompm, \mu \in \mathscr{P}_{\mathrm{atom}}\lf(\mathfrak{h}_{\omega}\ri)} \fpm[\psi, \mu] \;.
  \end{equation}
  On the other hand, since a measure concentrated in a single point is
  atomic,
  \begin{equation*}
    \inf_{\lf(\psi, \mu \ri) \in \dompm, \mu \in \mathscr{P}_{\mathrm{atom}}\lf(\mathfrak{h}_{\omega}\ri)} \fpm[\psi, \mu] \leq \inf_{z\in \mathfrak{h}_{\omega}} \inf_{\psi \in L^2(\R^{dN}), \lf\| \psi \ri\|_2 = 1} \fqc[\psi,z] = \inf_{z\in \mathfrak{h}_{\omega}} \mathcal{I}[z]\;,
  \end{equation*}
  which implies the first identity in \eqref{eq: prop2 proof 1}.
  	
  Now, let us prove the second equality above, namely
  \begin{equation}
    \inf_{z\in \mathfrak{h}_{\omega}}\mathcal{I}[z] = \epek \; .
  \end{equation}
  Let again $\delta>0$ and let $z_{\delta}$ be a minimizing family of vectors for
  $\mathcal{I}$, {\it i.e.}, such that $ \mathcal{I}[z_{\delta}] < \inf_{z\in
    \mathfrak{h}_{\omega}}\mathcal{I}[z] + \delta $. For each $z_{\delta}$, let $\psi_{\delta,
    z_{\delta}}$ be a
  minimizing vector for $ \fqc[\: \cdot \:,z_{\delta}] $, {\it i.e.}, such that
  \begin{equation*}
    \fqc\lf[\psi_{\delta, z_{\delta}}, z_{\delta} \ri] <\mathcal{I}[z_{\delta}]+\delta\; .
  \end{equation*}
  Now,
  \begin{equation*}
    \epek \leq \fpek\lf[\psi_{\delta, z_{\delta}} \ri] \leq \fqc\lf[\psi_{\delta, z_{\delta}}, z_{\delta}\ri]\; .
  \end{equation*}
  so that,
  \begin{equation}
    \epek \leq \inf_{z\in \mathfrak{h}_{\omega}} \mathcal{I}[z]\; .
  \end{equation}
  On the other hand, let $\psi_{\delta}$ be a minimizing family of states for $ \epek
  $, and,
  once fixed $\psi_{\delta}$, let $z_{\delta, \psi_{\delta}} $ be a minimizing family for $ \fqc[\psi_{\delta},\: \cdot \: ]$:
  \begin{equation}
    \label{eq: prop2 proof 2}
    \fqc\lf[ \psi_{\delta}, z_{\delta,\psi_{\delta}} \ri] < \epek +\delta\; .
  \end{equation}
  As above, we then get
  \begin{equation*}
    \inf_{z\in \mathfrak{h}_{\omega}} \mathcal{I}[z] \leq \inf_{\psi \in L^2(\R^{dN}), \lf\| \psi \ri\|_2 = 1} \fqc\lf[\psi, z_{\delta, \psi_{\delta}} \ri] \leq \fqc\lf[\psi_{\delta}, z_{\delta, \psi_{\delta}} \ri] < \epek +\delta\; .
  \end{equation*}
  which yields
  \begin{equation}
    \inf_{z\in \mathfrak{h}_{\omega}}\mathcal{I}[z] \leq \epek \; .
  \end{equation}
  	
  Finally, we prove that
  \begin{equation}
    \epek = \eqc \; .
  \end{equation}
  Now, let $ (\psi_{\delta}, z_{\delta, \psi_{\delta}})$ be as above, \emph{i.e.}, such that \eqref{eq: prop2 proof 2} holds true.
  Hence,
  \begin{equation*}
    \eqc \leq \fqc\lf[\psi_{\delta}, z_{\delta,\psi_{\delta}} \ri] < \epek + \delta\; ,
  \end{equation*}
  and thus $ \eqc \leq \epek $. On the other hand, let $(\psi_{\delta},z_{\delta}) $ be a
  minimizing family of configurations for $ \fqc $:
  \begin{equation*}
    \fqc\lf[\psi_{\delta},z_{\delta}\ri]< \eqc +\delta\; .
  \end{equation*}
  Clearly, now one has
  \begin{equation*}
    \epek \leq \fpek[\psi_{\delta}] \leq \fqc[\psi_{\delta},z_{\delta}] <  \eqc + \delta\; ,
  \end{equation*}
  yielding the opposite inequality, {\it i.e.}, $ \epek \leq \eqc $.
\end{proof}

\begin{remark}[Stability]
  \label{rem: stability}
  \mbox{}	\\
  In the above proof we have implicitly assumed that the energies under
  considerations are bounded from below, but in fact it is easy to see that,
  if one of the functionals in unbounded from below, then all the others must
  be unstable as well. We do not provide any detail of such an argument,
  because our main result (\cref{teo: ground state energy}) implies that \eqref{eq: vp1} holds true, so that \eqref{eq: vpp1}, \eqref{eq: vpsv1}
  and \eqref{eq: vppm1} immediately follow.
\end{remark}

The other important result concerns equivalences for the existence of
minimizers in the variational problems above.

\begin{proposition}[Quasi-classical minimizers]
  \label{prop:3}
  \mbox{}	\\
  Under the assumptions made,
  \beq
  \eqref{eq: vp2} \; \Longleftrightarrow\; \eqref{eq: vpp2} \; \Longleftrightarrow\; \eqref{eq: vpsv2} \; \Longleftrightarrow\; \eqref{eq: vppm2} \;.
  \eeq
  Furthermore, any minimizer $ \msv $ of \eqref{eq: vpsv2} is concentrated on
  the set of minimizers $ (\psiqc, \zqc) $ of \eqref{eq: vp2}.
\end{proposition}

Before proving \cref{prop:3}, we state a useful result about the
quasi-classical functional defined in \eqref{eq: fqc} or, more precisely,
about its variant $ \ffqc $ introduced in \eqref{eq: qc energy computed},
which is important to explore the connection with the Pekar-like functionals
\eqref{eq: fpek N} and \eqref{eq: fpek PF}.
	
\begin{lemma}
  \label{lemma: convex}
  \mbox{}	\\
  For any fixed $ \psi $, the functional $ \ffqc[\psi,\eta] $ is strictly convex in $
  \eta \in \mathfrak{h}_{\omega} $.
\end{lemma}
	
\begin{proof}
  In cases (a) and (b) the proof is trivial, since $ \ffqc $ contains only
  two terms depending on $ \eta $: one is quadratic in $ \eta $ (the free field
  energy) and therefore strictly convex, while the other (the interaction) is
  linear and thus convex.
		
  So we have to investigate in detail only case (c), namely the Pauli-Fierz
  quasi-classical energy, and, specifically, only the kinetic part of the
  energy involving the interaction, which reads
  \bdm
  \sum_{j = 1}^N \tx\frac{1}{2m_j} \lf( - i \nabla_j + 2 \Re \braketr{\eta}{\lf(\omega^{-1/2} \bm{\lambda}_j \ri)(\xv_j)}_{\mathfrak{h}} \ri)^2.
  \edm
  Let us then set $ \eta = \beta \eta_1 + (1 - \beta) \eta_2 $ for some $ \eta_1, \eta_2 \in
  \mathfrak{h}$ and $ \beta \in (0,1) $. Expanding the
  square and setting $
  \bm{\xi}_j(\xv) : = \omega^{-1/2} \lambda_j(\xv) $ for short, we get (for any non-zero $ \psi $)
  \begin{multline}
    \meanlrlr{\psi}{\lf( - i \nabla_j + 2 \Re \braket{\eta}{\bm{\xi}_j (\xv_j)}_{\mathfrak{h}} \ri)^2}{\psi}_{L^2(\R^{3N})} < \meanlrlr{\psi}{- \Delta_j }{\psi}_{L^2(\R^{3N})} \\
    - 2 \meanlrlr{\psi}{i \beta  \Re \braket{\eta_1}{\bm{\xi}_j(\xv_j)}_{\mathfrak{h}} \cdot \nabla_j + i (1- \beta)  \Re \braket{\eta_2}{\bm{\xi}_j(\xv_j)}_{\mathfrak{h}} \cdot \nabla_j}{\psi}_{L^2(\R^{3N})}	\\
    +4 \meanlrlr{\psi}{\beta \lf( \Re \braket{\eta_1}{\bm{\xi}_j(\xv_j)}_{\mathfrak{h}} \ri)^2 + (1- \beta) \lf( \Re \braket{\eta_2}{\bm{\xi}_j(\xv_j)}_{\mathfrak{h}} \ri)^2}{\psi}_{L^2(\R^{3N})} 
  \end{multline}
  again by the strict convexity of the square, {\it i.e.}, the bound
  $ (\beta a +
  (1 - \beta) b)^2 < \beta a^2 + (1-\beta) b^2 $, valid for any $ a, b \in \R $ and $ \beta \in
  (0,1) $. The result
  easily follows, since the remaining term in the functional depending on $ \eta
  $
  is the free field energy, which is quadratic in $ \eta $ and thus strictly convex
  as well.
\end{proof}

\begin{proof}[Proof of \cref{prop:3}]
  Some implications are easy to prove. Let us first prove that $ \eqref{eq:
    vp2} \; \Longrightarrow \; \eqref{eq: vppm2} $. 
  Let $(\psiqc,\zqc)$ be a minimizer of $
  \fqc $ in $ \domqc $. Then, evaluating the energy $ \fpm $ on
  the configuration $(\psiqc, \mu_0)$, with $ \mu_0=\delta_{\zqc}$, we get
  \begin{equation*}
    \fpm\lf[ \psiqc, \mu_0 \ri] = \int_{\mathfrak{h}_{\omega}}^{}\mathrm{d}\mu_0(z) \:  \fqc\lf[ \psiqc, z\ri] = \fqc\lf[\psiqc, \zqc\ri] = \eqc\; .
  \end{equation*}
  By \cref{prop:2}, $(\psiqc,\mu_0)$ is thus solving \eqref{eq: vppm2}. Analogously, let us prove $ \eqref{eq: vppm2}\; \Longrightarrow\; \eqref{eq:
    vpsv2} $:
  let $(\psipm,\mupm)$ be a minimizer for \eqref{eq: vppm2}; then, the state-valued measure $ \mathfrak{m}_0$, with
  $\mu_{\mathfrak{m}_0}=\mupm$ and $\gamma_{\mathfrak{m}_0}(z)=
  \ket{\psipm}\bra{\psipm} $, solves \eqref{eq: vp2} by \cref{prop:2}.

  We prove now that $ \eqref{eq: vpsv2} \; \Longrightarrow\; \eqref{eq: vp2} $. Given a minimizer $ \msv $ of $ \fsv $, for $\mu_{\msv}$-a.e.
  $z\in \mathfrak{h}_{\omega}$ there exist $\{\lambda_k(z)\}_{k\in \mathbb{N}}$, $\lambda_k(z)\geq 0$, $\sum_{k\in \mathbb{N}}^{}\lambda_k(z)=1$ and
  $\{\psi_k(z)\}_{k\in \mathbb{N}}$, $ \lf\| \psi_k(z) \ri\|_{L^2(\R^{dN})}^{}=1$, such that
  \begin{equation*}
    \esv = \fsv\lf[ \msv \ri] = \int_{\mathfrak{h}_{\omega}}^{}  \diff \mu_{\msv}(z) \: \sum_{k\in \mathbb{N}}^{} \lambda_k(z) \fqc\lf[\psi_k(z),z \ri] \; .
  \end{equation*}
  The above is due to the fact that $ \gamma_{\msv}(z)$ is a density matrix on
  $L^2(\R^{dN})$ for $\mu_{\msv}$-a.e. $ z $. The measure $ \mu_{\msv} \in
  \mathscr{P}(\mathfrak{h}_{\omega}) $ is a probability measure,
  hence the r.h.s.\ of the above equation is a (double) convex combination of
  numerical values of the real-valued function $ \fqc $. However, a convex
  combination of values of a function equals its infimum, if and only if the
  infimum is a minimum, and all variables appearing in the convex combination
  are minimizers. Therefore, $ \fqc $ admits at least one minimizer. Actually,
  the measure $ \msv $ is concentrated on the set of minimizers $ (\psiqc,
  \zqc) $, in the
  above sense.
  
  Finally, we consider the Pekar-like variational problem \eqref{eq: vpp2}
  and its equivalence with \eqref{eq: vp2}. Let us first prove that
  \eqref{eq: vpp2}  $ \Longrightarrow $ \eqref{eq: vp2}: given a
  Pekar minimizer $ \psipek \in
  L^2(\R^{dN}) $, we immediately deduce that $ \psipek \in H^1(\R^{dN}) $ by boundedness from above of the energy and regularity of the classical field $ \av(\xv) $, which is continuous and vanishing at infinity \cite[Rmk. 1.5]{correggi2017arxiv}. Furthermore, \cref{lemma: convex} guarantees the existence (and uniqueness) of  $ \zpek[\psipek] \linebreak\in \mathfrak{h} $ minimizing $ \fqc[\psipek, z] $ w.r.t. $ z $. 
  	Therefore, the configuration $ (\psipek, \zpek[\psipek]) $ is admissible for $ \fqc $ and we deduce from \cref{prop:2} that $ \fqc[\psipek, \zpek[\psipek]] = \eqc $. 
  	
  Conversely, given a minimizer $ (\psiqc, \zqc) \in \dom $ of $ \fqc $, we
  know that the configuration must satisfy the Euler-Lagrange equations
  \eqref{eq: el} at least in weak sense. However, the second equation in
  \eqref{eq: el} is easily seen to coincide with \eqref{eq: zpek N} or the
  first equation in \eqref{eq: fpek PF}, when the change of variable $ \eta =
  \omega^{1/2} z $ has been done.
  Furthermore, any weak solution $ \eta $ of such equations is in fact a strong
  solution, {\it i.e.}, $ \eta \in \mathfrak{h} $, under the assumptions made. Hence, by strict
  convexity of $ \ffqc[\psi,\eta] $ in $ \eta $ proven in \cref{lemma: convex} and then uniqueness of $ \zpek $, we
  deduce that $ \zpek[\psiqc] = \omega^{1/2} \zqc $ and the equivalence $
  \eqref{eq: vp2} \; \Longrightarrow \: \eqref{eq: vpp2} $ is readily
  proven via \cref{prop:2}.
\end{proof}

\begin{remark}[Minimizers for \eqref{eq: vppm2}]
  \label{rem:4}
  \mbox{}	\\
  The existence of a solution for \eqref{eq: vppm2} obtained here is trivial,
  \emph{i.e.}, it involves a measure concentrated in a single point $\zqc \in
  \mathfrak{h}_{\omega}$ and a $\psi_{\zqc}$
  dependent on such a point. It would be interesting, but outside the scope
  of this paper, to know whether there are non-trivial minimizers in which $\mu_0$
  is not concentrated at a single point. This is obviously related to the
  question of uniqueness of the minimizing configuration $ (\psiqc, \zqc)
  $. Note that
  this would not be in contradiction with \cref{lemma: convex}, since we prove there strict
  convexity of $ \ffqc[\psi,\eta] $ only in $ \eta $, while the full functional $
  \fqc[\psi,z] $ is in
  general not jointly convex in $ \psi $ and $ z $ nor in $
  \lf| \psi \ri|^2 $ and $ z $ (see also \cref{rem: uniqueness pekar}).
\end{remark}

Note that the combination of \cref{prop:2} with \cref{prop:3} provides the proof of \cref{pro: equivalence} stated in
\cref{sec: intro}.

\subsection{Minimization problem for generalized state-valued measures}
\label{sec:minim-probl-gener}

We discuss now the generalization of the concepts introduced above needed to
deal with the minimization \eqref{eq: fqct}, that is particularly useful to
treat small systems consisting of unconfined particles. Taking the double
dual, it is well known that $\mathscr{L}^1(L^2(\R^{dN}))$ can be continuously
embedded in $\mathscr{B}(L^2(\R^{dN}))'$, the dual of bounded operators, in a
positivity preserving way. By an abuse of notation, we will write
$\mathscr{L}^1(L^2(\R^{dN}))\subset \mathscr{B}(L^2(\R^{dN}))'$. We recall that we
denoted by $ \overline{\mathscr{L}^1}(L^2(\R^{dN})) $ the closure of
$\mathscr{L}^1(L^2(\R^{dN}))$ with respect to the weak-* topology
$\sigma\bigl(\mathscr{B}(L^2(\R^{dN}))',\mathscr{B}(L^2(\R^{dN}))\bigr)$ on
$\mathscr{B}(L^2(\R^{dN}))'$. Also, $
\overline{\mathscr{L}^1}_+(L^2(\R^{dN})) $ and $
\overline{\mathscr{L}^1}_{+,1}(L^2(\R^{dN})) $ stand for the subsets of
positive and normalized positive elements, respectively. A generalized
state-valued measure is then a measure on $ \mathfrak{h}_{\omega} $ with values in
the space of generalized states $ \overline{\mathscr{L}^1}_{+}(L^2(\R^{dN}))
$. Properties of generalized
state-valued measures are discussed in
\cref{sec:gener-state-valu}. Since the dual space $\mathscr{B}(L^2(\R^{dN}))'$ is not separable, it does not have the
Radon-Nikod\'{y}m property, therefore integration of functions
$\mathcal{F}:\mathfrak{h}_{\omega}\to \mathscr{B}(L^2(\R^{dN}))$ is restricted only to ones with separable range.

Such integration can be extended to functions valued in unbounded operators
in the following sense.

\begin{definition}[Domains of generalized Wigner measures]
  \label{def: domain g wigner}
  \mbox{}	\\
  Let $\mathcal{T}$ be a strictly positive unbounded operator on $L^2(\R^{dN})$. A generalized
  state-valued measure $\mathfrak{n}$ is {\it in the domain of} $\mathcal{T}$, if and only if there
  exists a measure $\mathfrak{n}_{\mathcal{T}} \in \mathscr{P}(\mathfrak{h}_{\omega},
  \overline{\mathscr{L}^1}_{+}(L^2(\R^{dN}))) $, such that for all $\mathcal{B}\in
  \mathscr{B}(L^2(\R^{dN}))$ and
  all Borel sets $ S \subseteq \mathfrak{h}_{\omega}$,
  \begin{equation}
    \mathfrak{n}_{\mathcal{T}}(S)\lf[\mathcal{T}^{-1/2}\mathcal{B}\mathcal{T}^{-1/2} \ri]=\mathfrak{n}(S) \lf[\mathcal{B} \ri]\; .
  \end{equation}
\end{definition}

Therefore, if $ \nn $ is in the domain of $ \mathcal{T} $, with a little abuse of notation, we
may write
\begin{equation}
  \mathfrak{n}(S)\lf[\mathcal{T}^{1/2}\,\cdot \,\mathcal{T}^{1/2}\ri] = \mathfrak{n}_{\mathcal{T}}(S)[\,\cdot \,]
\end{equation}
as a state valued measure ``absorbing a singularity'' of order
$\mathcal{T}$. Now, let $\mathcal{F}(z)$ be a function with values in
unbounded operators such that for all $z\in \mathfrak{h}_{\omega}$:
\begin{itemize}
\item $\mathcal{T}^{-1/2}\mathcal{F}(z)\mathcal{T}^{-1/2}\in
  \mathscr{B}(L^2(\R^{dN}))$;  
\item the range of $z\mapsto \mathcal{T}^{-1/2}\mathcal{F}(z)\mathcal{T}^{-1/2}$ is separable;
\item $\mathcal{T}^{-1/2}\mathcal{F}(z)\mathcal{T}^{-1/2}$ is $\mathfrak{n}_{\mathcal{T}}$-absolutely integrable.
\end{itemize}
Then, it follows that we can define the integral of $\mathcal{F}$ with
respect to $\mathfrak{n}$ as
\begin{equation}
  \int_{\mathfrak{h}_{\omega}}^{}\mathrm{d}\mathfrak{n}(z)\bigl[ \mathcal{F}(z) \bigr]:=\int_{\mathfrak{h}_{\omega}}^{}  \mathrm{d}\mathfrak{n}_{\mathcal{T}}(z)\bigl[ \mathcal{T}^{-1/2}\mathcal{F}(z)\mathcal{T}^{-1/2} \bigr]\; .
\end{equation}
A simple but useful example of such $\mathcal{F}(z)$ is the following: let
$\mathcal{S}$ be a self-adjoint operator, and let $\mathfrak{n}$ be in the
domain of $\mathcal{T}=\lvert \mathcal{S} \rvert_{}^{}+1$; then the function
$\mathcal{F}(z)=\mathcal{S}$ satisfies all above hypotheses and thus it makes
sense to write, for all Borel set $ S \subseteq \mathfrak{h}_{\omega}$,
\begin{equation}
  \int_{S}^{}  \mathrm{d}\mathfrak{n}(z)[\mathcal{S}]= \mathfrak{n}(S)[\mathcal{S}] := \mathfrak{n}_{\mathcal{T}}(S)\bigl[ \mathcal{T}^{-1/2}\mathcal{S} \mathcal{T}^{-1/2} \bigr] \in \mathbb{R}\; .
\end{equation}
The other cases useful for our analysis are discussed in
\cref{sec:ground-states-quasi}.

We are now in a position to define another quasi-classical minimization
problem. Recall the definition \eqref{eq: domgqc} of the domain $ \domgqc $,
the ground state energy $ \egqc $ given by \eqref{eq: egqc} and any
corresponding minimizing configuration $ ( \rhogqc, \zgqc ) \in \domgqc $; then
the analogues of \eqref{eq: vp1} and \eqref{eq: vp2} are
\beq
\egqc \overset{?}{>} -\infty,
\tag{GVP1}\label{eq: GVP1}
\eeq
\beq
\overset{?}{\exists} \lf(\rhogqc, \zgqc \ri) \in \domgqc\; ;
\tag{GVP2}\label{eq: GVP2}
\eeq
The functional $ \fgqc $ can indeed be seen as the generalized
quasi-classical energy: let $ H_z $ be the abstract realization of $
\mathcal{H}_z $ as an
operator affiliated to the abstract $C^*$-algebra $
\mathscr{B}(L^2(\R^{dN}))$. Then, given a
normalized pure state $ \rho \in
\overline{\mathscr{L}^1}_+(L^2(\R^{dN})) $, we define the corresponding
irreducible GNS representation by $(\mathscr{K}_{\rho},\pi_{\rho},\psi_{\rho})$, where
$\mathscr{K}_{\rho}$ is a suitable Hilbert space, $\pi_{\varrho}:
\mathscr{B}(L^2(\R^{dN})) \to \mathscr{B}(\mathscr{K}_{\rho})$ is a
$C^*$-homomorphism (that can be extended to operators affiliated to the
algebra) and $\psi_{\rho}\in \mathscr{K}_{\varrho}$ is the normalized cyclic vector
associated to $ \rho $. Therefore, it follows that
\begin{equation*}
  \fgqc[\rho,z]= \meanlrlr{\psi_{\rho}}{\pi_{\rho}\lf(H_z\ri)}{\psi_{\rho}}_{\mathscr{K}_{\rho}}\; .
\end{equation*}
This expression is analogous to the one for $ \fqc $ (see \eqref{eq: fqc})
and it reduces exactly to the latter whenever $ \rho$ is a pure state belonging
to $\mathscr{L}^1(L^2(\R^{dN}))$ (see next \cref{rem:5}).
  
The generalization of the variational problems for state-valued measures
\eqref{eq: vpsv1} and \eqref{eq: vpsv2} is obtained as follows: setting 
\beq
\domgsv : = \lf\{ \nn \in \overline{\mathscr{L}^1}_+(L^2(\R^{dN})) \: \bigg| \: \lf\| \nn(\mathfrak{h}_{\omega}) \ri\|_{\mathscr{B}'} = 1, \lf| \int_{\mathfrak{h}_{\omega}} \diff \nn(z)\lf[ \HH_z \ri] \ri| < + \infty \ri\},
\eeq
we consider
\beq
\egsv : = \inf_{\nn \in \domgsv} \int_{\mathfrak{h}_{\omega}} \diff \nn(z)\lf[ \HH_z \ri] \overset{?}{>} -\infty,
\tag{gvp1}\label{eq: gvp1}
\eeq \beq
\overset{?}{\exists} \nsv \in \domgsv \mbox{ s.t. } \int_{\mathfrak{h}_{\omega}} \diff
\nsv(z)\lf[ \HH_z \ri] = \egsv. \tag{gvp2}\label{eq: gvp2} \eeq

\begin{remark}[State-valued and generalized state-valued measures]
  \label{rem:5}
  \mbox{}	\\
  We point out that, if a generalized state-valued measure $\mathfrak{n} \in
  \domgsv $ is
  actually a state-valued measure, \emph{i.e.}, such that, for all Borel sets $ S \subseteq \mathfrak{h}_{\omega}$,
  \begin{equation*}
    \mathfrak{n}(S)\in \mathscr{L}^1_+(L^2(\R^{dN}))\; ,
  \end{equation*}
  then $ \nn \in \domsv $ and
  \begin{equation*}
    \int_{\mathfrak{h}_{\omega}}^{}\mathrm{d}\mathfrak{n}(z) \lf[\mathcal{H}_z \ri] = \fsv[\nn].
  \end{equation*}
\end{remark}

\begin{proposition}[Generalized quasi-classical ground state energy]
  \label{prop:5}
  \mbox{}	\\
  Under the assumptions made above,
  \begin{equation}
    \eqc = \egqc = \egsv.
  \end{equation}
\end{proposition}
\begin{proof}
  Firstly, let us prove that
  \begin{equation*}
    \eqc = \egqc.
  \end{equation*}
  Since $ \rho$ belongs to the weak-$*$ closure of $\mathscr{L}^1_{+,1}(L^2(\R^{dN}))$, there exists a filter
  base $\mathfrak{S}\subset 2^{\mathscr{L}^1_{+,1}(L^2(\R^{dN}))}$ such that $\mathfrak{S}\to \rho $ in weak-$*$ topology. Hence, for any fixed
  $z\in
  \mathfrak{h}_{\omega}$,\footnote{The notation
    $\tr_{L^2(\R^{dN})}\lf[\mathfrak{S}(\mathcal{H}_z) \ri]$ stands for
    filter base that is image of $\mathfrak{S}$ on $\mathbb{R}$ via the map $\gamma \mapsto
    \tr_{L^2(\R^{dN})} \lf[\gamma \mathcal{H}_z \ri]$: given any $ X \in
    \mathfrak{S}$, $ \lf\{ \tr_{L^2(\R^{dN})} \lf[\gamma \mathcal{H}_z \ri], \gamma\in X
    \ri \} \in \tr_{L^2(\R^{dN})} \lf[ \mathfrak{S}(\mathcal{H}_z) \ri]$.}
  \begin{equation*}
    \lim_{\mathfrak{S} \to \rho} \tr_{L^2(\R^{dN})} \lf[\mathfrak{S}(\mathcal{H}_z) \ri] = \rho \lf[\mathcal{H}_z \ri]\; .
  \end{equation*}
  Now, on one hand, each $ \ket{\psi}\bra{\psi}$, $ \psi \in L^2(\R^{dN}) $, is also a pure generalized state
  and therefore
  \begin{equation}
    \label{eq: prop5 proof 1}
    \egqc \leq \inf_{(\psi,z) \in \domqc} \fqc[\psi,z] = \eqc\;.
  \end{equation}
  On the other hand, let $(\rho_{\delta},z_{\delta}) \in \domgqc $ be a minimizing sequence:
  \begin{equation*}
    \fgqc\lf[\rho_{\delta},z_{\delta}\ri] = \rho_{\delta} \lf[\mathcal{H}_{z_{\delta}} \ri] < \egqc + \delta\; ;
  \end{equation*}
  for some $ \delta > 0 $, and $\mathfrak{S}_{\delta}\subset
  2^{\mathscr{L}^1_{+,1}(L^2(\R^{dN}))}$ the corresponding approximating
  filter base for $\rho_{\delta}$. Then,
  \begin{multline}
    \eqc = \inf_{(\psi,z) \in \domqc} \fqc[\psi,z] = \inf_{(\gamma,z)\in \mathscr{L}^1_{+,1}(L^2(\R^{dN}))\oplus \mathfrak{h}_{\omega}} \tr_{L^2(\R^{dN})} \lf[\gamma \mathcal{H}_z \ri] \leq \sup_{X \in \mathfrak{S}_{\delta}} \inf_{\gamma \in X} \tr_{L^2(\R^{dN})} \lf[ \gamma \mathcal{H}_{z_{\delta}} \ri]\\
    =\liminf_{\mathfrak{S}_{\delta}} \tr_{L^2(\R^{dN})} \lf[\mathfrak{S}_{\delta}(\mathcal{H}_{z_{\delta}}) \ri] = \lim_{\mathfrak{S}_{\delta}\to \rho_{\delta}} \tr_{L^2(\R^{dN})} \lf[\mathfrak{S}_{\delta}(\mathcal{H}_{z_{\delta}}) \ri] = \rho_{\delta} \lf[\mathcal{H}_{z_{\delta}} \ri] < \egqc + \delta\; .
  \end{multline}
  Since the above chain of inequalities is valid for all $\delta>0$, it follows
  that the opposite inequality of \eqref{eq: prop5 proof 1} holds true, {\it i.e.},
  \begin{equation}
    \eqc \leq \egqc
  \end{equation}
  which implies the claim.
  	
  The proof of the identity $ \egsv = \eqc $ is perfectly analogous, keeping
  in mind that it is possible to approximate any measure $ \mathfrak{n} \in
  \mathscr{P}(\mathfrak{h}_{\omega}, \overline{\mathscr{L}^1}_+(L^2(\R^{dN})) $
  with a filter base $ \mathfrak{T} \subset 2^{\mathscr{P}(\mathfrak{h}_{\omega},
    \mathscr{L}^1_+(L^2(\R^{dN}))}$ w.r.t. the product of weak-$*$ topologies
  \begin{equation*}
    \prod_{S \subset \mathfrak{h}_{\omega} \text{ Borel}} \sigma\bigl(\mathcal{B}(L^2(\R^{dN}))',\mathcal{B}(L^2(\R^{dN}))\bigr)\; ,
  \end{equation*}
  that implies the convergence of integrals\footnote{As before, the integral
    w.r.t. to $ \diff \mathfrak{T} $ is just a short-hand notation to denote
    the integral over elements belonging to the filter $ \mathfrak{T} $.}
  \begin{equation*}
    \lim_{\mathfrak{T} \to \mathfrak{n}} \tr_{L^2(\R^{dN})} \lf[ \int_{\mathfrak{h}_{\omega}} \diff \mathfrak{T}(z) \: \mathcal{H}_z \ri] = \int_{\mathfrak{h}_{\omega}}^{}\mathrm{d}\mathfrak{n}(z) \lf[ \mathcal{H}_z \ri]\; .
  \end{equation*}
\end{proof}

Finally, also for the generalized minimization problems, it is possible to
prove equivalence of existence of minimizers.

\begin{proposition}[Generalized quasi-classical minimizers]
  \label{prop:6}
  \mbox{}	\\
  Under the assumptions made,
  \beq
  \eqref{eq: GVP2} \; \Longleftrightarrow\; \eqref{eq: gvp2} \;.
  \eeq
  Furthermore, any minimizer $ \nsv $ of \eqref{eq: gvp2} is concentrated on the set of
  minimizers $ (\rhogqc, \zgqc) $ of \eqref{eq:
    GVP2}.
\end{proposition}
\begin{proof}
  The $ (\Longrightarrow) $ implication is trivial: let $(\rhogqc,\zgqc)$ be a minimizer for
  \eqref{eq: GVP2}, then, evaluating the energy of the generalized
  state-valued measure $ \mathfrak{n}_0 = \delta_{\zgqc} \rhogqc$, we get
  \begin{equation}
    \int_{\mathfrak{h}_{\omega}}^{}\mathrm{d}\mathfrak{n}_0(z) \lf[ \mathcal{H}_z \ri]= \rhogqc\lf[\mathcal{H}_{\zgqc} \ri] = \egqc \; .
  \end{equation}
  By \cref{prop:5}, $\mathfrak{n}_0$ is thus a minimizer for \eqref{eq:
    gvp2}.

  To prove the converse implication, note that the integral w.r.t. a
  generalized state-valued probability measure is a convex combination of
  expectations over possibly mixed generalized states. Since the mixed states
  are themselves convex combinations of pure states, it follows that the
  measure $ \nsv $ must be concentrated on the set of minimizers for
  \eqref{eq: gvp2}, and thus the latter is not empty.
\end{proof}

\section{Ground States Energies and Ground States in the Quasi-Classical
  Regime}
\label{sec:ground-states-quasi}

In this section we study the quasi-classical limit of ground state energies
and ground states of the microscopic models introduced in \cref{sec: intro}.


The microscopic interaction is described by a fully quantum system, in which
both the small system and the environment are quantum. The Hilbert space is
thus (see \eqref{eq: hilbert}) given by $ \mathscr{H}_{\eps} = L^2(\R^{dN})\otimes
\mathcal{G}_{\epsilon}(\mathfrak{h}) $, where $\mathcal{G}_{\epsilon}(\mathfrak{h}) =\bigoplus_{n\in \mathbb{N}}\mathfrak{h}^{\otimes_{\mathrm{s}} n} $ is
the symmetric Fock space over $\mathfrak{h}$ and $\varepsilon$ is the quasi-classical
parameter whose dependence is yielded by a semiclassical choice of canonical
commutation relations \eqref{eq: ccr}, {\it i.e.}, $
[a_{\varepsilon}(z),a^{\dagger}_{\varepsilon}(w)]=\varepsilon \braket{z}{w}_{\mathfrak{h}} $, with $a_{\varepsilon}^{\sharp}$
the annihilation and creation operators on the Fock space. A state of the
whole system is given by a density matrix
\begin{equation*}
  \Gamma_{\varepsilon}\in \mathscr{L}^1_{+,1} \lf(L^2(\R^{dN})\otimes \GG_{\eps}(\mathfrak{h}) \ri)\; ,
\end{equation*}
the positive trace-class operators with unit trace.

The dynamics of the system is described by a self-adjoint Hamiltonian
operator $H_{\varepsilon}$ whose general form is given in \eqref{eq: hamiltonian}. Such
operator is the partial Wick quantization of the quasi-classical Schrödinger
energy operator $ \mathcal{H}_z $ provided in \eqref{eq: classical
  hamiltonian}. Wick quantization consists in
substituting each $z$ appearing in $ \HH $ with $a_{\varepsilon}$ and each $\bar{z}$ with $a_{\varepsilon}^{\dagger}$, and of ordering
all $a^{\dagger}_{\varepsilon}$-s to the left of all $a_{\varepsilon}$-s. Such quantization procedure
is well-defined for symbols $\mathcal{F}_z $ that are polynomial in $z$ and
${z}^* $, as it is the case in concrete models we are considering (see
\cref{sec:concrete-models} for additional details and \cite{ammari2008ahp} for the rigorous procedure). Hence, we can write,
\begin{equation}
  H_{\varepsilon}= \mathrm{Op}^{\mathrm{Wick}}_{\varepsilon} \lf( \HH_{z} \ri)\;,
\end{equation}
and, more precisely, $\mathcal{H}_z$ can be split in three terms, at least in
the sense of quadratic forms, {\it i.e.},
\begin{equation}
  \mathcal{H}_z=\mathcal{K}_0+ \sum_{i = 1}^N \mathcal{V}_z(\xv_i) + \meanlrlr{z}{\omega}{z}_{\mathfrak{h}}\; ,
\end{equation}
with $\mathcal{K}_0$ self-adjoint and bounded from below, yielding
\begin{equation}
  H_{\varepsilon}=\mathcal{K}_0 \otimes 1+ 1 \otimes \mathrm{Op}^{\mathrm{Wick}}_{\varepsilon} \lf(\meanlrlr{z}{\omega}{z}_{\mathfrak{h}} \ri) +  \sum_{i = 1}^N \mathrm{Op}^{\mathrm{Wick}}_{\varepsilon} \lf( \mathcal{V}_z(\xv_i) \ri)\; ,
\end{equation}
as a quadratic form. The first and second terms on the r.h.s are the free
energies of the small system and environment, respectively, and the third
term is the \emph{small system-environment} interaction.

The minimization problem for the quantum system described by $H_{\varepsilon}$ is
defined in \eqref{eq: eeps}: the microscopic ground state energy is $
E_{\eps} : = \inf \sigma(H_{\eps}) $, while $ \psigs $
stands for any corresponding minimizer. Such a minimization problem has been
thoroughly studied, for the concrete models under consideration in this
paper; for bibliographical references the reader shall consult
\cref{sec:concrete-models}. The results are as follows.

\begin{proposition}[Stability and existence of the ground state]
  \label{prop:7}
  \mbox{}	\\
  Under the assumptions made, there exist finite constants $ c,C > 0 $
  independent of $ \eps $, such that
  \begin{equation}
    \label{eq: ul bounds}
    -c \leq E_{\varepsilon} \leq C\; .
  \end{equation}
  Furthermore, under suitable conditions on the operator $ \mathcal{K}_0 $
  (\emph{e.g.}, if $\mathcal{K}_0$ has compact resolvent), then $ E_{\varepsilon}\in
  \sigma_{\mathrm{pp}}(H_{\varepsilon})$ and thus $ \exists \psigs $ ground
  state of $ H_{\eps} $.
\end{proposition}

The proof of the above results is model-dependent and therefore it is
postponed to \cref{sec:concrete-models}.

We now investigate the link between the microscopic ground state problem and
the quasi-classical minimization problems described in
\cref{sec:quasi-class-minim}, starting from the proof of \cref{teo: ground
  state energy}. Recall the definition of quasi-classical 
and generalized quasi-classical Wigner measure defined in \cref{def: wigner}
and \cref{def: g wigner}, respectively. Although both cases could be treated at once, we
provide a separate discussion of the main results for trapped and non-trapped
particle systems, whose difference is apparent in the statements of
\cref{teo: minimizers 1} and \cref{teo: minimizers 2}.

\subsection{Trapped particle systems}
\label{sec: trapped}

The proof of \cref{teo: ground state energy} is divided into two steps (upper
and lower bounds for the microscopic energy). At the end of this section, we
also complete the proof of \cref{teo: minimizers 1} about the convergence of
minimizers.

\subsubsection{Energy upper bound}
\label{sec:upper-bound}

In the following, we denote by $\Psi_{\varepsilon,\delta}\in \dom(H_{\varepsilon})$, $\delta>0$, a minimizing
sequence for $H_{\varepsilon}$:
\begin{equation}
  \label{eq: minimizing}
  \meanlrlr{\Psi_{\varepsilon, \delta}}{H_{\varepsilon}}{\Psi_{\varepsilon,\delta}}_{\mathscr{H}_{\varepsilon}} < E_{\varepsilon}+\delta\; .
\end{equation}
The first step towards the proof of the energy convergence is given by the
proposition below.

\begin{proposition}[Energy upper bound]
  \label{pro: upper bound}
  \mbox{}	\\
  Under the above assumptions,
  \begin{equation}
    \label{eq: upper bound}
    \limsup_{\varepsilon\to 0} E_{\varepsilon} \leq \eqc.
  \end{equation}
\end{proposition}

In order to prove the upper bound we use a coherent trial state: let us
denote by $ \Omega_{\varepsilon}\in \GG_{\varepsilon}(\mathfrak{h}) $ the Fock vacuum and let
\begin{equation}
  \label{eq: trial state}
  \Xi_{\varepsilon}[\psi,z] := \psi \otimes W_{\varepsilon} \lf(\tfrac{z}{i\varepsilon} \ri) \Omega_{\varepsilon}\;,
\end{equation}
be a coherent product state constructed over the particle state $ \psi $ and the
classical configuration $ z \in \mathfrak{h} $. We shall restrict to $ \psi \in
\mathscr{Q}(\mathcal{K}_0)$, where $ \mathscr{Q}(\mathcal{K}_0)$ is
the form domain of $\mathcal{K}_0$, and $ z \in \mathfrak{h}$ such that $ \omega^{1/2} z \in
\mathfrak{h} $. As discussed in \cref{sec:concrete-models}, this
is sufficient to make $\Xi_{\varepsilon}[\psi,z]\in \mathscr{Q}(H_{\varepsilon})$ and $(\psi,z) \in \domqc
$. The energy of the above 
trial state is provided in the next lemma.

\begin{lemma}
  \label{lemma:1}
  \mbox{}	\\
  Under the above assumptions,
  \begin{equation}
    \meanlrlr{\Xi_{\varepsilon}[\psi,z]}{H_{\varepsilon}}{\Xi_{\varepsilon}[\psi,z]}_{\mathscr{H}_{\varepsilon}} = \fqc[\psi,z] + o_{\varepsilon}(1)\; .
  \end{equation}
\end{lemma}
\begin{proof}
  The proof of the result depends on the microscopic model involved.  The computation of the expectation over the trial states \eqref{eq: trial state} can be found in \cite[Prop. 3.11 \& Sect. 3.6]{correggi2017ahp}, for the Nelson and polaron models, and in \cite[Proof of Thm. 1.9]{correggi2017arxiv}, for the Pauli-Fierz model, respectively.
\end{proof}

\begin{proof}[Proof of \cref{pro: upper bound}]
 		By \cref{lemma:1}, we have that 
  		\beq 
  			E_{\varepsilon} \leq \inf_{(\psi,z) \in \domqc} \meanlrlr{\Xi_{\varepsilon}[\psi,z]}{H_{\varepsilon}}{\Xi_{\varepsilon}[\psi,z]}_{\mathscr{H}_{\varepsilon}} = \inf_{(\psi,z) \in \domqc} \fqc[\psi,z] + o_{\varepsilon}(1) = \eqc + o_{\eps}(1) \;.
    	\eeq
   		The result is then obtained by taking the $\limsup_{\varepsilon\to 0}$ on both sides.
\end{proof}

\subsubsection{Energy lower bound}
\label{sec:lower-bound}

The symmetric result of \cref{pro: upper bound} is stated in the following proposition.

\begin{proposition}[Energy lower bound]
  \label{pro: lower bound}
  \mbox{}	\\
  Under the above assumptions,
  \begin{equation}
    \label{eq: lower bound}
    \liminf_{\varepsilon\to 0} E_{\varepsilon} \geq \eqc.
  \end{equation}
\end{proposition}

Although not necessary in principle, we find convenient to present two
different proofs of \eqref{eq: lower bound}, one valid only when
$\mathcal{K}_0$ has compact resolvent, \emph{e.g.}, when the small system is
trapped, one valid for non-trapped small systems as well. The main reason is
that the former does not require the use of generalized Wigner measures,
since conventional state-valued measures are sufficient, resulting in a more
accessible proof.

If $ \mathcal{K}_0 $ has compact resolvent, the set of quasi-classical Wigner
measures (as in \cref{def: wigner}) associated with minimizing sequences for
$H_{\varepsilon}$ is not empty. In addition, the expectation of
$\mathrm{Op}_{\varepsilon}^{\mathrm{Wick}} (\mathcal{V}_z)$ converges to the
quasi-classical integral of $\mathcal{V}_z $. Let us formulate some
preliminary results about the convergence of the expectation values of the
operators involved. Such results rely on suitable a priori bounds on the
family of states $ \Psi_{\eps} \in \mathscr{H}_{\eps} $, as $ \eps $ varies in $
(0,1) $. \cref{lemma:5} below guarantees that there
exists a minimizing sequence $ \Psi_{\eps,\delta} $ in the sense of \eqref{eq: minimizing} satisfying
such bounds.

\begin{lemma}
  \label{lemma:2}
  \mbox{}	\\
  If $\mathcal{K}_0$ has compact resolvent and there exist $ C < +\infty $ such
  that, uniformly w.r.t $\varepsilon\in (0,1)$,
  \begin{equation}
    \lf| \meanlrlr{\Psi_{\eps}}{\lf(\mathcal{K}_0 + \mathrm{d}\mathcal{G}_{\varepsilon}(\omega) +  1 \ri)}{\Psi_{\varepsilon}}_{\mathscr{H}_{\varepsilon}}  \ri| \leq C\;,
  \end{equation}
  then $ \mathscr{W}\lf(\Psi_{\varepsilon}, \varepsilon\in (0,1) \ri) \neq \varnothing $. Furthermore, if $ \Psi_{\varepsilon_n}
  \xrightarrow[\eps_n \to 0]{\mathrm{qc}} \mathfrak{m} $, then
  $\tr_{L^2(\R^{dN})} \lf[\gamma_{\mathfrak{m}}(z) \mathcal{K}_0 \ri]$ is $\mu_{\mathfrak{m}}$-a.e. finite and $\mu_{\mathfrak{m}}$-absolutely
  integrable, and
  \beq
  \lim_{n\to + \infty} \meanlrlr{\Psi_{\varepsilon_n}}{\mathcal{K}_0}{\Psi_{\varepsilon_n}}_{\mathscr{H}_{\varepsilon_n}} = \int_{\mathfrak{h}_{\omega}}^{}\mathrm{d}\mu_{\mathfrak{m}}(z) \: \tr_{L^2(\R^{dN})} \lf[\gamma_{\mathfrak{m}}(z)\mathcal{K}_0 \ri] \; .
  \eeq
\end{lemma}
\begin{proof}
  For $\omega=1$ this proposition is proved in \cite[Props.\ 2.3 \&
  2.6]{correggi2019arxiv}. For a generic $\omega\geq 0$, the proof (in presence of
  semiclassical degrees of freedom only) can be found in
  \cite[Thm.\ 3.3]{falconi2017ccm}; the extension to the quasi-classical setting is
  straightforward, testing with compact observables of the small system, as
  in the aforementioned \cite[Props.\ 2.3 \& 2.6]{correggi2019arxiv}. Let us
  stress that the fact that all Wigner measures are probability measures,
  {\it i.e.}, there is no loss of mass and $ \mm(\mathfrak{h}_{\omega}) = 1 $, is
  due to the fact that $\mathcal{K}_0$ has compact resolvent. Otherwise,
  there may be a loss of probability mass due to the interplay between the
  particle system and the environment (see \cite[Cor.\ 1.7 \&
  Rmk.\ 1.9]{correggi2019arxiv} for additional details).
\end{proof}

In order to control the convergence of the free field energy, we first have
to regularize it: we pick a sequence of positive self-adjoint compact operators $
\lf\{ \mathds{1}_r \ri\}_{r\in \mathbb{N}} \subset \mathscr{B}(\mathfrak{h})$ approximating the
identity: for all $ r\in
\mathbb{N}$, $ \mathds{1}_r \leq \mathds{1} $, and for all $ z\in \mathfrak{h}_{\omega}$,
\begin{equation}
  \label{eq: omega r}
  \lim_{r\to +\infty}\meanlrlr{z}{\omega_r}{z}_{\mathfrak{h}}=\lim_{r\to +\infty} \meanlrlr{z}{\mathds{1}_r}{z}_{\mathfrak{h}_{\omega}}= \lVert z  \rVert^{2}_{\mathfrak{h}_{\omega}}=\meanlrlr{z}{\omega}{z}_{\mathfrak{h}}\; ,
\end{equation}
where we have denoted $\omega_{r}:=\omega^{\frac{1}{2}} \mathds{1}_r \omega^{\frac{1}{2}}$. Recall also that
$\mathrm{Op}_{\varepsilon}^{\mathrm{Wick}}\bigl(\meanlrlr{z}{\omega}{z}_{\mathfrak{h}}
\bigr) = 1 \otimes \mathrm{d}\mathcal{G}_{\varepsilon}(\omega)$, where
$\mathrm{d}\mathcal{G}_{\varepsilon}(\omega)$ stands for the second quantization of $\omega$ as above.

\begin{lemma}
  \label{lemma:3}
  \mbox{}	\\
  If $\mathcal{K}_0$ has compact resolvent and there exist $ C < +\infty $ such
  that, uniformly w.r.t $\varepsilon\in (0,1)$,
  \begin{equation}
    \lf| \meanlrlr{\Psi_{\eps}}{\lf(\mathcal{K}_0+ \mathrm{d}\mathcal{G}_{\varepsilon}(\omega)+1 \ri)}{\Psi_{\varepsilon}}_{\mathscr{H}_{\varepsilon}}  \ri| \leq C\;,
  \end{equation}
  then, if $\Psi_{\varepsilon_n} \xrightarrow[\eps_n \to 0]{\mathrm{qc}} \mathfrak{m}\in
  \mathscr{W}(\Psi_{\varepsilon},\varepsilon\in (0,1))$, it follows that
  \begin{equation}
    \int_{\mathfrak{h}_{\omega}}^{}\mathrm{d}\mu_{\mathfrak{m}}(z) \: \meanlrlr{z}{\omega}{z}_{\mathfrak{h}}  \leq C\; ,
  \end{equation}
  and, for all $r\in \mathbb{N}$,
  \begin{equation}
    \lim_{n\to + \infty} \meanlrlr{\Psi_{\varepsilon_n}}{1 \otimes \mathrm{d}\mathcal{G}_{\varepsilon_n}(\omega_r)}{\Psi_{\varepsilon_n}}_{\mathscr{H}_{\varepsilon_n}} = \int_{\mathfrak{h}_{\omega}}^{}  \mathrm{d}\mu_{\mathfrak{m}}(z) \: \meanlrlr{z}{\omega_{r}}{z}_{\mathfrak{h}}=\int_{\mathfrak{h}_{\omega}}^{}  \mathrm{d}\mu_{\mathfrak{m}}(z)\, \meanlrlr{z}{\mathds{1}_r}{z}_{\mathfrak{h}_{\omega}}\; .
  \end{equation}
\end{lemma}
\begin{proof}
  The proof of $\mu_{\mathfrak{m}}$-integrability of
  $\meanlrlr{z}{\omega}{z}_{\mathfrak{h}}$ (and the relative bound) is a
  consequence of the corresponding result for semiclassical (scalar) Wigner
  measures proved in \cite{ammari2008ahp,falconi2017ccm}. Analogously, the
  convergence holds because $\meanlrlr{z}{\mathds{1}_r}{z}_{\mathfrak{h}_{\omega}}$ is a
  compact scalar symbol (see \cite{falconi2017ccm} for the convergence of
  compact symbols in $\mathfrak{h}_{\omega}$, and \cite[Props. 2.3 \&
  2.6]{correggi2019arxiv} for additional details on the generalization of
  results in semiclassical analysis to the quasi-classical case).
\end{proof}

\begin{lemma}
  \label{lemma:4}
  \mbox{}	\\
  If $\mathcal{K}_0$ has compact resolvent and there exists $ C < +\infty $, such
  that, uniformly w.r.t $\varepsilon\in (0,1)$,
  \begin{equation}
    \label{eq: apriori 4}
    \lf| \meanlrlr{\Psi_{\eps}}{\lf(\mathcal{K}_0+\mathrm{d}\mathcal{G}_{\varepsilon}(\omega)^2+1 \ri)}{\Psi_{\varepsilon}}_{\mathscr{H}_{\varepsilon}}  \ri| \leq C\;,
  \end{equation}
  then, if $\Psi_{\varepsilon_n} \xrightarrow[\eps_n \to 0]{\mathrm{qc}} \mathfrak{m}$, for
  any $ i = 1, \ldots, N $,
  \begin{equation}
    \label{eq: potential convergence}
    \lim_{n\to +\infty} \meanlrlr{\Psi_{\varepsilon_n}}{\mathrm{Op}_{\varepsilon_n}^{\mathrm{Wick}} \lf(\mathcal{V}_z(\xv_i) \ri)}{\Psi_{\varepsilon_n}}_{\mathscr{H}_{\eps_n}} = \int_{\mathfrak{h}_{\omega}}^{}\mathrm{d}\mu_{\mathfrak{m}}(z) \: \tr_{L^2(\R^{dN})} \lf[\gamma_{\mathfrak{m}}(z) \mathcal{V}_z(\xv_i) \ri]  \; .
  \end{equation}
\end{lemma}

\begin{lemma}
  \label{lemma:5}
  \mbox{}	\\
  There exists a minimizing sequence $ \lf\{ \Psi_{\varepsilon,\delta} \ri\}_{\varepsilon,\delta\in (0,1)}$, such that, for all fixed
  $\delta \in (0,1) $, \eqref{eq: minimizing} holds true and there exists $ C_{\delta} < +\infty$, such that
  \begin{equation}
    \label{eq:4}
    \lf| \meanlrlr{\Psi_{\eps}}{\lf(\mathcal{K}_0+\mathrm{d}\mathcal{G}_{\varepsilon}(\omega)^2+1 \ri)}{\Psi_{\varepsilon}}_{\mathscr{H}_{\varepsilon}}  \ri| \leq C_{\delta}\; .
  \end{equation}
\end{lemma}

The proofs of \cref{lemma:4} and \cref{lemma:5} above, like the form of the quasi-classical
potential $ \mathcal{V}_z $, depend on the model considered. We thus provide them in \cref{sec:concrete-models}.

We are now in a position to prove the lower bound in the trapped case.

\begin{proof}[Proof of \cref{pro: lower bound}]
  Let $\Psi_{\varepsilon,\delta}$ be the minimizing sequence for $H_{\varepsilon}$ of
  \cref{lemma:5}. Since for any $r\in \mathbb{N}$, $\omega_r\leq \omega$, it follows that
  $\mathrm{d}\mathcal{G}_{\varepsilon}(\omega_r)\leq \mathrm{d}\mathcal{G}_{\varepsilon}(\omega)$. Hence,
  \begin{equation}
    \label{eq: minimizing delta}
    \meanlrlr{\Psi_{\varepsilon,\delta}}{\lf(\mathcal{K}_0+ \mathrm{d}\mathcal{G}_{\varepsilon}(\omega_r) + \mathrm{Op}^{\mathrm{Wick}}_{\varepsilon} \lf( \tx\sum_i \mathcal{V}_z(\xv_i) \ri) \ri)}{\Psi_{\varepsilon,\delta}}_{\mathscr{H}_{\varepsilon}} \leq \meanlrlr{\Psi_{\varepsilon,\delta}}{H_{\varepsilon}}{\Psi_{\varepsilon,\delta}}_{\mathscr{H}_{\varepsilon}} < E_{\varepsilon}+\delta\; .
  \end{equation}
  Now, let us recall that, by \crefrange{lemma:2}{lemma:5},
  \begin{itemize}
  \item for any $\delta>0$, $ \mathscr{W}\lf(\Psi_{\varepsilon,\delta},\varepsilon\in (0,1)\ri)\neq\varnothing $;
  \item the expectation value of each term in the Hamiltonian converges as $
    \eps \to 0 $
    or, more precisely, there exists $ \mm \in \mathscr{W} \lf(\Psi_{\varepsilon,\delta},\varepsilon\in (0,1) \ri) $ such that
    \begin{multline}
      		\int_{\mathfrak{h}_{\omega}}^{}\mathrm{d}\mu_{\mathfrak{m}}(z) \: \tr_{L^2(\R^{dN})} \lf[\gamma_{\mathfrak{m}} \lf( \mathcal{K}_0 + \meanlrlr{z}{\omega_r}{z}_{\mathfrak{h}} + \tx\sum_i \mathcal{V}_z(\xv_i) \ri) \ri]  \\
      \leq \liminf_{\varepsilon\to 0} \meanlrlr{\Psi_{\varepsilon,\delta}}{\lf(\mathcal{K}_0+\mathrm{d}\mathcal{G}_{\varepsilon}(\omega_r) + \mathrm{Op}^{\mathrm{Wick}}_{\varepsilon} \lf(\tx\sum_i \mathcal{V}_z(\xv_i) \ri) \ri)}{\Psi_{\varepsilon,\delta}}_{\mathscr{H}_{\varepsilon}}\; .
    \end{multline}
  \end{itemize}
  Hence, we deduce that
  \begin{equation}
    	\int_{\mathfrak{h}_{\omega}}^{}\mathrm{d}\mu_{\mathfrak{m}}(z) \: \tr_{L^2(\R^{dN})} \lf[\gamma_{\mathfrak{m}} \lf( \mathcal{K}_0 + \meanlrlr{z}{\omega_r}{z}_{\mathfrak{h}} + \tx\sum_i \mathcal{V}_z(\xv_i) \ri) \ri] < \liminf_{\varepsilon\to 0} E_{\varepsilon}+\delta\; .
  \end{equation}
	
  Now, by construction,  $ \meanlrlr{z}{\omega_r}{z}_{\mathfrak{h}} \xrightarrow[r \to +
  \infty]{} \meanlrlr{z}{\omega}{z}_{\mathfrak{h}} $ for any $z \in \mathfrak{h}_{\omega} $,
  and, by \cref{lemma:2}, any $\mathfrak{m} \in \mathscr{W} \lf(\Psi_{\varepsilon},\varepsilon\in (0,1)
  \ri)$ is concentrated on $\mathfrak{h}_{\omega}$. Furthermore,
  \begin{equation*}
    \int_{\mathfrak{h}_{\omega}}^{}\mathrm{d}\mu_{\mathfrak{m}}(z) \: \meanlrlr{z}{\omega_r}{z}_{\mathfrak{h}} \leq \int_{\mathfrak{h}_{\omega}}^{} \mathrm{d}\mu_{\mathfrak{m}}(z) \: \meanlrlr{z}{\omega}{z}_{\mathfrak{h}} \leq C < +\infty\; .
  \end{equation*}
  Hence, by dominated convergence,
  \begin{equation}
    \lim_{r\to +\infty} \int_{\mathfrak{h}_{\omega}}^{}\mathrm{d}\mu_{\mathfrak{m}}(z) \: \meanlrlr{z}{\omega_r}{z}_{\mathfrak{h}} = \int_{\mathfrak{h}_{\omega}}^{} \mathrm{d}\mu_{\mathfrak{m}}(z) \: \meanlrlr{z}{\omega}{z}_{\mathfrak{h}} \; .
  \end{equation}
  Thus, one gets
  \begin{equation}
    \inf_{\mathfrak{m} \in \mathscr{W} \lf(\Psi_{\varepsilon,\delta},\varepsilon\in (0,1) \ri)} \int_{\mathfrak{h}_{\omega}}^{}\mathrm{d}\mu_{\mathfrak{m}}(z) \: \tr_{L^2(\R^{dN})} \lf[\gamma_{\mathfrak{m}} \HH_z \ri]  < \liminf_{\varepsilon\to 0} E_{\varepsilon}+\delta\; .
  \end{equation}
  which, via \cref{prop:2}, implies that
  \begin{equation*}
    \eqc \leq \inf_{\mathfrak{m} \in \mathscr{W} \lf(\Psi_{\varepsilon,\delta},\varepsilon\in (0,1) \ri)} \int_{\mathfrak{h}_{\omega}}^{}\mathrm{d}\mu_{\mathfrak{m}}(z) \: \tr_{L^2(\R^{dN})} \lf[\gamma_{\mathfrak{m}} \HH_z \ri] < \liminf_{\varepsilon\to 0} E_{\varepsilon}+\delta\; .
  \end{equation*}
  Since $\delta>0$ is arbitrary, the claim follows.
\end{proof}

\subsubsection{Convergence of minimizing sequences and minimizers}

Once the energy convergence is proven, we investigate the behavior of
minimizing sequences and minimizers, if any. We may thus assume that the
microscopic system admits a ground state $ \psigs $.

\begin{proof}[Proof of \cref{thm:2}]
  Let $\Psi_{\varepsilon,\delta}\in \mathscr{D}(H_{\varepsilon})$ be a minimizing sequence. Then by
  \crefrange{lemma:2}{lemma:5}, any $\mathfrak{m}_\delta \in \mathscr{W}(\Psi_{\varepsilon,\delta},\varepsilon\in
  (0,1))$, corresponding to a sequence $\{\Psi_{\varepsilon_n,\delta}\}_{n\in \mathbb{N}}$, $\varepsilon_n\to 0$,
  satisfies
  \begin{equation*}
    \int_{\mathfrak{h}_{\omega}}^{} \mathrm{d}\mu_{\mathfrak{m}_\delta}(z) \: \tr_{L^2(\R^{dN})} \lf[\gamma_{\mathfrak{m}_\delta}(z) \mathcal{H}_z \ri] = \lim_{n \rightarrow + \infty} \lf\langle \Psi_{\eps_n,\delta} \lf. \ri|H_{\eps_n}  \lf| \ri. \Psi_{\eps_n,\delta} \ri\rangle_{\mathscr{H}_{\eps_n}} <  \lim_{n\to \infty} E_{\varepsilon_n} + \delta =E_{\mathrm{qc}}+\delta\; ,
  \end{equation*}
  as proven in \cref{teo: ground state energy}.
\end{proof}

\begin{proof}[Proof of \cref{cor:2}]
  If $\delta=o_{\varepsilon}(1)$, then considering $\mathfrak{m}_0\in
  \mathscr{W}(\Psi_{\varepsilon,o_\varepsilon(1)},\varepsilon\in (0,1))$, corresponding to a sequence
  $\{\Psi_{\varepsilon_n,\delta}\}_{n\in\mathbb{N}}$, $\varepsilon_n\to 0$, it satisfies
  \begin{equation*}
    \int_{\mathfrak{h}_{\omega}}^{} \mathrm{d}\mu_{\mathfrak{m}_0}(z) \: \tr_{L^2(\R^{dN})} \lf[\gamma_{\mathfrak{m}_0}(z) \mathcal{H}_z \ri] \leq   \lim_{n\to \infty} \lf(E_{\varepsilon_n} + o_{\varepsilon_n}(1) \ri)=E_{\mathrm{qc}}\; .
  \end{equation*}
  By \cref{prop:2} it follows that $ \mathfrak{m}_0 $ is a minimizer of
  \eqref{eq: vpsv2} and, by \cref{prop:3}, it is concentrated on the set $ (\psiqc, \zqc) $ of
  minimizers of \eqref{eq: vp2}.
\end{proof}

\begin{proof}[Proof of \cref{teo: minimizers 1}]
  Let $ \psigs $ be a ground state of $H_{\varepsilon}$. Then, it is also a (exact)
  minimizing sequence with $\delta=0$, and thus as above $ \mathfrak{m}_0 $ is a
  minimizer of \eqref{eq: vpsv2}, and it is concentrated on the set $
  (\psiqc, \zqc) $ of minimizers of \eqref{eq: vp2}.
\end{proof}

\subsection{Non-trapped particle systems}

In the non-trapped case, the strategy of proof is very similar, however it is
not ensured that the set of quasi-classical Wigner measures for the
minimizing sequence is not empty. It is then necessary to use generalized
Wigner measures (recall \cref{def: g wigner}). We note however that the proof
of the upper bound stated in \cref{pro: upper bound} applies to the
non-trapped case too and therefore we have just to provide an alternative
proof of \cref{pro: lower bound}, without the assumption of compactness of
the resolvent of $ \mathcal{K}_0 $.

We first generalize the preparatory lemmas that we needed in the trapped case
to the general situation. Note that for \cref{lemma:5} it is not necessary
that $\mathcal{K}_0$ has compact resolvent and therefore we can use it
directly also in the non-trapped case. We also use the same notation as in
the trapped case; in particular, we make use of the same compact
approximation $\omega_r$ of $\omega$ we introduced in \eqref{eq: omega r}.

\begin{lemma}
  \label{lemma:6}
  \mbox{}	\\
  If there exist $C < +\infty $ such that, uniformly w.r.t. $\varepsilon\in (0,1)$,
  \begin{equation}
    \lf| \meanlrlr{\Psi_{\eps}}{\lf(\mathcal{K}_0 + \mathrm{d}\mathcal{G}_{\varepsilon}(\omega) +  1 \ri)}{\Psi_{\varepsilon}}_{\mathscr{H}_{\varepsilon}}  \ri| \leq C\;,
  \end{equation}
  then $\mathscr{GW}\lf(\Psi_{\varepsilon}, \varepsilon\in (0,1) \ri) \neq \varnothing $. Furthermore, if $ \Psi_{\varepsilon_n}
  \xrightarrow[\eps_n \to 0]{\mathrm{gqc}} \mathfrak{n} $, then $\mathfrak{n}$ is
  in the domain of $\mathcal{K}_0+1$ in the sense of \cref{def: domain g
    wigner} and
  \beq
  \lim_{n\to + \infty}\meanlrlr{\Psi_{\varepsilon_n}}{\mathcal{K}_0}{\Psi_{\varepsilon_n}}_{\mathscr{H}_{\varepsilon_n}} =\int_{\mathfrak{h}_{\omega}}^{} \mathrm{d}\mathfrak{n}(z) \lf[ \mathcal{K}_0 \ri]\;.
  \eeq
\end{lemma}

\begin{lemma}
  \label{lemma:7}
  \mbox{}	\\
  If there exist $ C < +\infty $ such that, uniformly w.r.t $\varepsilon\in (0,1)$,
  \begin{equation}
    \lf| \meanlrlr{\Psi_{\eps}}{\lf(\mathcal{K}_0+ \mathrm{d}\mathcal{G}_{\varepsilon}(\omega)+1 \ri)}{\Psi_{\varepsilon}}_{\mathscr{H}_{\varepsilon}}  \ri| \leq C\;,
  \end{equation}
  then if $\Psi_{\varepsilon_n} \xrightarrow[\eps_n \to 0]{\mathrm{gqc}} \mathfrak{n}\in
  \mathscr{GW}(\Psi_{\varepsilon},\varepsilon\in (0,1))$, it follows that
  \begin{equation}
    \int_{\mathfrak{h}_{\omega}}^{}\mathrm{d} \nn(z)[1] \: \meanlrlr{z}{\omega}{z}_{\mathfrak{h}}  \leq C\; ,
  \end{equation}
  and, for all $r\in \mathbb{N}$,
  \begin{equation}
    \lim_{n\to + \infty} \meanlrlr{\Psi_{\varepsilon_n}}{1 \otimes \mathrm{d}\mathcal{G}_{\varepsilon_n}(\omega_r)}{\Psi_{\varepsilon_n}}_{\mathscr{H}_{\varepsilon_n}} = \int_{\mathfrak{h}_{\omega}}^{}  \mathrm{d}\nn(z)[1] \: \meanlrlr{z}{\omega_r}{z}_{\mathfrak{h}}\; .
  \end{equation}
\end{lemma}

\begin{proof}[Proof of \cref{lemma:6,lemma:7}]
  These lemmas extend to generalized Wigner measures \cref{lemma:2,lemma:3}
  respectively. Their proofs are, \emph{mutatis mutandis}, completely
  analogous to the ones of the latters. Contrarily to \cref{lemma:2}, since
  now $\mathcal{K}_0$ has a non-compact resolvent, the set of Wigner measures
  of $\Psi_{\varepsilon}$ may be empty and there might be a loss of mass along the
  quasi-classical convergence. The set of generalized Wigner measures is,
  however, always non-empty: no mass is lost due to the fact that
  \bdm
  \lf\| \Psi_{\varepsilon} \ri\|_{\mathscr{H}_{\varepsilon}}^{2} = \meanlrlr{\Psi_{\varepsilon}}{1\otimes W_{\varepsilon}(0)}{\Psi_{\varepsilon}}_{\mathscr{H}_{\varepsilon}} = 1,
  \edm
  and the identity operator belongs to $\mathscr{B}(L^2(\R^{dN})) $ but it is
  not compact. More precisely, the above quantity can be immediately
  identified, in the limit $ \eps \to
  0 $, with the total mass of all generalized
  Wigner measures associated to $\Psi_{\varepsilon}$, as defined in \cref{def: g wigner}, whereas it is
  \emph{a priori} only bigger or equal than the total mass of measures
  defined by the convergence in \cref{def: wigner} (if all cluster points for
  the aforementioned convergence have total mass strictly less than one, the
  set of Wigner measures associated to $\Psi_{\varepsilon}$, that are required by \cref{def:
    wigner}   to have total mass one, is thus empty).
\end{proof}

\begin{lemma}
  \label{lemma:8}
  \mbox{}	\\
  If there exists $ C < +\infty $, such that, uniformly w.r.t $\varepsilon\in (0,1)$,
  \begin{equation}
    \lf| \meanlrlr{\Psi_{\eps}}{\lf( \mathcal{K}_0+\mathrm{d}\mathcal{G}_{\varepsilon}(\omega)^2+1 \ri)}{\Psi_{\varepsilon}}_{\mathscr{H}_{\varepsilon}}  \ri| \leq C\;,
  \end{equation}
  then, if $\Psi_{\varepsilon_n} \xrightarrow[\eps_n \to 0]{\mathrm{gqc}} \mathfrak{n}$, for
  any $ i = 1, \ldots, N $,
  \begin{equation}
    \label{eq: g potential convergence}
    \lim_{n\to +\infty} \meanlrlr{\Psi_{\varepsilon_n}}{\mathrm{Op}_{\varepsilon_n}^{\mathrm{Wick}} \lf(\mathcal{V}_z(\xv_i) \ri)}{\Psi_{\varepsilon_n}}_{\mathscr{K}_{\varepsilon_n}} = \int_{\mathfrak{h}_{\omega}}^{}\mathrm{d}\nn(z) \lf[\mathcal{V}_z(\xv_i)\ri]  \; .
  \end{equation}	
\end{lemma}

As for its analogue \cref{lemma:4}, the proof of \cref{lemma:8} is model-dependent and thus given
in \cref{sec:concrete-models}.

The proof of the lower bound for the non-trapped case is now equivalent to
the one in the trapped case, using generalized Wigner measures.

\begin{proof}[Proof of \cref{pro: lower bound}]
  Let $\Psi_{\varepsilon,\delta}$ be the minimizing sequence for $H_{\varepsilon}$ of \cref{lemma:5} satisfying \eqref{eq: minimizing delta}. Now, by
  \crefrange{lemma:5}{lemma:8},
  \begin{itemize}
  \item for any $\delta>0$, $\mathscr{GW}\lf(\Psi_{\varepsilon,\delta},\varepsilon\in (0,1) \ri)\neq\varnothing $;  
  \item as for Wigner measures, there exists $ \mathfrak{n} \in \mathscr{GW} \lf(\Psi_{\varepsilon,\delta},\varepsilon\in (0,1) \ri) $ such that
    \begin{multline}
      \int_{\mathfrak{h}_{\omega}}^{}\mathrm{d}\nn(z) \lf[ \mathcal{K}_0 + \meanlrlr{z}{\omega_r}{z}_{\mathfrak{h}} + \tx\sum_i \mathcal{V}_z(\xv_i) \ri]  \\
      \leq \liminf_{\varepsilon\to 0} \meanlrlr{\Psi_{\varepsilon,\delta}}{\lf(\mathcal{K}_0+ \mathrm{d}\mathcal{G}_{\varepsilon}(\omega_r) + \mathrm{Op}^{\mathrm{Wick}}_{\varepsilon} \lf(\tx\sum_i \mathcal{V}_z(\xv_i) \ri) \ri)}{\Psi_{\varepsilon,\delta}}_{\mathscr{H}_{\varepsilon}}\;,
    \end{multline}
  \end{itemize}
   and therefore
  \begin{equation*}
    \int_{\mathfrak{h}_{\omega}}^{}\mathrm{d}\nn(z) \lf[ \mathcal{K}_0 + \meanlrlr{z}{\omega_r}{z}_{\mathfrak{h}} + \tx\sum_i \mathcal{V}_z(\xv_i) \ri]  < \liminf_{\varepsilon\to 0} E_{\varepsilon}+\delta\; .
  \end{equation*}
  However, by dominated convergence, see \cref{thm:a3} in \cref{sec:gener-state-valu},
  \begin{equation*}
    \lim_{r\to +\infty} \int_{\mathfrak{h}_{\omega}}^{}\mathrm{d}\nn(z)[1] \: \meanlrlr{z}{\omega_r}{z}_{\mathfrak{h}} = \int_{\mathfrak{h}_{\omega}}^{} \mathrm{d}\nn(z)[1] \: \meanlrlr{z}{\omega}{z}_{\mathfrak{h}} \; .
  \end{equation*}
  Hence,
  \begin{equation*}
    \egqc \leq \inf_{\mathfrak{n} \in \mathscr{GW} \lf(\Psi_{\varepsilon,\delta},\varepsilon\in (0,1) \ri)} \int_{\mathfrak{h}_{\omega}}^{}\mathrm{d}\nn(z) \lf[ \HH_z \ri]  < \liminf_{\varepsilon\to 0} E_{\varepsilon}+\delta\;,
  \end{equation*}
  and the result follows from the arbitrarity of $ \delta > 0 $, via \cref{prop:5}.
\end{proof}

The proof of \cref{thm:5} is also completely analogous to the proof of \cref{thm:2} for trapped
systems:

\begin{proof}[Proof of \cref{thm:5}]
  If $\mathcal{K}_0$ does not have compact resolvent, then by
  \crefrange{lemma:6}{lemma:8} and \cref{lemma:5}, any $ \mathfrak{n}_\delta\in
  \mathscr{GW}(\Psi_{\varepsilon,\delta},\varepsilon\in (0,1))$ satisfies
  \begin{equation}
    \int_{\mathfrak{h}_{\omega}}^{} \mathrm{d}\mathfrak{n}_\delta(z) \lf[ \mathcal{H}_z \ri] < \lim_{\varepsilon_n\to 0} E_{\varepsilon_n} +\delta = \eqc +\delta= E_{\mathrm{gqc}}+\delta \; ,
  \end{equation}
  by \cref{teo: ground state energy,prop:5}.
\end{proof}
\begin{proof}[Proof of \cref{cor:3}]
  If $\delta=o_{\varepsilon}(1)$, it follows that $\mathfrak{n}_0\in
  \mathscr{GW}(\Psi_{\varepsilon,o_{\varepsilon}(1)},\varepsilon\in (0,1))$ satisfies
  \begin{equation*}
    \int_{\mathfrak{h}_{\omega}}^{} \mathrm{d}\mathfrak{n}_0(z) \lf[ \mathcal{H}_z \ri] = \lim_{\varepsilon_n\to 0} E_{\varepsilon_n} = E_{\mathrm{gqc}} \; .
  \end{equation*}
  Therefore $\mathfrak{n}_0$ solves \eqref{eq: gvp2}, and thus it is concentrated on minimizers
  solving \eqref{eq: GVP2}.
\end{proof}
\begin{proof}[Proof of \cref{teo: minimizers 2}]
  This proof is completely analogous to the one above.
\end{proof}

\section{Concrete Models}
\label{sec:concrete-models}

In this section we discuss the concrete models introduced in \cref{sec:
  intro}, and in
particular we provide the proof of results used in \cref{sec:ground-states-quasi} that require a
model-dependent treatment.

\subsection{The Nelson model}
\label{sec:nelson-model}

The simplest model under consideration is the so-called Nelson model
\cite{nelson1964jmp}. It consists of a small system of $N$ non-relativistic
particles coupled with a scalar bosonic field, both moving in $d$ spatial
dimensions.

We recall the explicit expression of the quasi-classical energy \eqref{eq:
  classical hamiltonian} in the Nelson model:
\bdm
\mathcal{H}_z = \sum_{j=1}^N \lf\{ -\Delta_j + \mathcal{V}_z(\xv_j) \ri\} + \mathcal{W} \lf(\xv_1, \ldots, \xv_N \ri) + \meanlrlr{z}{\omega}{z}_{\mathfrak{h}} ,
\edm
acting on $L^2(\R^{dN})$ and dependent of $z \in \mathfrak{h}$, where $ \mathcal{V}_z $ is the potential \eqref{eq: potential N}, {\it i.e.}, $ \mathcal{V}_z(\xv) = 2 \Re \braket{z}{\lambda(\xv)}_{\mathfrak{h}} $,
$\mathcal{W} \in L^1_{\mathrm{loc}}(\mathbb{R}^{dN}; \mathbb{R}_+)$ is a field-independent potential\footnote{Of course we may allow for a negative part of the
  potential $ \mathcal{W} $, provided it is bounded, but we choose a positive
  potential for the sake of simplicity.}, {\it e.g.}, a trap or an
interaction between the particles, $\omega\geq 0$ is a self-adjoint operator on
$\mathfrak{h}$ with an inverse that is possibly unbounded and $\lambda, \omega^{-1/2} \lambda
\in L^{\infty}(\mathbb{R}^d,\mathfrak{h})$. Both $\mathcal{W}$
and $\mathcal{V}_z $ are multiplication operators and $\mathcal{H}_z $ is self-adjoint on $
\dom(-\Delta+\mathcal{W})$ and
bounded from below for all $z\in
\mathfrak{h}_{\omega}$. The associated quasi-classical energy of the system is the quadratic form $ \fqc $, whose form domain is thus contained in $ \mathscr{Q}(-\Delta+\mathcal{W}) \oplus
\mathscr{Q}(\omega)$, where 
we recall that $ \mathscr{Q}(A) $ stands for the quadratic form domain associated with the 
self-adjoint operator $ A $.

The quasi-classical Wick quantization of $ \mathcal{H}_z $ yields the quantum field Hamiltonian
\begin{equation*}
  H_{\varepsilon}= \sum_{j=1}^N \lf\{ -\Delta_j \otimes 1 + a_{\varepsilon}\lf(\lambda(\xv_j)\ri) + a_{\varepsilon}^{\dagger}\lf(\lambda(\xv_j)\ri)  \ri\} + \mathcal{W}\lf(\xv_1, \ldots, \xv_N\ri) \otimes 1 + 1 \otimes \mathrm{d}\mathcal{G}_{\varepsilon}(\omega),
\end{equation*}
acting on $\mathscr{H}_{\varepsilon} = L^2(\R^{dN}) \otimes \mathcal{G}_{\varepsilon} (\mathfrak{h}) $, where we have explicitly highlighted the trivial action of some terms of
$H_{\eps} $ on either the particle's or the field's degrees of
freedom. Whenever $ \lambda \in L^{\infty}(\R^d; \mathfrak{h}) $, the operator $H_{\varepsilon}$ is
self-adjoint, with domain of essential self-adjointness 
$ \dom
\lf(-\Delta+\mathcal{W}+\mathrm{d}\mathcal{G}_{\varepsilon}(\omega) \ri) \cap \mathscr{C}_0^{\infty}
\lf(\mathrm{d}\mathcal{G}_{\varepsilon}(1) \ri)$, where the latter is the set of
vectors with finite number of field's excitations \cite{falconi2015mpag}, but
it may be unbounded from below, if $0\in \sigma(\omega)$. It is however well-known that,
if for a.e. $ \xv \in \mathbb{R}^d$, $\lambda(\xv) \in \dom(\omega^{-1/2})$, that we assume in
\eqref{eq: lambda N}, then $H_{\varepsilon}$ is bounded from below by Kato-Rellich's
theorem. Nonetheless, it may still not have a ground state, if $0\in \sigma(\omega)$ or
if $\mathcal{W}$ is not regular enough. We refer to the list of works
\cite{arai2001rmp, betz2002rmp, derezinski2003ahp, pizzo2003ahp,
  georgescu2004cmp, moller2005ahp, hirokawa2006prims, gerard2011cmp,
  abdesselam2012cmp, hiroshima2019arxiv} and references therein for a detailed
discussion of the existence of ground states for the Nelson model. We simply
remark here that the ground state exists, if $0\notin \sigma(\omega)$ and $-\Delta+\mathcal{W}$
has compact resolvent (trapped particle system), or if $0\in \sigma(\omega) $ and $\lambda$ and
$\mathcal{W}$ satisfy suitable conditions, irrespective of compactness of the
resolvent of $-\Delta+\mathcal{W}$.

\begin{proof}[Proof of \cref{prop:7}]
  The upper and lower bounds in \eqref{eq: ul bounds} are well known (see,
  \emph{e.g.}, \cite{ginibre2006ahp,ammari2014jsp,correggi2017ahp}). The
  lower bound is a direct consequence of Kato-Rellich's inequality, while the
  upper bound is proved using coherent states for the field. We provide some
  details for the sake of completeness.
  
  Setting\footnote{Even if not stated explicitly, we use the notation
    $H_{\mathrm{free}}$ also in
    \cref{sec:polaron-model,sec:pauli-fierz-model}, with the same meaning.}
  \begin{equation}
    \hfree : = \mathcal{K}_0 \otimes 1 + 1 \otimes \mathrm{d}\mathcal{G}_{\varepsilon}(\omega)\;,
  \end{equation}
  we get, for all $\alpha>0$ and all $\Psi_{\varepsilon}\in \dom(\hfree)$,
  \begin{multline}
     \label{eq:7}
    \lf\| \sum_{j=1}^N \lf( a_{\varepsilon}\lf( \lambda(\xv_j) \ri) + a_{\varepsilon}^{\dagger}\lf(\lambda(\xv_j) \ri)  \ri) \Psi_{\varepsilon}  \ri\|_{\mathscr{H}_{\varepsilon}}^{} \leq 2N  \lf\| \omega^{-1/2} \lambda \ri\|_{L^{\infty}(\mathbb{R}^d;\mathfrak{h})}^{} \lf\| \mathrm{d}\mathcal{G}_{\varepsilon}(\omega)^{1/2} \Psi_{\varepsilon} \ri\|_{\mathscr{H}_{\varepsilon}}^{}	\\
    + \sqrt{\varepsilon} \lf\| \lambda  \ri\|_{L^{\infty}(\mathbb{R}^d;\mathfrak{h})}^{} \lf\| \Psi_{\varepsilon}  \ri\|_{\mathscr{H}_{\varepsilon}}^{}\\
    \leq \alpha \meanlrlr{\Psi_{\eps}}{\mathrm{d}\mathcal{G}_{\varepsilon}(\omega)}{\Psi_{\varepsilon}}_{\mathscr{H}_{\varepsilon}}^{}+ \lf[\frac{N^2}{\alpha} \lf\| \omega^{-1/2} \lambda \ri\|_{L^{\infty}(\mathbb{R}^d;\mathfrak{h})}^2 + \sqrt{\varepsilon} \lf\| \lambda \ri\|_{L^{\infty}(\mathbb{R}^d;\mathfrak{h})}^{} \ri] \lf\| \Psi_{\varepsilon} \ri\|_{\mathscr{H}_{\varepsilon}}^{}\; .   
  \end{multline}
  Therefore, choosing $ \alpha = 1 $, we deduce that (recall that $ \eps \in (0,1)$)
  \begin{equation}
    E_{\varepsilon}\geq - N^2 \lf\| \omega^{-1/2} \lambda  \ri\|_{L^{\infty}(\mathbb{R}^d;\mathfrak{h})}^2 - \lf\| \lambda  \ri\|_{L^{\infty}(\mathbb{R}^d; \mathfrak{h})}^{}\; .
  \end{equation}
  
  The upper bound is trivial to show by exploiting \eqref{eq:7} and
  evaluating the energy on any state such that $
  \meanlrlr{\Psi_{\eps}}{\mathrm{d}\mathcal{G}_{\varepsilon}(\omega)}{\Psi_{\varepsilon}}_{\mathscr{H}_{\varepsilon}}^{}
  \leq C < +\infty $,
  {\it e.g.}, a product state $ \Psi_{\eps} = \psi \otimes \Omega_{\eps} $, with $
  \psi \in \dom(\mathcal{K}_0) $ and $ \Omega_{\eps} $ the field
  vacuum. Note that the uniform boundedness of $ E_{\eps} $ from above could
  as well be deduced by the boundedness of $ E_0 $, which in turn follows
  from the evaluation of $
  \fqc $ on, {\it e.g.}, a configuration $(\psi,0)$, with
  $ \psi \in
  \dom(\mathcal{K}_0) $.
\end{proof}

We now prove \cref{lemma:4,lemma:5} for the Nelson model. We have however to
state first a technical result, which generalizes the convergence of
expectation values proven in \cite{correggi2019arxiv}: indeed, in
\cite[Prop. 2.6]{correggi2019arxiv} it is shown that\footnote{In \cite[Prop.\
  2.6]{correggi2019arxiv} the result is proved for $\omega=1$. The extension to a
  generic $\omega$ is done straightforwardly combining the proof of Prop.\ 2.6
  with the techniques introduced in \cite{falconi2017ccm}.}, if
\begin{equation*}
  \langle \Psi_{\varepsilon}  , \bigl(\mathrm{d}\mathcal{G}_{\varepsilon}(\omega)+1\bigr)^{\delta}\Psi_{\varepsilon} \rangle_{L^2(\R^{dN})\otimes \mathscr{K}_{\varepsilon}}\leq C\;, 
\end{equation*}
for any $\delta>\frac{1}{2}$, and $\Psi_{\varepsilon_n} \xrightarrow[n \to + \infty]{\mathrm{qc}}
\mathfrak{m}$, then, for all $ \mathcal{K} \in \mathscr{L}^{\infty}(L^2(\R^{dN}))$,
\begin{equation}
  \label{eq:6}
  \lim_{n \to + \infty} \braketr{\Psi_{\varepsilon_n}}{\mathrm{Op}_{\varepsilon_n}^{\mathrm{Wick}} \lf(\mathcal{V}_z  \ri) \mathcal{K} \Psi_{\varepsilon_n}}_{\mathscr{H}_{\varepsilon_n}} = \int_{\mathfrak{h}_{\omega}}^{} \mathrm{d}\mu_{\mathfrak{m}}(z) \: \tr_{L^2(\R^{dN})} \lf[\gamma_{\mathfrak{m}}(z) \mathcal{V}_z \mathcal{K} \ri]  \;,
\end{equation}
but our goal is to apply the above convergence to the identity, which is not
compact. We have then to approximate it with compact operators.

\begin{lemma}
  \label{lemma:9}
  \mbox{}	\\
  Let $\mathcal{K}_0$ have compact resolvent. If there exist $ C < +\infty $ and
  $\delta \geq 1$, such that, uniformly w.r.t $\varepsilon\in (0,1)$,
  \begin{equation}
    \lf| \meanlrlr{\Psi_{\eps}}{\lf(\mathcal{K}_0+ \mathrm{d}\mathcal{G}_{\varepsilon}(\omega)^{\delta} +1 \ri)}{\Psi_{\varepsilon}}_{\mathscr{H}_{\varepsilon}}  \ri| \leq C\;,
  \end{equation}
  and $\Psi_{\varepsilon_n} \xrightarrow[\eps_n \to 0]{\mathrm{qc}} \mathfrak{m}$, then, for
  all $\mathcal{B}\in \mathscr{B}(L^2(\R^{dN}))$ and any $ j = 1, \ldots, N $,
  \begin{equation}
    \lim_{n\to + \infty} \braketr{\Psi_{\varepsilon_n}}{\mathrm{Op}_{\varepsilon_n}^{\mathrm{Wick}}\lf(\mathcal{V}_z(\xv_j) \ri) \mathcal{B} \Psi_{\varepsilon_n}}_{\mathscr{H}_{\varepsilon_n}} = 		\int_{\mathfrak{h}_{\omega}}^{} \mathrm{d}\mu_{\mathfrak{m}}(z) \: \tr_{L^2(\R^{dN})} \lf[\gamma_{\mathfrak{m}}(z)  \mathcal{V}_z(\xv_j) \mathcal{B} \ri] \; .
  \end{equation}
\end{lemma}
 
\begin{proof}
  Let us introduce compact approximate identities $ \lf\{ 1_m \ri\}_{m\in \mathbb{N}} \subset
  \mathscr{L}^{\infty}(L^2(\R^{dN}))$ as
  follows:
  \begin{equation*}
    1_m : = \one_{[-m,m]}(\mathcal{K}_0)\;,
  \end{equation*}
  where $ \one_{[-m,m]}: \mathbb{R}\to \{0,1\}$ is the characteristic function of the
  interval $[-m,m]$, so that the r.h.s. of the above expression is the usual
  spectral projector of $ \mathcal{K}_0 $ constructed via spectral
  theorem. For later convenience, let us also define $ \mathcal{B}_m
  :=\mathcal{B} 1_m\; $. Therefore, we
  have that 
  \begin{multline}
    \braketr{\Psi_{\varepsilon_n}}{\mathrm{Op}_{\varepsilon_n}^{\mathrm{Wick}}\lf(\mathcal{V}_z(\xv_j)\ri) \mathcal{B} \Psi_{\varepsilon_n}}_{\mathscr{H}_{\varepsilon_n}} =\braketr{\Psi_{\varepsilon_n}}{\mathrm{Op}_{\varepsilon_n}^{\mathrm{Wick}}\lf(\mathcal{V}_z(\xv_j)
      \ri) \mathcal{B}_m \Psi_{\varepsilon_n}}_{\mathscr{H}_{\varepsilon_n}} \\
    +\braketr{\Psi_{\varepsilon_n}}{\mathrm{Op}_{\varepsilon_n}^{\mathrm{Wick}}\lf(\mathcal{V}_z(\xv_j)\ri) \lf(\mathcal{B} - \mathcal{B}_m \ri) \Psi_{\varepsilon_n}}_{\mathscr{H}_{\varepsilon_n}} \; .
  \end{multline}
  The first term on the r.h.s. converges when $ n \to + \infty$, for any fixed $m\in
  \mathbb{N}$,
  since $\mathcal{B}_m \in \mathscr{L}^{\infty} (L^2(\R^{dN}))$ (see \eqref{eq:6}),
  {\it i.e.},
  \bdm
  \lim_{n\to +\infty}\braketr{\Psi_{\varepsilon_n}}{\mathrm{Op}_{\varepsilon_n}^{\mathrm{Wick}}\lf(\mathcal{V}_z(\xv_j)\ri) \mathcal{B}_m \Psi_{\varepsilon_n}}_{\mathscr{H}_{\varepsilon_n}} = \int_{\mathfrak{h}_{\omega}}^{}\mathrm{d}\mu_{\mathfrak{m}}(z) \: \tr_{L^2(\R^{dN})} \lf[\gamma_{\mathfrak{m}}(z)\mathcal{V}_z(\xv_j) \mathcal{B}_m \ri] \; .
  \edm
  By dominated convergence, we can then take the limit $m\to + \infty$, to obtain
  \begin{equation}
    \lim_{m \to + \infty} \lim_{n\to +\infty} \braketr{\Psi_{\varepsilon_n}}{\mathrm{Op}_{\varepsilon_n}^{\mathrm{Wick}}\lf(\mathcal{V}_z(\xv_j) \ri) \mathcal{B}_m \Psi_{\varepsilon_n}}_{\mathscr{H}_{\varepsilon_n}} = \int_{\mathfrak{h}_{\omega}}^{} \mathrm{d}\mu_{\mathfrak{m}}(z) \: \tr_{L^2(\R^{dN})} \lf[\gamma_{\mathfrak{m}}(z)  \mathcal{V}_z(\xv_j) \mathcal{B} \ri].
  \end{equation}
  It remains to prove that
  \beq
  \label{eq: lemma9 proof 1}
  \lim_{m\to + \infty} \sup_{\varepsilon\in (0,1)}\lf|\braketr{\Psi_{\varepsilon}}{\mathrm{Op}_{\varepsilon}^{\mathrm{Wick}}\lf(\mathcal{V}_z(\xv_j)\ri) \lf(\mathcal{B} - \mathcal{B}_m \ri) \Psi_{\varepsilon}}_{\mathscr{H}_{\varepsilon}} \ri|=0\;.
  \eeq
  For any $0<s\leq \frac{1}{2}$ and for any $ c_0 > \lf| \inf \sigma \lf(
  \mathcal{K}_0 \ri) \ri| $,
  \begin{multline*}
    \lf| \braketr{\Psi_{\varepsilon}}{\mathrm{Op}_{\varepsilon}^{\mathrm{Wick}}\lf(\mathcal{V}_z(\xv_j) \ri) \lf(\mathcal{B} - \mathcal{B}_m \ri) \Psi_{\varepsilon}}_{\mathscr{H}_{\varepsilon}} \ri|  
    \leq 2 \lf\| \lf(\mathcal{B}-\mathcal{B}_m \ri) \lf(\mathcal{K}_0+c_0 \ri)^{-\frac{s}{2}} \ri\|_{\mathscr{B}(L^2(\R^{dN}))}^{} 	\\
    \times \lf\| \lf( \mathrm{d}\mathcal{G}_{\varepsilon}(\omega)^{\frac{1}{2}}+1 \ri)^{-\frac{1}{2}} a_{\varepsilon}(\lambda(\xv_j)) \lf(\mathrm{d}\mathcal{G}_{\varepsilon}(\omega)^{\frac{1}{2}}+1 \ri)^{-\frac{1}{2}}  \ri\|_{\mathscr{B}(\mathscr{H}_{\varepsilon})}^{} \lf\| \lf( \mathrm{d}\mathcal{G}_{\varepsilon}(\omega)^{\frac{1}{2}}+1 \ri)^{\frac{1}{2}} \lf(\mathcal{K}_0 + c_0 \ri)^{\frac{s}{2}} \Psi_{\varepsilon} \ri\|_{\mathscr{H}_{\varepsilon}}^2  \\
    \leq C \lf\| \mathcal{B} \ri\|_{\mathscr{B}} \lf\| \lf(1-1_m \ri) \lf(\mathcal{K}_0+c_0 \ri)^{-\frac{s}{2}} \ri\|_{\mathscr{B}}^{} \lf\| 
    \omega^{-\frac{1}{2}}\lambda \ri\|_{L^{\infty}(\mathbb{R}^{d},\mathfrak{h})}^{} \meanlrlr{\Psi_{\varepsilon}}{ \mathcal{K}_0^{2s} + \mathrm{d}\mathcal{G}_{\varepsilon}(1) + 1}{\Psi_{\varepsilon}}_{\mathscr{H}_{\varepsilon}}       \\
    \leq C  \sup_{\eta \in \lf[ (-\infty, -m) \cup (m, + \infty)  \ri] \cap \sigma(\mathcal{K}_0)} \frac{1}{\lf( \eta + c_0 \ri)^{\frac{s}{2}}} \leq C m^{-\frac{s}{2}},
  \end{multline*}
  for $ m $ large enough, {\it e.g.}, larger than $ |\inf \sigma \lf(
  \mathcal{K}_0 \ri)| $. Therefore, since the
  above quantity vanishes as $ m
  \to + \infty $ uniformly w.r.t. $ \eps \in (0,1) $, we conclude
  that \eqref{eq: lemma9 proof 1}  holds true and the result follows.
\end{proof}

\begin{proof}[Proof of \cref{lemma:4}]
  The result follows by taking $ \mathcal{B} = 1 $ in \cref{lemma:9}. Again,  this makes crucial use of
  the fact that $\mathcal{K}_0=-\Delta+\mathcal{W}$ has compact resolvent, and  that $\Psi_{\varepsilon}$ is regular
  enough w.r.t. $\mathcal{K}_0$.
\end{proof}

\begin{proof}[Proof of \cref{lemma:5}]
  The proof of \cref{lemma:5} stems from a known result that allows to
  compare the expectation of the square of the free energy $ \hfree^2 $ with
  the expectation of the square of the full Hamiltonian $H_{\varepsilon}^2$. This is a
  consequence of Kato-Rellich's inequality: there exists $
  C>0 $
  (independent of $ \eps $), such that
  \begin{equation}
    \label{eq: kato-rellich}
    \meanlrlr{\Psi_{\varepsilon}}{\hfree^{2}}{\Psi_{\varepsilon}}_{\mathscr{H}_{\varepsilon}} \leq C \meanlrlr{\Psi_{\varepsilon}}{H_{\varepsilon}^2+1}{\Psi_{\varepsilon}}_{\mathscr{H}_{\varepsilon}}\; .
  \end{equation}
  The idea of the proof of this standard inequality goes as follows: from the
  triangular inequality, we get 
  \bdm
  \meanlrlr{\Psi_{\varepsilon}}{\hfree^{2}}{\Psi_{\varepsilon}}_{\mathscr{H}_{\varepsilon}} \leq 2 \meanlrlr{\Psi_{\varepsilon}}{H_{\eps}^{2}}{\Psi_{\varepsilon}}_{\mathscr{H}_{\varepsilon}} + 2\meanlrlr{\Psi_{\varepsilon}}{\lf(H_{\eps} - \hfree \ri)^{2}}{\Psi_{\varepsilon}}_{\mathscr{H}_{\varepsilon}}\;.  
  \edm 
  Now, using inequality \eqref{eq:7}, we get that for any
  $\alpha<1/\sqrt{2}$, 
  \bdm 
  \lf( 1 - 2 \alpha^2 \ri) \meanlrlr{\Psi_{\varepsilon}}{\hfree^{2}}{\Psi_{\varepsilon}}_{\mathscr{H}_{\varepsilon}} \leq 2 \meanlrlr{\Psi_{\varepsilon}}{H_{\eps}^{2}}{\Psi_{\varepsilon}}_{\mathscr{H}_{\varepsilon}} + C_\alpha \lf\| \Psi_{\varepsilon} \ri\|_{L^2(\mathscr{H}_{\varepsilon}}^2\; , 
  \edm
  with $C_\alpha $ independent of $\varepsilon$. The result then easily follows.

  It remains to prove that there exists a minimizing sequence
  $ \lf\{ \Psi_{\varepsilon,\delta} \ri\}_{\varepsilon,\delta\in (0,1)} \subset \dom(H_{\varepsilon})$ for $H_{\varepsilon}$, such that
  \begin{equation}
    \meanlrlr{\Psi_{\varepsilon}}{H_{\varepsilon}^2}{\Psi_{\varepsilon}}_{\mathscr{H}_{\varepsilon}} \leq \max \lf\{ E_{\varepsilon}^2,(E_{\varepsilon}+\delta)^2 \ri\} \leq C ,
  \end{equation}
  with the last inequality given by \cref{prop:7}. Indeed, combining the
  above estimate with \eqref{eq: kato-rellich}, we immediately deduce that
  \eqref{eq: apriori 4} holds true. Let us denote by $ \one_{(a,b)}(H_{\varepsilon})$
  the spectral projections of $H_{\varepsilon}$, and by $\mathscr{P}_{(a,b)}:=
  \one_{(\alpha,\beta)}(H_{\varepsilon}) \mathscr{H}_{\varepsilon}$ the associated
  spectral subspaces. Let now choose, for any $\delta>0$,
  \begin{equation*}
    \Psi_{\varepsilon,\delta}\in \lf\{ \Psi \in \mathscr{P}_{(E_{\varepsilon}-\delta,E_{\varepsilon}+\delta)} \: \big| \: \lf\| \Psi \ri\|_{\mathscr{H}_{\eps}} = 1 \ri\}.
  \end{equation*}
  Each spectral subspace above is not empty by definition of $ E_{\varepsilon}= \inf
  \sigma(H_{\varepsilon})$.
  Therefore, on one hand,
  \begin{equation*}
    \meanlrlr{\Psi_{\varepsilon,\delta}}{H_{\varepsilon}}{\Psi_{\varepsilon,\delta}}_{\mathscr{H}_{\varepsilon}} \leq E_{\varepsilon}+\delta\; ,
  \end{equation*}
  and, on the other,
  \begin{equation*}
    \lf\| H_{\varepsilon} \Psi_{\varepsilon,\delta}  \ri\|_{\mathscr{H}_{\varepsilon}}^2 \leq \max \lf\{ E_{\varepsilon}^2, (E_{\varepsilon}+\delta)^2 \ri\}\; .
  \end{equation*}
\end{proof}

It remains only to prove \cref{lemma:8}, used in the non-trapped case.

\begin{proof}[Proof of \cref{lemma:8}]
  To prove the result, it is sufficient to show that, if $\Psi_{\varepsilon}$ is such that
  \begin{equation*}
    \lf| \meanlrlr{\Psi_{\varepsilon}}{\lf(\mathrm{d}\mathcal{G}_{\varepsilon}(\omega)+1 \ri)^{\delta}}{\Psi_{\varepsilon}}_{\mathscr{H}_{\varepsilon}}  \ri|_{}^{}\leq C \;,  			\end{equation*}
  for some $ \delta \geq 1/2 $ and some finite constant $ C $, and, if $\Psi_{\varepsilon_n}
  \xrightarrow[n \to + \infty]{\mathrm{gqc}} \mathfrak{n}$, then \eqref{eq: g
    potential convergence}
  holds true, {\it i.e.}, for all $\mathcal{B}\in
  \mathscr{B}(L^2(\R^{dN}))$,
  \begin{equation*}
    \lim_{n\to + \infty} \braketr{\Psi_{\varepsilon_n}}{\mathrm{Op}_{\varepsilon_n}^{\mathrm{Wick}} \lf( \mathcal{V}_z) \ri) \mathcal{B} \Psi_{\varepsilon_n}}_{\mathscr{H}_{\varepsilon}} = \int_{\mathfrak{h}_{\omega}}^{}  \mathrm{d}\mathfrak{n} \lf[\mathcal{V}_z \mathcal{B} \ri]\; .
  \end{equation*}
  Such a result is however a special case of \cite[Prop.
  2.6]{correggi2019arxiv}, if in that statement Wigner measures are
  substituted by generalized Wigner measures, the test with compact operators
  of the small system is replaced with the test with bounded operators, and
  $\mathrm{d}\mathcal{G}_{\varepsilon}(1)$ is replaced by
  $\mathrm{d}\mathcal{G}_{\varepsilon}(\omega)$. The proof given there is generalized to
  this setting straightforwardly, recalling the properties of generalized
  Wigner measures outlined in \cref{sec:gener-state-valu}. There is only one
  thing that is worth to remark explicitly: the integration of
  operator-valued functions w.r.t. generalized Wigner measures makes sense
  only if $ \mathrm{Ran} (z\mapsto \mathcal{V}_z) \subset \mathscr{B}(L^2(\R^{dN}))$ is
  separable in the norm topology of $\mathscr{B}(L^2(\R^{dN}))$. Let us check
  explicitly that $\mathrm{Ran} (z\mapsto \mathcal{V}_z)$ is indeed separable:
  since $\mathfrak{h}_{\omega}$ is separable, let us denote by $\mathfrak{k}\subset
  \mathfrak{h}_{\omega}$ a countable dense subset and
 denote by
  \begin{equation*}
    \mathcal{V}_\mathfrak{k} : = \lf\{ \mathcal{V}_{\zeta}(\xv) \in \mathscr{B}(L^2 (\mathds{R}^{dN}   )), \zeta \in \mathfrak{k} \ri\} 
  \end{equation*}
  the image of $\mathfrak{k}$ by means of $z\mapsto\mathcal{V}_z $. Now, for any
  $z\in \mathfrak{h}_{\omega}$, $\zeta\in \mathfrak{k}$, we have that
  \begin{equation*}
    \lf\| \mathcal{V}_z - \mathcal{V}_\zeta \ri\|_{\mathscr{B}(L^2(\R^{dN}))}^{} \leq 2 \lf\| \omega^{-\frac{1}{2}}\lambda  \ri\|_{L^{\infty}(\mathbb{R}^d;\mathfrak{h})}^{} \lf\| z-\zeta  \ri\|_{\mathfrak{h}_{\omega}}^{}\;,
  \end{equation*}
  which implies that $\mathcal{V}_\mathfrak{k}$ is dense in $\mathrm {Ran}(z\mapsto
  \mathcal{V}_z)$ w.r.t.\ the $\mathscr{B}(L^2(\R^{dN}))$-norm
  topology.
\end{proof}

\subsection{The polaron model}
\label{sec:polaron-model}

The polaron model, introduced in \cite{frohlich1937prslA}, describes $N$
electrons (spinless for simplicity) subjected to the vibrational (phonon)
field of a lattice. This model is similar to Nelson's, however the coupling
is slightly more singular. The one-excitation space is $ \mathfrak{h} =
L^2(\mathbb{R}^d) $, while the
form factor is given by \eqref{eq: lambda polaron}: the quasi-classical energy has the same form as 
in the Nelson model, as well as the effective potential $ \mathcal{V}_z $ (see \eqref{eq: potential N}),
although now 
\bdm 
\lambda(\xv; \kv) = \sqrt{\alpha} \frac{e^{-i \kv\cdot \xv}}{\lvert \kv \rvert_{}^{\frac{d-1}{2}}}, \qquad \omega = 1\,, 
\edm
where $\alpha>0$ is a constant measuring the coupling's strength. The assumptions
on $ \mathcal{K}_0 = -\Delta + \mathcal{W} $ are the same as in the Nelson
model. Let us remark that in this case since $\omega=1$,
$\mathfrak{h}_{\omega}=\mathfrak{h}$.

The key difference with the aforementioned Nelson model is thus that $ \exists z \in
\mathfrak{h} $
such that
\begin{equation*}
  \mathcal{V}_z(\: \cdot \:) \notin L^{\infty}(\mathbb{R}^d)\; ,
\end{equation*}
due to the fact that $ \lambda \notin L^{\infty}(\R^d; \mathfrak{h}) $. However, it is
possible to write $\mathcal{V}_z$ as the sum of an $L^{\infty}$ function and the
commutator between an $L^{\infty}$ vector function and the momentum operator
$-i\nabla_{\xv} $:
\begin{equation}
  \label{eq: potential splitting}
  \mathcal{V}_z(\xv)= \sqrt{\alpha} \lf( {\mathcal{V}}_{<,z}(\xv)+ \lf[-i\nabla_{\xv}\,,\,\bm{\mathcal{V}}_{>,z}(\xv) \ri] \ri)\;,
\end{equation}
where 
\beqn
{\mathcal{V}}_{<,z}(\xv) &=& 2\Re \mathscr{F}^{-1} \lf[ \lambda_< z \ri](\xv)\;,	\qquad		\lambda_{<}(\kv): = \one_{\lf| \kv   \ri|_{}^{}\leq \varrho} \, \lf| \kv  \ri|_{}^{-\frac{d-1}{2}}, \nonumber\\
\bm{\mathcal{V}}_{>,z}(\xv) &=& 2\Re \mathscr{F}^{-1} \lf[ \bm{\lambda}_> z\ri](\xv)\;, \qquad \bm{\lambda}_>(\kv) : = \one_{\lf| \kv \ri|_{}^{} > \varrho}\, \lf| \kv\ri|_{}^{-\frac{d+1}{2}} \hat{\kv}, \nonumber 
\eeqn
where $\hat{\kv} :=\frac{\kv}{\lvert \kv \rvert_{}^{}} $ and $ \mathscr{F} $ stands for
the Fourier transform in $ \R^d $. Note that, for any $ \varrho > 0 $, $ \lambda_< \in
\mathfrak{h} $ and $ \bm{\lambda}_> \in \mathfrak{h} \otimes \mathbb{C}^d $. 
By KMLN theorem, it then follows that $\mathcal{H}_z$ is self-adjoint and
bounded from below for all $z\in \mathfrak{h}$, with $z$-independent form domain
$ \mathscr{Q}(\mathcal{H}_z) = \mathscr{Q}(\mathcal{K}_0)$. Let us remark that, choosing
$\rho$ suitably large (independent of $ z $) in the above decomposition, it is
possible to make the operator $ \HH_z $ bounded from below uniformly
w.r.t. $z\in \mathfrak{h}$ (see, \emph{e.g.},
\cite[Prop. 3.21]{correggi2017ahp}).

The quasi-classical Wick quantization of $\mathcal{H}_z $ formally yields the
same expression as in the Nelson model (with $ \omega = 1 $ and $ \lambda $ as
above). Such a formal operator gives rise to a closed and bounded from below
quadratic form, via the decomposition \eqref{eq: potential splitting} (this
can also be proved by KLMN theorem, choosing $\varrho$ sufficiently large (see,
\emph{e.g.}, \cite{lieb1997cmp,frank2014lmp})). We still denote the
corresponding self-adjoint operator by $ H_{\eps} $ with a little abuse of
notation. The polaron Hamiltonian $H_{\varepsilon}$ has a ground state, if $-\Delta+
\mathcal{W}$ has
compact resolvent by an application of the HVZ theorem
analogous to the one for the Nelson model (see the aforementioned result in
\cite{derezinski1999rmp}). It is known that ground states exist also for
non-confining but suitably regular external potentials $\mathcal{W}$.

\begin{proof}[Proof of \cref{prop:7}]
  These lower and upper bounds are well-known (see,
  \emph{e.g.},\cite{lieb1997cmp,correggi2017ahp}). The lower bound is a
  direct consequence of KLMN theorem, while the upper bound is proved using
  coherent states for the field in a fashion that is completely analogous to
  the one discussed for the Nelson model. Thus here we focus on the lower
  bound.
  
  Let us introduce the unperturbed operator $ \hfree = \mathcal{K}_0 \otimes 1 + 1
  \otimes \mathrm{d}\mathcal{G}_{\varepsilon}(1) $, as in the
  Nelson model. Then, for any $\Psi_{\varepsilon} \in \mathscr{Q}(\hfree)$, for all $\varrho>0$, and for all $\beta>0$, we
  can bound the interaction term in the polaron quadratic form via 
  \begin{multline}
        \lf| \meanlrlr{\Psi_{\varepsilon}}{\mathrm{Op}_{\varepsilon}^{\mathrm{Wick}} \lf(\mathcal{V}_{<,z}(\xv) \ri) - i\nabla_{\xv} \cdot \mathrm{Op}_{\varepsilon}^{\mathrm{Wick}} \lf(\bm{\mathcal{V}}_{>,z}(\xv) \ri) + i\mathrm{Op}_{\varepsilon}^{\mathrm{Wick}} \lf( \bm{\mathcal{V}}_{>,z}(\xv) \ri) \cdot \nabla_{\xv}}{\Psi_{\varepsilon} }_{\mathscr{H}_{\varepsilon}} \ri| \\
    \leq 2 \lf\| \lambda_< \ri\|_{\mathfrak{h}}^{} \meanlrlr{\Psi_{\varepsilon}}{\hfree^{1/2}}{\Psi_{\varepsilon}}_{\mathscr{H}_{\varepsilon}} + 4 \lf\| \bm{\lambda}_> \ri\|_{\mathfrak{h}}^{}\meanlrlr{\Psi_{\varepsilon}}{\hfree}{\Psi_{\varepsilon}}_{\mathscr{H}_{\varepsilon}} \\
    \leq \tx\frac{1}{\beta} \lf\| \lambda_< \ri\|^2_{\mathfrak{h}} \lf\| \Psi_{\varepsilon}\ri\|^2_{\mathscr{H}_{\varepsilon}} + \lf( \beta + 4 \lf\| \bm{\lambda}_>\ri\|_{\mathfrak{h}}^{} \ri)\meanlrlr{\Psi_{\varepsilon}}{\hfree}{\Psi_{\varepsilon}}_{\mathscr{H}_{\varepsilon}}.
  \end{multline}
  Obviously, the norms of $ \lambda_< $ and $ \bm{\lambda}_> $ depend on $ \varrho $. However,
  since the norm of $ \bm{\lambda}_> $ diverges as $ \varrho \to 0 $ and vanishes as $ \varrho \to
  + \infty $, we can always choose $\varrho= \varrho(\beta)$, such that
  \beq
  4 \lf\| \bm{\lambda}_> \ri\|_{\mathfrak{h}}^{} = \beta\;.
  \eeq
  Hence, we can bound
  \bdm
  \lf|\meanlrlr{\Psi_{\eps}}{H_I}{\Psi_{\varepsilon}}_{\mathscr{H}_{\eps}} \ri| \leq \sqrt{\alpha} N \lf[2 \beta \meanlrlr{\Psi_{\varepsilon}}{\hfree}{\Psi_{\varepsilon}}_{\mathscr{H}_{\varepsilon}} + \tx\frac{1}{\beta}\lf\| \lambda_< \ri\|^2_{\mathfrak{h}} \lf\| \Psi_{\varepsilon} \ri\|^2_{\mathscr{H}_{\varepsilon}}\ri],
  \edm
  so that, taking $ \beta = (2 \sqrt{\alpha} N)^{-1} $, we conclude that
  \begin{equation}
    E_{\varepsilon} \geq - 2 \alpha  N^2 \lf\| \lambda_< \ri\|^2_{\mathfrak{h}},
  \end{equation}
  where the last norm is evaluated at $ \varrho((2 \sqrt{\alpha} N)^{-1}) $.
\end{proof}

Let us now prove \cref{lemma:4,lemma:5}. The assumption in the former takes
the following simplified form for the polaron model: assuming that there
exists a finite constant $C $, such that \beq
\label{eq: assumption lemma4}
\lf| \meanlrlr{\Psi_{\eps}}{\lf(\mathcal{K}_0+ \mathrm{d}\mathcal{G}_{\varepsilon}(1)^2 +
  1 \ri)}{\Psi_{\varepsilon}}_{\mathscr{H}_{\varepsilon}} \ri| \leq C\;, \eeq then the convergence
\eqref{eq: potential convergence} holds true for any limit point in $
\mathscr{W}\lf(\Psi_{\eps}, \eps \in (0,1) \ri) $.

\begin{proof}[Proof of \cref{lemma:4}]
  Using again the splitting \eqref{eq: potential splitting}, we immediately
  see that the term involving the quantization of $ \mathcal{V}_{<,z} $
  converges by \cref{lemma:9}. Let us consider then the other
  term. Analogously to the proof of \cref{lemma:9}, we define compact
  approximate identities $ \lf\{ 1_m \ri\}_{m\in \mathbb{N}} \subset
  \mathscr{L}^{\infty}(L^2(\R^{dN}))$ as $ {1}_m := \one_{[-m,m]}(\mathcal{K}_0)
  $.
  
  We can now rewrite explicitly the term involving the quantization of $
  \bm{\mathcal{V}}_{>,z} $,
  by introducing $ \bm{\xi} \in L^{\infty}(\R^d;
  \mathfrak{h}) $ given by
  \beq
  \label{eq: xi}
  \bm{\xi}(\xv; \kv) : = \bm{\lambda}_> e^{- i \kv \cdot \xv},
  \eeq
  as
  \begin{multline}
    \sqrt{\alpha} \sum_{j=1}^N \meanlrlr{\Psi_{\varepsilon_n}}{\lf[-i\nabla_j\,,\,\mathrm{Op}_{\varepsilon_n}^{\mathrm{(Wick)}} \lf( \bm{\mathcal{V}}_{>}(\xv_j) \ri) \ri]}{\Psi_{\varepsilon_n}}_{\mathscr{H}_{\varepsilon_n}} \\
    = 2\sqrt{\alpha} \sum_{j=1}^N \Re \braketr{-i\nabla_{j} \Psi_{\varepsilon_n}}{\lf[ a^{\dagger}_{\varepsilon_n}\lf(\bm{\xi}(\xv_j) \ri) + a_{\varepsilon_n} \lf(\bm{\xi}(\xv_j) \ri) \ri]\Psi_{\varepsilon_n}}_{\mathscr{H}_{\varepsilon_n}}\;.
  \end{multline}
  In order to prove its convergence, we estimate
  \begin{multline}
    \lf| \braketr{-i\nabla_{j} \Psi_{\varepsilon_n}}{\lf[ a^{\dagger}_{\varepsilon_n}\lf( \bm{\xi}(\xv_j) \ri) + a_{\varepsilon_n} \lf(\bm{\xi}(\xv_j) \ri) \ri] \Psi_{\varepsilon_n}}_{\mathscr{H}_{\varepsilon_n}} \ri|	\\
    \leq \lf| \braketr{-i\nabla_{j} \Psi_{\varepsilon_n}}{\lf[ a^{\dagger}_{\varepsilon_n}\lf( \bm{\xi}(\xv_j) \ri)+ a_{\varepsilon_n} \lf(\bm{\xi}(\xv_j) \ri) \ri] 1_m \Psi_{\varepsilon_n}}_{\mathscr{H}_{\varepsilon_n}}\ri|\\
    + \lf| \braketr{-i\nabla_{j} \Psi_{\varepsilon_n}}{\lf[ a^{\dagger}_{\varepsilon_n}\lf( \bm{\xi}(\xv_j) \ri)+ a_{\varepsilon_n} \lf(\bm{\xi}(\xv_j) \ri) \ri] \lf(1 - 1_m\ri)\Psi_{\varepsilon_n}}_{\mathscr{H}_{\varepsilon_n}} \ri|
  \end{multline}
      
  The first term on the r.h.s.\ converges, when $n\to +\infty$ and $m\in \mathbb{N}$ is fixed,
  thanks to \cite[Prop.\ 7.1]{correggi2019arxiv}; then, a dominated
  convergence argument allows to take the limit $m\to + \infty$, yielding the sought
  result. It remains therefore to prove that the second term on the
  r.h.s. converges to zero as $m\to + \infty$, uniformly w.r.t. $\varepsilon\in (0,1)$. This is
  done as follows:
  \begin{multline}
        \lf| \braketr{-i\nabla_{j} \Psi_{\varepsilon_n}}{\lf[ a^{\dagger}_{\varepsilon_n}\lf( \bm{\xi}(\xv_j) \ri) + a_{\varepsilon_n} \lf(\bm{\xi}(\xv_j) \ri) \ri] \lf(1 - 1_m\ri) \Psi_{\varepsilon_n}}_{\mathscr{H}_{\varepsilon_n}} \ri| \\
    \leq \lf\| \nabla_{j}\Psi_{\varepsilon}  \ri\|_{\mathscr{H}_{\varepsilon}}^{} \lf\| \lf[ a^{\dagger}_{\varepsilon_n}\lf( \bm{\xi}(\xv_j) \ri) + a_{\varepsilon_n} \lf(\bm{\xi}(\xv_j) \ri) \ri] \lf(1 - 1_m\ri) \Psi_{\varepsilon}  \ri\|_{\mathscr{H}_{\varepsilon}}\\
    \leq 2 \lf(\varepsilon + \lf\| \bm{\xi} \ri\|_{L^{\infty}(\mathbb{R}^d;\mathfrak{h})}^{} \ri)\lf\|\nabla_{j}\Psi_{\varepsilon} \ri\|_{\mathscr{H}_{\varepsilon}}^{} \lf\|\lf(\mathrm{d}\mathcal{G}_{\varepsilon}(1)+1 \ri)^{1/2} \lf(1 - 1_m\ri) \Psi_{\varepsilon}\ri\|_{\mathscr{H}_{\varepsilon}} \;.  
  \end{multline}
  Thus, for all $\beta > 0$, $\varepsilon\in (0,1)$ and $s>0$ and for any $ c_0 >
  |\inf\sigma(\mathcal{K}_0) | $: 
  \begin{multline}
    \lf| \braketr{-i\nabla_{j} \Psi_{\varepsilon_n}}{\lf[ a^{\dagger}_{\varepsilon_n}\lf( \bm{\xi}(\xv_j) \ri) + a_{\varepsilon_n} \lf(\bm{\xi}(\xv_j) \ri) \ri] \lf(1 - 1_m\ri) \Psi_{\varepsilon_n}}_{\mathscr{H}_{\varepsilon_n}} \ri|	\leq  \lf(1 + \lf\| \bm{\xi} \ri\|_{L^{\infty}}^{} \ri) \lf[ \beta \lf\| \mathcal{K}_0^{1/2} \Psi_{\varepsilon}  \ri\|_{\mathscr{H}_{\varepsilon}}^2 \ri.	\\
    \lf. + \tfrac{1}{\beta} \lf\| \lf({1}-{1}_m\ri) \lf(\mathcal{K}_0+c_0 \ri)^{-\frac{s}{2}} \ri\|_{\mathscr{B}(L^2(\R^{dN}))}^2 \lf\| \lf(\mathrm{d}\mathcal{G}_{\varepsilon}(1)+1 \ri)^{1/2} \lf(\mathcal{K}_0 + c_0 \ri)^{\frac{s}{2}} \Psi_{\varepsilon} \ri\|_{\mathscr{H}_{\varepsilon}}^2 \ri] \\
    \leq 2 \lf(1 + \lf\| \bm{\xi} \ri\|_{L^{\infty}}^{} \ri) \lf(\beta +\tfrac{1}{\beta} \lf\|\lf({1}-{1}_m\ri) \lf(\mathcal{K}_0+c_0 \ri)^{-\frac{s}{2}}\ri\|_{\mathscr{B}}^2 \ri) \meanlrlr{\Psi_{\varepsilon}}{\mathcal{K}_0+\mathcal{K}_0^{2s}+\mathrm{d}\mathcal{G}_{\varepsilon}(1)^2+1}{\Psi_{\varepsilon}}_{\mathscr{H}_{\varepsilon}}\;    .  
  \end{multline}
  Hence, using \eqref{eq: assumption lemma4}, for any $s\leq \frac{1}{2}$, we get 
  \beq
  \label{eq: proof lemma4 1}
  \lf| \braketr{-i\nabla_{j} \Psi_{\varepsilon_n}}{\lf[ a^{\dagger}_{\varepsilon_n}\lf( \bm{\xi}(\xv_j) \ri) +a_{\varepsilon_n} \lf(\bm{\xi}(\xv_j) \ri) \ri] \lf(1 - 1_m\ri)\Psi_{\varepsilon_n}}_{\mathscr{H}_{\varepsilon_n}} \ri| \leq C \beta_m\;,
  \eeq
  where we have chosen
  \bdm
  \beta = \beta_m : = \lf\| \lf({1}-{1}_m\ri) \lf(\mathcal{K}_0+c_0\ri)^{-\frac{s}{2}} \ri\|_{\mathscr{B}} = \sup_{\eta \in \lf[ (-\infty, -m) \cup (m, +\infty) \ri] \cap \sigma(\mathcal{K}_0)} \frac{1}{\lf( \eta + c_0 \ri)^{\frac{s}{2}}}\xrightarrow[m \to + \infty]{} 0.
  \edm
  Since the r.h.s. of \eqref{eq: proof lemma4 1} is independent of $\varepsilon$ and
  converges to zero as $m\to+ \infty$, the result is proven.
\end{proof}

\begin{proof}[Proof of \cref{lemma:5}]
  The proof is analogous to the one for the Nelson model. The expectation of
  the number operator squared is bounded via the \emph{pull-through formula}
  by means of the expectation of $H_{\varepsilon}^2$. As discussed in
  \cite{correggi2019arxiv}, the pull-through formula was originally proved
  for the renormalized Nelson Hamiltonian with a bound that is $\varepsilon$-dependent
  in \cite{ammari2000mpag}; the uniformity of such bound with respect to $\varepsilon\in
  (0,1)$ has been proved in \cite{ammari2017sima}. Since the renormalized
  Nelson model ``contains'' all type of terms appearing in the polaron model,
  the proof of the formula extends to the polaron model immediately (see
  \cite{olivieri2020sap} for additional details).

  The pull-through formula reads as follows: there exists a finite constant
  $C $ (independent of $ \eps $), such that
  \begin{equation}
    \label{eq: pull-through}
    \meanlrlr{\Psi_{\varepsilon}}{\mathrm{d}\mathcal{G}_{\varepsilon}(1)^{2}}{\Psi_{\varepsilon}}_{\mathscr{H}_{\varepsilon}} \leq C \meanlrlr{\Psi_{\varepsilon}}{(H_{\varepsilon} +1)^2}{\Psi_{\varepsilon}}_{\mathscr{H}_{\varepsilon}}\; .
  \end{equation}

  The expectation of $ \hfree $ is bounded by means of the expectation of
  $H_{\varepsilon}$, using the KMLN inequality, already discussed in the proof of
  \cref{prop:7}, in the very same way we used Kato-Rellich inequality for the
  Nelson model. The fact that there exists a minimizing sequence such that
  the expectation of $H_{\varepsilon}^2$ is bounded uniformly w.r.t. $\varepsilon\in (0,1)$ is also
  discussed in the proof for the Nelson model and it does not depend on the
  model at hand. We omit further details for the sake of brevity.
\end{proof}

It remains only to prove \cref{lemma:8} for non-trapping potentials.

\begin{proof}[Proof of \cref{lemma:8}]
  The proof is done using the following fact: if $\Psi_{\varepsilon}$ is such that there
  exists $\delta \geq1$ and a finite constant $C $, such that
  \begin{equation}
    \lf| \meanlrlr{\Psi_{\eps}}{\lf(\mathcal{K}_0+ \mathrm{d}\mathcal{G}_{\varepsilon}(1)^\delta + 1 \ri)}{\Psi_{\varepsilon}}_{\mathscr{H}_{\varepsilon}}  \ri| \leq C\;,
  \end{equation}
  then, if $\mathfrak{n}\in \mathscr{GW} \lf(\Psi_{\varepsilon},\varepsilon\in (0,1) \ri)$ and $ \Psi_{\varepsilon_n}
  \xrightarrow[n \to + \infty]{\mathrm{gqc}} \mathfrak{n}$, one has that \eqref{eq:
    g potential convergence}
  holds true.
  	
  Such a result is proved by a combination of \cite[Props.\ 2.6 \&
  7.1]{correggi2019arxiv}, if in these propositions Wigner measures are
  substituted by generalized Wigner measures and the test with compact
  operators of the small system is substituted by the test with the identity
  operator. The proof given there is generalized to this setting
  straightforwardly, recalling the properties of generalized Wigner measures
  outlined in \cref{sec:gener-state-valu}.
	
  As in the proof for the Nelson model, let us check explicitly that
  $\mathrm{Ran}(z\mapsto \mathcal{V}_z)$ is separable in the norm operator
  topology\footnote{More precisely, we prove that $\mathrm{Ran} (z\mapsto
    (\mathcal{K}_0+1)^{-\frac{1}{2}}
    \mathcal{V}_z(\mathcal{K}_0+1)^{-\frac{1}{2}}) $ has separable
    range. This is sufficient to prove that $\mathcal{V}_{( \, \cdot \, )}$ is
    integrable w.r.t. $\mathfrak{n}$, since the latter is in the domain of
    $\mathcal{K}_0+1$.}. By using the decomposition \eqref{eq: potential
    splitting}, we see that the part containing
  $ \mathcal{V}_{<,z} $ has separable range, since it is equivalent to the one appearing in the
  Nelson model. Let us focus then on the remaining one containing the
  expectation of the operator $ \lf[-i\nabla_x\,,\,\bm{\mathcal{V}}_{>,z} \ri]$. Such an operator
  is not bounded. Nonetheless, it is $\mathfrak{n}_{\mathcal{T}}$-integrable
  with $\mathcal{T}=\mathcal{K}_0 +1$ by \cref{lemma:7}, provided that
  \begin{equation}
    \label{eq: lemma8 proof map}
    \mathfrak{h}\ni z\mapsto \sum_{j=1}^N \mathcal{T}^{-\frac{1}{2}} \lf[-i\nabla_j\,,\,\bm{\mathcal{V}}_{>,z}(\xv_j) \ri] \mathcal{T}^{-\frac{1}{2}}\in \mathscr{B}(L^2(\R^{dN}))
  \end{equation}
  has separable range. Since $\mathfrak{h}$ is separable, let us denote by
  $\mathfrak{k}\subset \mathfrak{h}$ a countable dense subset and denote by
  \begin{equation*}
    \mathcal{T}^{-\frac{1}{2}} \widetilde{\mathcal{V}}_{\mathfrak{k}}\mathcal{T}^{-\frac{1}{2}} : = \lf\{ \tx\sum_j \mathcal{T}^{-\frac{1}{2}} \lf[-i\nabla_{j}\,,\,\bm{\mathcal{V}}_{>,\zeta(\xv_j)} \ri] \mathcal{T}^{-\frac{1}{2}} \in \mathscr{B}(L^2 (\mathds{R}^{dN}   )), \zeta\in \mathfrak{k} \ri\} 
  \end{equation*}
  the image of $\mathfrak{k}$ through $\mathcal{T}^{-\frac{1}{2}} \sum_j
  \lf[-i\nabla_{j}\,,\,\bm{\mathcal{V}}_{>,\, \cdot \,}(\xv_j) \ri]
  \mathcal{T}^{-\frac{1}{2}}$. Now, for any $z\in \mathfrak{h}$, $\zeta\in
  \mathfrak{k}$ and
  $ j = 1, \ldots, N $, we have that (recall \eqref{eq: xi})
  \begin{equation*}
    \lf\|  \mathcal{T}^{-\frac{1}{2}} \lf[-i\nabla_{j}\,,\,\bm{\mathcal{V}}_{>,z}(\xv_j) \ri] \mathcal{T}^{-\frac{1}{2}} - \mathcal{T}^{-\frac{1}{2}} \lf[-i\nabla_{j}\,,\,\bm{\mathcal{V}}_{>,\zeta}(\xv_j) \ri]  \mathcal{T}^{-\frac{1}{2}}  \ri\|_{\mathscr{B}(L^2(\R^{dN}))}^{} \leq 4 \lf\| \bm{\xi} \ri\|_{L^{\infty}}^{} \lf\|  z-\zeta  \ri\|_{\mathfrak{h}}^{}\; ,
  \end{equation*}
  which implies that $\mathcal{T}^{-\frac{1}{2}}
  \widetilde{\mathcal{V}}_{\mathfrak{k}} \mathcal{T}^{-\frac{1}{2}}$ is dense
  in the image of the map \eqref{eq: lemma8 proof map} w.r.t. the norm
  topology in $\mathscr{B}(L^2(\R^{dN}))$.
\end{proof}

\subsection{The Pauli-Fierz model}
\label{sec:pauli-fierz-model}

The Pauli-Fierz model describes $N$ spinless charges (with an extended and
sufficiently smooth charge distribution) interacting with the electromagnetic
field in the Coulomb gauge, in three dimensions. Generalizations to other
gauges, to particles with spin or to two dimensions are possible without much
effort. The one-excitation Hilbert space is thus $\mathfrak{h} =
L^2(\mathbb{R}^3;\mathbb{C}^{2}) $. Let the charge
density of each particle be given by $
\lambda_j(\xv) $, with $ \lambda_j \in L^{\infty}\bigl(\mathbb{R}^3;L^2(\mathbb{R}^3)\bigr)$, $ j = 1, \ldots, N $,
such that $-i\nabla_j \lambda_j(\xv; \kv) = \kv \lambda_j(\xv; \kv)$ and let the polarization
vectors be denoted $\vec{e}_{p}\in L^{\infty}(\mathbb{R}^3;\mathbb{R}^3)$, $p =1,2$, such that for
a.e.\ $ \kv \in \mathbb{R}^3$, $\vec{e}_{p}(\kv)\cdot \vec{e}_{p'}(\kv) = \delta_{pp'}$, $ \kv \cdot
\vec{e}_{p}(\kv)=0$ (Coulomb gauge). The quasi-classical 
energy functional is then given by \eqref{eq: classical hamiltonian PF}, 
{\it i.e.}\footnote{W.l.o.g. we fix the charge $e=1$ since it does not play any relevant role in these arguments.}, 
\bdm
\mathcal{H}_z = \sum_{j = 1}^N \frac{1}{2m_{j}}\Bigl( - i \nabla_j + \av_{z,j}(\xv_j)\Bigr)^2 + \mathcal{W}(\xxv) + \meanlrlr{z}{\omega}{z}_{\mathfrak{h}} 
\edm
where the classical field is 
\bdm 
\av_{z,j}(\xv) = 2 \Re \braket{z}{\bm{\lambda}_j(\xv)}_{\mathfrak{h}} = 2 \Re \sum_{p=1}^2 \braket{z_p}{\lambda_j(\xv)\vec{e}_{p}}_{L^2(\R^3)} \in \mathbb{C}^3 
\edm
and, as usual, $ \mathcal{W} $ is an external positive potential acting on
the particles. Note that the field free energy reads
\begin{equation*}
  \meanlrlr{z}{\omega}{z}_{\mathfrak{h}} = \sum_{p=1}^2 \meanlrlr{z_p}{\omega}{z_p}_{L^2(\R^3)}.
\end{equation*}
The operator $\mathcal{H}_z $ is self-adjoint for all $z\in \mathfrak{h}_{\omega}$,
with domain of self-adjointness $ \dom(\mathcal{K}_0)$, where we recall that
$ \mathcal{K}_0 = -\Delta + \mathcal{W} $, where in this case we adopt the
notation
\begin{equation*}
  -\Delta=\sum_{j=1}^N -\frac{\Delta_{j}}{2m_j}\; .
\end{equation*}

The quasi-classical Wick quantization of $\mathcal{H}_z$ yields the
Pauli-Fierz Hamiltonian in \eqref{eq: hamiltonian PF}: 
\bdm
\label{eq: hamiltonian PF quant}
H_{\eps} = \sum_{j=1}^N{\frac{1}{2m_{j}}\bigl(-i\nabla_j + \aav_{\varepsilon,j}(\xv_j)\bigr)}^2+ \mathcal{W}(\xv_1,\dotsc, \xv_N) + 1 \otimes \diff \GG_{\eps}(\omega), 
\edm 
where
\begin{equation*}
  \aav_{\eps, j}(\xv) = a^{\dagger}_{\varepsilon}(\bm{\lambda}_j(\xv)) + a_{\varepsilon}(\bm{\lambda}_j(\xv)) = \sum_{p=1}^{2} \lf( a^{\dagger}_{\varepsilon, p}(\lambda_j(\xv) \vec{e}_p) +  a_{\varepsilon, p}(\lambda_j(\xv)\vec{e}_p) \ri) 
\end{equation*}
is the quantized magnetic potential. The Pauli-Fierz Hamiltonian is
self-adjoint on $ \dom(\mathcal{K}_0+\mathrm{d}\mathcal{G}_{\varepsilon}(\omega))$, provided
that for almost all $\xv \in \mathbb{R}^3$, and for all $j=1,\dotsc,N$, $ \lambda_{j}(\xv) \in
\mathscr{Q}(\omega + \omega^{-1})$ (see
\cite{hiroshima2000cmp, hiroshima2002ahp, hasler2008rmp, falconi2015mpag,  
  matte2017mpag}), that we assumed in \eqref{eq: lambda PF}. The Pauli-Fierz Hamiltonian has a
ground state for suitable choices of the potential $\mathcal{W}$,
\emph{e.g.}, if it is the sum of single particle and pair potentials with
suitable properties (clustering, binding, etc.) (see, \emph{e.g.},
\cite{arai1999jfa,gerard2000ahp,griesemer2001im,hiroshima2001tams} and
references therein). In particular, this holds true when the field is massive
\cite{gerard2000ahp}, {\it i.e.}, for $ \omega > 0 $. As for the other models, we
refrain from giving a detailed description of the conditions allowing to have
a ground state, since for our purposes it is sufficient that a ground state
do exist in some cases.

\begin{proof}[Proof of \cref{prop:7}]
  The lower bound follows from the diamagnetic inequality
  \cite{matte2017mpag}:
  \begin{equation}
    \label{eq:9}
    \meanlrlr{\Psi_{\varepsilon}}{-\Delta_j}{\Psi_{\eps}}_{\mathscr{H}_{\varepsilon}} \leq   \meanlrlr{\Psi_{\eps}}{\lf( -i\nabla_j + \aav_{\varepsilon,j}(\xv_j) \ri)^2}{ \Psi_{\varepsilon}}_{\mathscr{H}_{\varepsilon}}\;,
  \end{equation}
  which in particular implies that $ H_{\eps} $ is positive. The upper bound
  is proved using coherent states for the field, analogously to the Nelson
  model and the polaron.
\end{proof}

Let us now prove \cref{lemma:4,lemma:5} for the Pauli-Fierz model. The former
takes the following form.

\begin{proof}[Proof of \cref{lemma:4}]
  The ``potential'' \eqref{eq: potential PF} is composed of two parts:
  \bdm
  \mathcal{V}_z(\xv) = 2 \sum_{j = 1}^N \frac{1}{m_j} \lf[ - i \Re\braket{z}{\bm{\lambda}_j(\xv)}_{\mathfrak{h}} \cdot \nabla_j + \lf( \Re\braket{z}{\bm{\lambda}_j(\xv)}_{\mathfrak{h}} \ri)^2 \ri]
  \edm
  as well as its Wick quantization. The convergence of the quantization of
  the second term is perfectly analogous to the one given for the Nelson
  model in \cref{lemma:9}. The proof of convergence for the quantization of
  the term involving the gradient is given in the proof of \cref{lemma:4} for
  the polaron.
\end{proof}

\begin{proof}[Proof of \cref{lemma:5}]
  The proof follows from the following estimate, due to F.\ Hiroshima, and
  whose detailed proof will be given in \cite{ammari2019prep}. There exists a finite constant
  $C>0$ such that, for all $\Psi_{\varepsilon}\in \mathscr{D}(H_{\mathrm{free}})$,
  \begin{equation}
    \label{eq:2}
    \lVert H_{\mathrm{free}}\Psi_{\varepsilon}  \rVert_{\mathscr{H}_{\varepsilon}}^{}\leq C \lVert H_{\varepsilon}\Psi_{\varepsilon}  \rVert_{\mathscr{H}_{\varepsilon}}^{}\; .
  \end{equation}
  Let us remark that the expectation of $ \mathcal{K}_0 = - \sum_j \frac{1}{2
    m_j} \Delta_j + \mathcal{W} $ could also be
  bounded by means of the expectation of $H_{\varepsilon}$ using the diamagnetic
  inequality \eqref{eq:9}. Hence if $\omega>0$, \eqref{eq:4} could be proved
  combining the diamagnetic inequality and the pull-through formula
  \eqref{eq: pull-through}.

  Finally, the fact that there exists a minimizing sequence such that the
  expectation of $H_{\varepsilon}^2$ is bounded uniformly w.r.t. $\varepsilon\in (0,1)$ is also
  discussed in the proof of \cref{lemma:5} for the Nelson model.
\end{proof}

It remains only to prove \cref{lemma:8} for non-trapped systems.

\begin{proof}[Proof of \cref{lemma:8}]
  The proof here is obtained combining the proofs given for the Nelson and
  polaron models. In fact, the quadratic terms can be treated exactly as the
  linear terms in the Nelson model and the gradient terms are equivalent to
  the ones appearing in the polaron.
\end{proof}

\appendix

\section{Algebraic State-Valued Measures}
\label{sec:gener-state-valu}

The quasi-classical Wigner measures are state-valued by construction
\cite{falconi2017arxiv,correggi2019arxiv}. In other words, quasi-classical
measures are countably additive (in a sense to be clarified below) measures
on the measurable phase space of classical fields, taking values in quantum
states, or, more generally, in the Banach cone $\mathfrak{A}'_+$ of positive
elements in the dual of a C*-algebra $\mathfrak{A}$. In addition, the
quasi-classical symbols are measurable functions from the phase space to a
W*-algebra $\mathfrak{B}\supseteq \mathfrak{A}$ of observables (operators), where
$\mathfrak{A}$ is supposed to be an ideal of $\mathfrak{B}$. It is therefore
necessary to properly define integration of operator-valued symbols w.r.t.\ a
state-valued measure. In this appendix we collect some technical properties
of state-valued measures and integration, from a general algebraic standpoint
that includes both state-valued and generalized state-valued measures, as
used throughout the paper. The ideas developed here in great generality are
particularly suited for what we called generalized state-valued measures, and
they are mostly taken from \cite{bartle1956sm} and \cite{neeb1998mm}. In
fact, if state-valued measures have been already studied in semiclassical
analysis and adiabatic theories (see \cite{ba,fg,pg,gms,teu} and references
therein contained), the reader might not be so familiar with generalized
state-valued measures. Since for the latter there is no Radon-Nikod\'{y}m
property, their description is more abstract, and there are some limitations,
especially concerning integration of operator-valued functions. This
justifies the abstract approach followed in this appendix.

\subsection{Algebraic State-Valued Measures}
\label{sec:state-valu-meas-1}

Let $\mathfrak{A}$ be a C*-algebra and denote by $\mathfrak{A}'_+$ the cone
of positive elements in the dual of $\mathfrak{A}$. In addition, let $(X,\Sigma)$
be a measurable space. There are two \emph{equivalent} ways of defining an
$\mathfrak{A}'_+$-valued measure on $(X,\Sigma)$.

\begin{definition}[State-valued measure \cite{neeb1998mm}]
  \label{def:a1}
  \mbox{}	\\
  A family of real-valued measures ${(\mu_A)}_{A\in \mathfrak{A}_+}$ defines a
  weak-$*$ $\sigma$-additive measure $\mm: \Sigma\to \mathfrak{A}'_+$ as
  \begin{equation*}
    \lf[ \mm(S) \ri] (A_1-A_2+iA_3-iA_4)=\mu_{A_1}(S)-\mu_{A_2}(S)+i\mu_{A_3}(S)-i\mu_{A_4}(S)\;,
  \end{equation*}
  for any $S\in \Sigma$ and $A_1,A_2,A_3,A_4\in \mathfrak{A}_+$, iff for any $A,B\in
  \mathfrak{A}_+$ and $\lambda\in \mathbb{R}_+$,
  $\mu_{\lambda A+B}=\lambda\mu_A+\mu_B$.
\end{definition}

\begin{definition}[Algebraic state-valued measure \cite{bartle1956sm}]
  \label{def:a2}
  \mbox{}	\\
  An application $\mm:\Sigma\to \mathfrak{A}'_+$ is a measure iff $\mm(\varnothing )=0$, and
  for any family ${(S_n)}_{n\in \mathbb{N}}\subset \Sigma$ of mutually disjoint measurable sets,
  \begin{equation*}
    \mm \lf(\bigcup_{n\in \mathbb{N}}S_n \ri)=\sum_{n\in \mathbb{N}}^{}\mm(S_n)\; ,
  \end{equation*}
  where the r.h.s. converges unconditionally in the norm of $\mathfrak{A}'$.
\end{definition}
It is clear that any $\mm$ satisfying \cref{def:a2} satisfies also
\cref{def:a1}, since $\sigma$-additivity in norm implies weak-$*$
$\sigma$-additivity. The converse, \emph{i.e.}, that a $\mm$ satisfying
\cref{def:a1} also satisfies \cref{def:a2} is nontrivial, and follows from
properties of uniform boundedness in Banach spaces, as proved by
\cite[Chapter II]{dunford1938tams}. We use these two definitions
interchangeably, depending on the context. Let us remark that with the
definitions above, any state-valued measure is automatically finite, since
$\mm(X)\in \mathfrak{A}'_+$. Actually, in the main body of the paper, we
consider probability measures, \emph{i.e.}, $\lVert \mm(X)
\rVert_{\mathfrak{A}'}^{}=1$.

\begin{remark}[State-valued and generalized state-valued measures]
  \label{rem:1}
  \mbox{} \\
  The state-valued measures used in the paper correspond to choosing
  $\mathfrak{A}=\mathscr{L}^{\infty}(\mathscr{H})$; generalized state-valued
  measures are in a subset of the measures obtained by picking $\mathfrak{A}=\mathscr{L}^1(\mathscr{H})$.
\end{remark}

For algebraic state-valued (cylindrical) measures on vector spaces, Bochner's
theorem holds, and the Fourier transforms are completely positive maps that
are weak-* continuous when restricted to any finite-dimensional subspace (see
\cite{falconi2017arxiv} for additional details). An algebraic state-valued
measure is also \emph{monotone}:

\begin{lemma}
  \label{lemma:a1}
  \mbox{}	\\
  For any $S_1\subseteq S_2\in \Sigma$,
  \begin{equation*}
    \mm(S_1)\leq \mm(S_2)\;,
  \end{equation*}
  {\it i.e.}, $\mm(S_2)-\mm(S_1)\in \mathfrak{A}'_+$.
\end{lemma}

\begin{proof}
  The scalar measures $\mu_A$, $A\in \mathfrak{A}_+$, are monotonic. Therefore,
  for all $A\in \mathfrak{A}_+$,
  \begin{equation}
    \lf[ \mm(S_2) \ri](A):= \mu_A(S_2)\geq \mu_A(S_1)=:\lf[ \mm(S_1) \ri](A)\; .
  \end{equation}
  Hence, for all $A\in \mathfrak{A}_+$,
  \begin{equation*}
    \bigl[ \mm(S_2)-\mm(S_1)\bigr](A)\geq 0\; .
  \end{equation*}
\end{proof}
	
We can now introduce the scalar norm measure $m$, satisfying $\mu_A(S)\leq \lVert A
\rVert_{\mathfrak{A}}^{}m(S)$,
for any $S\in \Sigma$, that proves to be a very useful
tool to compare vector integrals with scalar integrals.

\begin{definition}[Norm measure]
  \label{def:a7}
  \mbox{}	\\
  Let $\mm$ be an algebraic state-valued measure. Then, its norm measure
  $m:\Sigma\to \mathbb{R}_+$ is defined as
  \begin{equation}
    m(S):= \lf\lVert \mm(S)  \ri\rVert_{\mathfrak{A}'}^{}\; ,
  \end{equation}
  for any measurable set $ S $.
\end{definition}
Using the cone properties of positive states in a C*-algebra, it is possible
to prove that $m$ is a finite measure. Let us recall that the C*-algebra
$\mathfrak{A}$ may not be unital, so from now on we assume that there exists
a W*-algebra $\mathfrak{B}\supseteq \mathfrak{A}$. If
$\mathfrak{A}=\mathscr{L}^{\infty}(\mathscr{K})$, the compact operators on a
separable Hilbert space $\mathscr{K}$, and
$\mathfrak{B}=\mathscr{B}(\mathscr{K})$, it is well-known that the
aforementioned property is satisfied: $\mathfrak{A}$ is actually in this case
a two-sided ideal of $\mathfrak{B}$. Let us denote by $ e \in \mathfrak{B}$ the
identity element.

\begin{proposition}[Properties of the norm measure]
  \label{prop:a3}
  \mbox{}	\\
  Let $\mm$ be an algebraic state-valued measure. Then, its norm measure $m$
  is a finite measure on $(X,\Sigma) $ and $\mm \ll m$.
\end{proposition}
	
\begin{proof}
  The proof that $m(\varnothing )=0$ and $m(X)<+\infty$ follows immediately from the
  definition, while $\sigma$-additivity is proved as follows: let ${(S_n)}_{n\in \mathbb{N}}\subset
  \Sigma$ be a
  family of mutually disjoint measurable sets, we are going to prove
  that, for any $N\in \mathbb{N}$,
  \begin{equation}
    m \lf(\bigcup_{n=1}^N S_n \ri)=\sum_{n=1}^N m(S_n)\;  
  \end{equation}
  Indeed, let ${(e_{\alpha})}_{\alpha\in I}\subset \mathfrak{A}_+$ be an approximate identity
  of $ e \in \mathfrak{B}$. It is well-known that for any $\omega\in \mathfrak{A}'_+$,
  $\lVert \omega \rVert_{\mathfrak{A}'}^{}=\lim_{\alpha\in I}\omega(w_{\alpha})$. Hence, by \cref{def:a1} and
  \cref{def:a7}:
  \begin{equation*}
    m \lf(\bigcup_{n=1}^N S_n \ri)=\lim_{\alpha\in I} \mm \lf(\bigcup_{n=1}^N S_n \ri)(e_{\alpha})=\lim_{\alpha\in I} \mu_{e_{\alpha}} \lf(\bigcup_{n=1}^N S_n \ri)=\lim_{\alpha\in I}\sum_{n=1}^N\mu_{e_{\alpha}}(S_n)=\sum_{n=1}^N m(S_n)\; .
  \end{equation*}
  Next, we show
  \begin{equation}
    \lim_{N\to \infty} m \lf(\bigcup_{n\in \mathbb{N}} S_n \ri)-\sum_{n=1}^N m(S_n)=0\;,
  \end{equation}
  which directly implies $\sigma$-additivity: using again the approximate identity
  on the left hand side, we obtain
  \begin{equation*}
    \lim_{N\to \infty}\lim_{\alpha\in I} m \lf(\bigcup_{n\in \mathbb{N}} S_n \ri)-\sum_{n=1}^N\mu_{e_{\alpha}}(S_n)\; .
  \end{equation*}
  We know that every $\mu_{e_{\alpha}}$, $\alpha\in I$, is $\sigma$-additive, and therefore that
  $\lim_{N\to \infty}\sum_{n=1}^N\mu_{e_{\alpha}}(S_n)=\mu_{e_{\alpha}} \lf(\bigcup_{n\in \mathbb{N}} S_n \ri)$, and
  $\lim_{\alpha\in I}\mu_{e_{\alpha}} \lf(\bigcup_{n\in \mathbb{N}} S_n \ri)=m \lf(\bigcup_{n\in \mathbb{N}} S_n
  \ri)$. Hence, it remains to show that the limits in $N$
  and $\alpha$ can be exchanged. In order to do that, it is sufficient to show
  that the limit in $\alpha$ exists uniformly w.r.t.\ $N$:
  \begin{multline}
    \lim_{\alpha\in I}\sup_{N\in \mathbb{N}}\, \lf\lvert m\lf(\bigcup_{n=1}^N S_n \ri)-\sum_{n=1}^N\mu_{e_{\alpha}}(S_n) \ri\rvert_{}^{} =\lim_{\alpha\in I}\sup_{N\in \mathbb{N}} \lf\lvert m\lf(\bigcup_{n=1}^N S_n \ri) - \mu_{e_{\alpha}} \lf(\bigcup_{n=1}^NS_n \ri) \ri\rvert\\
=\lim_{\alpha\in I}\sup_{N\in \mathbb{N}} \lf(m-\mu_{e_{\alpha}} \ri) \lf(\bigcup_{n=1}^NS_n \ri)\leq \lim_{\alpha\in I} \lf( m-\mu_{e_{\alpha}} \ri)(X)=0 \; ,
  \end{multline}
  where we have used finite additivity of $m$ and $m-\mu_{e_{\alpha}}$ and the fact
  that for any $S\in \Sigma$, $\mu_{e_{\alpha}}(S)\leq m(S)$.
  
  It remains to prove that $\mm$ is absolutely continuous w.r.t.\ $m$. For
  absolute continuity of a vector measure with respect to a scalar one, we
  adopt the definition of \cite[Section I.2, Definition
  3]{diestel1977ms}. Since both $\mm$ and $m$ are countably additive, it is
  sufficient to prove that, for any $S\in \Sigma$, $m(S)=0$ implies
  $\mm(S)=0$. However, since $m(S)=\lVert \mm(S) \rVert_{\mathfrak{A}'}^{}$, and $\lVert \,\cdot
  \, \rVert_{\mathfrak{A}'}^{}$ is a norm, then the
  aforementioned implication follows directly by the properties of norms.
\end{proof}

\subsection{Integration of Scalar Functions}
\label{sec:integr-scal-funct}

The theory of integration for algebraic state-valued measures could be done
in a unified way for scalar- and operator-valued functions. However, it is
instructive to deal with scalar functions first. Let us recall that a
function $g:X\to \mathbb{R}^+$ is simple if there exist a number $N\in \mathbb{N}$, mutually
disjoint measurable sets $S_1,\dotsc,S_N\in \Sigma$ and non-negative numbers $c_1,\dotsc,c_N\in
\mathbb{R}^+$, such that for all $x\in X$,
\begin{equation}
  g(x)=\sum_{j=1}^N c_j\one_{S_j}(x)\; ,
\end{equation}
where $\one_{S_j}$ is the characteristic function of the set
$S_j$. Integration of simple functions w.r.t.\ an algebraic state-valued
measure $\mu$ is straightforwardly defined as
\begin{equation}
  {\int_{X}^{}} \mathrm{d}\mm(x) \: g(x)=\sum_{j=1}^N c_j\mm(S_j)\in \mathfrak{A}'_+\; . 
\end{equation}
The integral of a non-simple function can be defined again in two equivalent
ways:

\begin{definition}[Integrability I \mbox{\cite[Lemma I.12]{neeb1998mm}}]
  \label{def:a3}
  \mbox{}	\\
  A measurable function $f:X\to \mathbb{R}^+$ is $\mm$-integrable iff $f$ is
  $\mu_A$-integrable for any $A\in \mathfrak{A}_+$. Furthermore, its integral
  belongs to $\mathfrak{A}'_+$ and is uniquely defined by the integral
  w.r.t. $ \mu_A $, {\it i.e.}, 
  \begin{multline}
        \lf( \int_S^{}  \mathrm{d}\mm(x) \: f(x) \ri)(A_1-A_2+iA_3-iA_4) = \int_S^{}  \mathrm{d}\mu_{A_1}(x)\: f(x) - \int_S^{}  \mathrm{d}\mu_{A_2}(x) \: f(x) \\
    + i\int_S^{} \mathrm{d}\mu_{A_3}(x) \: f(x) - i\int_S^{} \mathrm{d}\mu_{A_4}(x) \:f(x)\; .
  \end{multline}
  for any $A_1,A_2,A_3,A_4\in \mathfrak{A}_+$.
\end{definition}
	
\begin{definition}[Integrability II \mbox{\cite[Definition 1]{bartle1956sm}}]
  \label{def:a4}
  \mbox{}	\\
  A measurable function $f:X\to \mathbb{R}^+$ is $\mm$-integrable iff for any $S\in \Sigma$ the
  sequence of simple integrals
  \begin{equation*}
    \lf\{{\int_{X}^{}}\mathrm{d}\mm(x) \: f_n(x) \one_S(x) \ri\}_{n\in \mathbb{N}} \in  \mathfrak{A}'\; ,
  \end{equation*}
  where ${(f_n)}_{n\in \mathbb{N}}$ is any approximation of $f$ in terms of simple
  functions, is Cauchy. The integral is then defined as
  \begin{equation}
    \int_S^{}  \mathrm{d}\mm(x)f(x)=\lim_{n\to \infty} \int_{X}^{}\mathrm{d}\mm(x) \: f_n(x) \one_S(x)  \; ,
  \end{equation}
  and it is independent of the chosen approximation.
\end{definition}

In both cases one says that a complex function $f:X\to \mathbb{C}$ is $\mu$-integrable if
and only if $\lvert f \rvert_{}^{}$ is $\mm$-integrable and, in this case, its integral
is given by the complex combination of the integrals of its real positive,
real negative, imaginary positive and imaginary negative parts.

Since the weak-$*$ and strong limits coincide if they both exist, it follows
that the integrals of a function that is $\mm$-integrable
w.r.t. \cref{def:a3} and \cref{def:a4} coincide. In addition, if $f$ is
$\mm$-integrable in the ``strong'' sense of \cref{def:a4}, then it is also
$\mm$-integrable in the weak-$*$ sense of \cref{def:a3}. It remains to show
that if $f$ is $\mm$-integrable in the sense of \cref{def:a3}, then it is
$\mm$-integrable in the sense of \cref{def:a4}, but this can be done
exploiting the norm measure $m$.

\begin{lemma}
  \label{lemma:a3}
  \mbox{}	\\
  If a measurable function $f : X \to \R^+ $ is $\mm$-integrable in the sense
  of \cref{def:a3}, then it is $m$-integrable as well.
\end{lemma}
\begin{proof}
  If $f$ is $\mm$-integrable, then for any $S\in \Sigma$, $\int_{S}^{}
  \mathrm{d}\mu_A(x)f(x)$ is finite and
  non-negative for any $A\in
  \mathfrak{A}_+$. Applying \cite[Lemma I.5]{neeb1998mm}, we
  deduce that there exists a finite constant $C$, depending only on $S$,
  $\mm$, and $f$, such that
  \begin{equation}
    \int_{S}^{} \mathrm{d}\mu_A(x) \: f(x)\leq C \lVert A  \rVert_{\mathfrak{A}}^{}\; .
  \end{equation}
  Now, let ${(f_n)}_{n\in \mathbb{N}}$ be a simple pointwise non-decreasing
  approximation of $f$ from below. Then, by monotone convergence theorem,
  \begin{equation*}
    \int_S^{}  \mathrm{d}m(x) \: f(x)=\lim_{n\to \infty} \int_{X}^{}  \mathrm{d}m(x) \: f_n(x) \one_S(x) \; .
  \end{equation*}
  Hence, by \cref{def:a7}, and $\mu_{e_{\alpha}}$-integrability of $f$,
  \begin{multline}
        \int_{X}^{}  \mathrm{d}m(x) \: f_n(x) \one_S(x)=\lim_{\alpha\in I} \int_{X}^{}  \mathrm{d}\mu_{e_{\alpha}}(x) \: f_n(x) \one_S(x) \leq \lim_{\alpha\in I}\int_S^{}  \mathrm{d}\mu_{e_{\alpha}}(x) \: f(x)\\
    \leq C \lim_{\alpha\in I} \lVert e_{\alpha} \rVert_{\mathfrak{A}}^{} \leq C\;,
  \end{multline}
  and taking the limit $ n \to + \infty $, we get the result.
\end{proof}	
	
\begin{proposition}[Equivalence of \cref{def:a3} and \cref{def:a4}]
  \label{prop:a4}
  \mbox{}	\\
  If a measurable function $f : X \to \R^+ $ is $\mm$-integrable in the sense
  of \cref{def:a3}, then it is $\mm$-integrable in the sense of
  \cref{def:a4}. In addition, for any $S\in \Sigma$,
  \beq
  \label{eq: bound norm measure}
  \lf\lVert \int_S^{} \mathrm{d}\mm(x) \: f(x) \ri\rVert_{\mathfrak{A}'}^{}\leq \int_S^{}\mathrm{d}m(x) \: f(x)\; .
  \eeq
\end{proposition}
\begin{proof}
  We prove that
  \bdm \lf\{ \int_S^{} \mathrm{d}\mm(x)f_n(x) \ri\}_{n\in \mathbb{N}} \in \mathfrak{A}'_+,
  \edm
  where ${(f_n)}_{n\in \mathbb{N}}$ is a non-decreasing simple approximation of $f$, is
  a Cauchy sequence. Observe that for any $n\geq m\in \mathbb{N}$, $f_n-f_m$ is a simple
  positive function, which can be written as
  \begin{equation}
    f_n-f_m=\sum_{j=1}^{N(n,m)} c_j^{(n,m)} \one_{S_j^{(n,m)}}\; .
  \end{equation}
  Hence,
  \begin{multline}
    \lf\| \int_S^{}  \mathrm{d}\mm(x)\bigl(f_n(x)-f_m(x)\bigr) \ri\|_{\mathfrak{A}'}^{}\leq \sum_{j=1}^{N(n,m)} c_j^{(n,m)} m\lf(S_j^{(n,m)}\cap S\ri) \\
    =\int_S^{} \mathrm{d}m(x) \lf(f_n-f_m \ri)(x) \xrightarrow[n, m \to \infty]{} 0\; ,
  \end{multline}
  where in the last limit we have used the dominated convergence theorem,
  since $f_n-f_m\leq 2f$, and $f$ is $m$-integrable by \cref{lemma:a3}. This
  proves both $\mm$-integrability of $f$ in the sense of \cref{def:a4}, and
  the bound \eqref{eq: bound norm measure}.
\end{proof}

Therefore, the two definitions are indeed equivalent: \cref{def:a4} has the
advantage of identifying constructively the integral as the limit of the
integrals of simple approximations of the integrand, while \cref{def:a3} is
useful to prove properties of the integral. The integral defined above is
indeed linear in the integrand and monotonic:

\begin{lemma}
  \label{lemma:a2}
  \mbox{}	\\
  Let $f,g:X\to \mathbb{R}$ be two $\mm$-integrable functions. If for $\mm$-a.e.\ $x\in X$
  $ g(x)\leq f(x) $, then
  \begin{equation}
    \int_X^{}  \mathrm{d}\mm(x) \lf(f(x) - g(x) \ri) \in \mathfrak{A}'_+ .
  \end{equation}
\end{lemma}	
\begin{proof}
  The result follows from \cref{def:a3},and monotonicity of the usual
  integral
\end{proof}

The dominated convergence theorem holds in a general form (see \cref{thm:a1}
and \cref{thm:a3} below), which in particular implies that it applies to
scalar functions.

\subsection{Integration of Operator-Valued Functions}
\label{sec:integr-oper-valu}

The integration of operator-valued functions is defined similarly to
\cref{def:a4}. Let us discuss first the integration of simple operator-valued
functions and the approximation with simple functions in this context. An
operator valued function $g:X\to \mathfrak{B}$ is simple if there exist $N\in \mathbb{N}$,
mutually disjoint measurable sets $S_1,\dotsc,S_N\in \Sigma$, and $c_1,\dotsc,c_N\in
\mathfrak{B}$ such that for all $x\in X$,
\begin{equation}
  g(x)=\sum_{j=1}^N c_j \one_{S_j}(x)\; .
\end{equation}
Let us recall that since $\mathfrak{A}\subset \mathfrak{B}$, for any $\omega\in
\mathfrak{A}'$ and $B\in \mathfrak{B}$, we can define $\omega \circ B \in \mathfrak{A}'$
as
\begin{equation}
  \lf( \omega \circ B \ri) (\,\cdot \,):= \omega(\, \cdot \, B) \qquad \text{ or } \qquad \lf( \omega \circ B \ri)(\,\cdot \,):=\omega(B\,\cdot \,)\; ,
\end{equation}
depending on which side $\mathfrak{A}$ is an ideal of $\mathfrak{B}$. If it
is a two-sided ideal, both definitions are equivalent. Keeping this
definition in mind, we can define the integral of simple functions as
\begin{equation}
  \int_{X} \mathrm{d}\mm(x) \: g(x) = \sum_{j=1}^N \mm(S_j) \circ c_j \in \mathfrak{A}'\; .
\end{equation}

Next, we recall hypotheses under which an operator-valued function admits a
simple approximation.

\begin{proposition}[Simple approximation \mbox{\cite[Proposition
    E.2]{cohn2013mt}}]
  \label{prop:a1}
  \mbox{}	\\
  Let $f:X\to \mathfrak{B}$ be a measurable function. If $f(X)$ is separable,
  then $f$ admits a simple approximation, {\it i.e.}, there exists a sequence
  $\lf\{f_n \ri\}_{n\in \mathbb{N}}$ of simple functions such that for all $x\in X$ and
  $n\in \mathbb{N}$
  \beq
  \lf\lVert f_n(x) \ri\rVert_{\mathfrak{B}}^{}\leq \lf\lVert f(x)\ri\rVert_{\mathfrak{B}}^{}\;, \qquad \lim_{n\to \infty} \lf\lVert f(x)-f_n(x)\ri\rVert_{\mathfrak{B}}^{} =0\; .
  \eeq
\end{proposition}	
	
Due to this result, in the following we only consider operator-valued
functions with \emph{ separable range}, even if not stated explicitly.

\begin{definition}[Integrability III]
  \label{def:a5}
  \mbox{}	\\
  A measurable function with separable range $f:X\to \mathfrak{B}$ is
  $\mm$-integrable iff, for any $S\in \Sigma$, the sequence of simple integrals
  \begin{equation}
    \lf\{ \int_{X} \mathrm{d}\mm(x) \: f_n(x) \one_S(x) \ri\}_{n\in \mathbb{N}} \in \mathfrak{A}'\; ,
  \end{equation}
  where $ \lf\{ f_n \ri\}_{n\in \mathbb{N}}$ is any approximation of $f$ in terms simple
  functions, is Cauchy. The integral is then defined as
  \begin{equation}
    \int_S^{}  \mathrm{d}\mm(x)f(x)=\lim_{n\to \infty}  \int_{X} \mathrm{d}\mm(x) \: f_n(x) \one_S(x) \; ,
  \end{equation}
  and it is independent of the chosen approximation.
\end{definition}

\begin{definition}[Absolute integrability]
  \label{def:a6}
  \mbox{}	\\
  A measurable function with separable range $f:X\to \mathfrak{B}$ is
  $\mm$-absolutely integrable iff $\lVert f(\,\cdot \,) \rVert_{\mathfrak{B}}^{}$ is
  $m$-integrable.
\end{definition}

In fact, any $\mm$-absolutely integrable function is also $\mm$-integrable.

\begin{proposition}[Integrability and absolute integrability]
  \label{prop:a2}
  \mbox{}	\\
  Let $f:X\to \mathfrak{B}$ be a $\mm$-absolutely integrable function. Then,
  $f$ is also $\mm$-integrable and, for all $S\in \Sigma$,
  \begin{equation}
    \lf\lVert \int_S^{}  \mathrm{d}\mm(x) \: f(x)  \ri\rVert_{\mathfrak{A}'}^{} \leq \int_S^{}  \mathrm{d}m(x) \lVert f(x)  \rVert_{\mathfrak{B}}^{}\; .
  \end{equation}
\end{proposition}
\begin{proof}
  The proof is completely analogous to the proof of \cref{prop:a4}. We omit
  it for the sake of brevity.
\end{proof}

\begin{corollary}[Integrability of bounded functions]
  \label{cor:a1}
  \mbox{}	\\
  Any function with separable range $f:X\to \mathfrak{B}$ such that $\lVert f (\cdot )
  \rVert_{\mathfrak{B}}^{}$ is $m$-a.e.\
  uniformly bounded is $\mm$-integrable.
\end{corollary}

We are now in a position to state two versions of the dominated convergence
theorem for operator-valued functions. The second, that makes crucial use of
absolute integrability, is the most convenient in our concrete
applications. Note that both results easily applies to the special case of
scalar functions discussed in the previous section.

\begin{thm}[Dominated convergence I \mbox{\cite[Theorem 6]{bartle1956sm}}]
  \label{thm:a1}
  \mbox{}	\\
  Let $\lf\{ f_n \ri\}_{n\in \mathbb{N}}$, $f_n:X\to \mathfrak{B}$ for all $n\in \mathbb{N}$, be a
  sequence of $\mm$-integrable operator-valued functions strongly converging
  $\mm$-a.e.\ to $f:X\to \mathfrak{B}$. If there exists a $\mm$-integrable
  operator-valued function $g$, such that for all $n\in \mathbb{N}$ and $S\in \Sigma$
  \begin{equation}
    \lf\lVert \int_S^{}  \mathrm{d}\mm(x) \: f_n(x)  \ri\rVert_{}^{}\leq \lf\lVert \int_S^{}  \mathrm{d}\mm(x) \: g(x)  \ri\rVert_{}^{}\; ,
  \end{equation}
  then, $f$ is $\mm$-integrable and for any $S\in \Sigma$
  \begin{equation}
    \int_S^{}  \mathrm{d}\mm(x)f(x) = \lim_{n\to \infty}\int_S^{}  \mathrm{d}\mm(x) f_n(x)\; .
  \end{equation}
\end{thm}

\begin{thm}[Dominated convergence II]
  \label{thm:a3}
  \mbox{}	\\
  Let $\lf\{ f_n \ri\}_{n\in \mathbb{N}}$, $f_n:X\to \mathfrak{B}$ for all $n\in \mathbb{N}$, be a
  sequence of operator-valued functions strongly converging $\mu$-a.e.\ to
  $f:X\to \mathfrak{B}$. If there exists a $m$-integrable function $G:X\to \mathbb{R}^+$
  such that $\mm$-a.e.
  \begin{equation}
    \lf\lVert f_n(x)  \ri\rVert_{\mathfrak{B}}^{}\leq G(x)\; ,
  \end{equation}
  then, for any $n\in \mathbb{N}$, $f_n, f$ are $\mm$-absolutely integrable, and
  \begin{equation}
    \int_S^{}  \mathrm{d}\mm(x) \: f(x) = \lim_{n\to \infty}\int_S^{}  \mathrm{d}\mm(x) \: f_n(x)\; .
  \end{equation}
\end{thm}
\begin{proof}
  By dominated convergence theorem for scalar measures and functions, applied
  to $m$ and $ \lf\{ \lVert f_n(\,\cdot \,) \rVert_{\mathfrak{B}}^{} \ri\}_{n\in \mathbb{N}}$,
  respectively, we get that the $\lVert f_n(\,\cdot \,) \rVert_{\mathfrak{B}}^{},\lVert f (\,\cdot \,)
  \rVert_{\mathfrak{B}}^{}$ are both $m$-integrable and
  therefore, by \cref{prop:a2}, it follows that $f_n,f$ are also $\mm$-integrable. Now, for
  any $S\in \Sigma$, again by \cref{prop:a2},
  \begin{equation*}
    \lf\lVert \int_S^{}  \mathrm{d}\mm(x) (f-f_n)(x) \ri\rVert_{\mathfrak{A}'}^{}\leq \int_S^{}  \mathrm{d}m(x) \lf\lVert (f-f_n)(x)  \ri\rVert_{\mathfrak{B}}^{}\; .
  \end{equation*}
  Hence by dominated convergence theorem for $m$, applied to the sequence of
  scalar functions $\lf\{\lf\lVert (f-f_n)(x) \ri\rVert_{\mathfrak{B}}\ri\}_{n\in \mathbb{N}}$, it
  follows that in the strong topology of $\mathfrak{A}'$,
  \begin{equation*}
    \int_S^{}  \mathrm{d}\mm(x)f(x) = \lim_{n\to \infty}\int_S^{}  \mathrm{d}\mm(x) f_n(x)\; .
  \end{equation*}
\end{proof}

\subsection{Integration of functions with values in unbounded operators}
\label{sec:funct-with-valu}

Let us restrict the attention, for this section, to the concrete case
$\mathfrak{A}=\mathscr{B}(L^2(\R^{dN}))$. In the applications described
above, it is sometimes necessary to integrate functions from some measurable
space $X$ to the unbounded operators on $L^2(\R^{dN})$ (albeit with a rather
explicit form). It is possible to define the integration of such functions
with respect to suitable generalized state-valued measures, as already
outlined in \cref{sec:minim-probl-gener}. Let us repeat here the argument for
the sake of completeness.

Let $\mathcal{T}>0$ be an operator on $L^2(\R^{dN})$, possibly unbounded. A
generalized state-valued measure is in the domain of $\mathcal{T}$ iff there
exists a generalized state-valued measure $\mathfrak{n}_{\mathcal{T}}$ such
that for all $\mathcal{B}\in \mathscr{B}(L^2(\R^{dN}))$, and for any $S\in \Sigma$,
\begin{equation*}
  \mathfrak{n}_{\mathcal{T}}(S)\bigl[\mathcal{T}^{-\frac{1}{2}}\mathcal{B}\mathcal{T}^{-\frac{1}{2}}\bigr]=\mathfrak{n}(S)\bigl[\mathcal{B}\bigr]\; .
\end{equation*}

Given a measure in the domain of $\mathcal{T}$, we can integrate functions
singular ``at most as $\mathcal{T}$''. Let $\mathcal{F}$ be a function from
$X$ to the (closed and densely defined) operators on $L^2(\R^{dN})$. Then
$\mathcal{F}$ is $\mathfrak{n}$-absolutely integrable, with $\mathfrak{n}$ in
the domain of $\mathcal{T}$, iff for $\mathfrak{n}$-a.e.\ $x\in X$:
\begin{itemize}
\item $\mathcal{T}^{-\frac{1}{2}}\mathcal{F}(x)\mathcal{T}^{-\frac{1}{2}}\in
  \mathscr{B}(L^2(\R^{dN}))$;  
\item $\mathcal{T}^{-\frac{1}{2}}\mathcal{F}(x)\mathcal{T}^{-\frac{1}{2}}$ is $\mathfrak{n}_{\mathcal{T}}$-absolutely integrable.
\end{itemize}

Given an absolutely integrable function, we can define the integral as
follows: for any $S\in \Sigma$,
\begin{equation*}
  \int_{S}^{}  \mathrm{d}\mathfrak{n}(x)\bigl[\mathcal{F}(x)\bigr]=\int_S^{}  \mathrm{d}\mathfrak{n}_{\mathcal{T}}(x) \lf[\mathcal{T}^{-\frac{1}{2}}\mathcal{F}(x)\mathcal{T}^{-\frac{1}{2}}\ri]\; .
\end{equation*}

\subsection{Two-Sided Integration}
\label{sec:two-sided-integr}

If $\mathfrak{A}$ is a two-sided ideal of $\mathfrak{B}$, we can give a
slight generalization of the operator-valued integration, to accommodate
integration of one function to the left and one function to the right of the
measure. We use the notations and definitions of
\cref{sec:integr-oper-valu}. Let $g,h:X\to \mathfrak{B}$ be two simple
functions,
\begin{equation*}
  g(x)=\sum_{j=1}^N c_j \one_{S_j}(x)\;, \qquad h(x)=\sum_{j=1}^M d_j \one_{T_j}(x) \; .
\end{equation*}
In addition, for any $B,C\in \mathfrak{B}$ and for any $\omega\in \mathfrak{A}'$, let
us define $B \circ \omega \circ C\in \mathfrak{A}'$ by
\begin{equation}
  \lf( B \circ \omega \circ C \ri)(\,\cdot \,):=\omega(B\,\cdot \,C)\; .
\end{equation}
Hence, it is possible to define two-sided simple integration as
\begin{equation}
  \int_{X}^{} g(x) \: \mathrm{d}\mm(x) \: h(x) = \sum_{j=1}^N\sum_{k=1}^M c_j \circ \mu(S_j\cap T_k) \circ d_k\; .
\end{equation}
Moreover, if $f_1,f_2:X\to \mathfrak{B}$ have separable range, it is
straightforward to extend \cref{def:a5} to define the two-sided integral
\begin{equation}
  \int_S^{} f_1(x) \: \mathrm{d}\mm(x) f_2(x)\in \mathfrak{A}'\; .
\end{equation}
If the above integral exists, we say that the pair $f_1, f_2$ is
$\mm$-two-sided-integrable (the order is relevant). This notion also
preserves positivity: for all $f$ such that $f^{*}\,\cdot \,f$ is
$\mm$-two-sided-integrable, then
\begin{equation}
  \int_S^{} f^{*}(x) \: \mathrm{d}\mm(x) \: f(x)\in \mathfrak{A}'_+\; .
\end{equation}

A pair of functions with separable range $f_1,f_2:X\to \mathfrak{B}$ are
$\mm$-two-sided-absolutely integrable iff $\lVert f_1(\,\cdot \,) \rVert_{\mathfrak{B}}^{}\lVert
f_2(\,\cdot \,) \rVert_{\mathfrak{B}}^{}$ is $m$-integrable. The analogue of
\cref{prop:a2} is the following

\begin{proposition}[Integrability and absolute integrability]
  \label{prop:a5}
  \mbox{}	\\
  Let $f_1,f_2:X\to \mathfrak{B}$ be $\mm$-two-sided-absolutely integrable. Then, $f_1, f_2$ and $f_2, f_1$
  are both $\mm$-two-sided-integrable and, for all $S\in \Sigma$,
  \beq
  \lf\| \int_S^{}f_1(x) \:\mathrm{d}\mm(x) \: f_2(x) \ri\|_{\mathfrak{A}'}^{}\leq \int_S^{} \mathrm{d}m(x) \:\lVert f_1(x) \rVert_{\mathfrak{B}}^{}\lVert f_2(x) \rVert_{\mathfrak{B}}^{},
  \eeq
  with analogous bound when $ f_1 $ and $ f_2 $ are exchanged on the left hand side.
\end{proposition}

Finally, dominated convergence applies to two-sided integration too.

\begin{thm}[Dominated convergence III]
  \label{thm:a4}
  \mbox{}	\\
  Let $ \lf\{ f_n \ri\}_{n\in \mathbb{N}}, \lf\{ g_n \ri\}_{n\in \mathbb{N}}$, $f_n,g_n:X\to
  \mathfrak{B}$ for all $n\in \mathbb{N}$, be two sequences of
  operator-valued functions strongly converging $\mm$-a.e.\ to $f,g:X\to
  \mathfrak{B}$,
  respectively. If there exists a $m$-square-integrable function $G:X\to \mathbb{R}^+$
  such that $\mm$-a.e.
  \begin{equation}
    \lVert f_n(x)  \rVert_{\mathfrak{B}}^{}\leq G(x)\;,	\qquad \lVert g_n(x)  \rVert_{\mathfrak{B}}^{}\leq G(x)\; ,
  \end{equation}
  then, for any $n\in \mathbb{N}$, $f_n,g_n$ and $f,g$ are $\mm$-two-sided-absolutely
  integrable, and
  \beqn
  &\disp\int_S^{}f(x) \: \mathrm{d}\mm(x) \: g(x)=\lim_{n\to \infty} \disp\int_S^{}f_n(x) \: \mathrm{d}\mm(x) \: g_n(x)\; ;\\
  &\disp\int_S^{}g(x) \: \mathrm{d}\mm(x) \: f(x)=\lim_{n\to \infty} \disp\int_S^{}g_n(x) \:\mathrm{d}\mm(x) \: f_n(x)\; .
  \eeqn
\end{thm}

\subsection{Radon-Nikod\'ym Property and Push-forward}
\label{sec:radon-nikodym-prop}

If an operator-valued function does not have a separable range, it may fail
to have an approximation with simple functions. It is possible to give an
alternative definition of integration if $\mathfrak{A}'$ is a
\emph{separable} space, as it is the case for the trace class operators on a
separable Hilbert space $\mathscr{L}^1(\mathscr{K})$, thanks to the following
property.

\begin{thm}[Radon-Nikod\'ym property \mbox{\cite[Theorem
    2.1.0]{dunford1940tams}}]
  \label{thm:a6}
  \mbox{}	\\
  If $\mathfrak{A}'$ is separable, then it has the \emph{Radon-Nikod\'ym
    property}: for every algebraic state-valued measure $\mm$, there exists a
  function $\varrho:X\to \mathfrak{A}'_+$, which is $m$-Bochner-integrable and such
  that, for all $S\in \Sigma$,
  \begin{equation}
    \mm(S)=\int_S^{}  \mathrm{d}m(x) \: \varrho(x)\; .
  \end{equation}
  The function $\varrho$ is the \emph{Radon-Nikod\'ym derivative} of $\mm$ w.r.t.\
  $m$, denoted by $\varrho=\frac{\mathrm{d}\mm}{\mathrm{d}m}$.
\end{thm}
	
Therefore, it is natural to give the following alternative definition of
integrability. Recall that for any $\Gamma\in \mathfrak{A}'$, and $B\in \mathfrak{B}$
we define $ \lf( \Gamma \circ B \ri) (\,\cdot \,)=\Gamma(B\,\cdot \,)$, if $\mathfrak{A}$ is a left
ideal of $\mathfrak{B}$, and $ \lf( \Gamma \circ B \ri)(\, \cdot \,)=\Gamma(\,\cdot \,B)$, if
$\mathfrak{A}$ is a right ideal of $\mathfrak{B}$. If $\mathfrak{A}$ is a
two-sided ideal, the notation $\Gamma B$ denotes indifferently any of the two. In
this case, for any $B,C\in \mathfrak{B}$ we can define $ \lf( B \circ \Gamma \circ C \ri)
(\,\cdot \,)=\Gamma(B\,\cdot \,C)$.

\begin{definition}[Integrability IV]
  \label{def:a8}
  \mbox{}	\\
  Suppose that $\mathfrak{A}'$ is separable, and let $f,g:X\to \mathfrak{B}$ be
  measurable functions (possibly with non-separable range) and $\mm$ an
  algebraic state-valued measure with Radon-Nikod\'ym derivative
  $\varrho=\frac{\mathrm{d}\mm}{\mathrm{d}m}$. Then, $f$ is $\mm$-integrable iff $\varrho
  \circ f\in \mathfrak{A}'$ is $m$-Bochner-integrable and,
  for any $S\in \Sigma$,
  \begin{equation}
    \int_S^{}  \mathrm{d}\mm(x)f(x):=\int_S^{}  \mathrm{d}m(x) \: \varrho(x) \circ f(x)\in \mathfrak{A}'\; .
  \end{equation}
  If in addition $\mathfrak{A}$ is a two-sided ideal of $\mathfrak{B}$, then
  $f, g$ is $\mm$-two-sided-integrable iff $f\varrho g\in \mathfrak{A}'$ is
  $m$-Bochner-integrable, and, for any $S\in \Sigma$,
  \begin{equation}
    \int_S^{}  f(x) \: \mathrm{d}\mm(x) \: g(x):=\int_S^{}  \mathrm{d}m(x) \: f(x) \circ\varrho(x) \circ g(x)\in \mathfrak{A}'\; .
  \end{equation}
\end{definition}

It is straightforward to see that \cref{def:a8} is equivalent to
\cref{def:a5} and the analogous one for the two-sided integral for any $f,g$
with separable range, and therefore \cref{def:a8} extends \cref{def:a5} to
any separable $\mathfrak{A}'$. In addition, since $m$-Bochner-integrability
is equivalent to $\mm$-absolute integrability, it follows that, if
$\mathfrak{A}'$ is separable, then $\mm$-integrability is equivalent to
$\mm$-absolute-integrability. Hence, all the results of
\cref{sec:integr-scal-funct,sec:integr-oper-valu,sec:two-sided-integr}
extend, if $\mathfrak{A}'$ is separable, to functions with non-separable
range.

Suppose now that $X$ is a topological vector space and $\Sigma$ the corresponding
Borel $\sigma$-algebra. In this context, Bochner's theorem holds for algebraic
state-valued measures \cite{falconi2017arxiv}: the Fourier transform, with
$
\xi \in X' $,
\begin{equation}
  \widehat{\mm}(\xi) := \int_X^{} \mathrm{d}\mm(x) \:e^{2i\xi(x)} \in \mathfrak{A}'
\end{equation}
identifies uniquely a measure. Therefore, the push-forward of an algebraic
state-valued measure $\mm$ by means of a linear continuous map $\Phi:X\to Y$,
where $Y$ is again a topological vector space with the Borel $\sigma$-algebra, is
conveniently defined using the Fourier transform, and this definition
suffices for the purposes of this paper: more precisely, the push-forward
measure $\Phi\, \sharp \, \mm$ is the measure on $Y$ whose Fourier transform is defined
by, with $ \eta \in Y' $,
\begin{equation}
  \widehat{(\Phi\, _{\star} \, \mm)}(\eta):= \int_X^{} \mathrm{d}\mm(x) \: e^{2i\eta(\Phi(x))} \in \mathfrak{A}'\; .
\end{equation}


\newcommand{\etalchar}[1]{$^{#1}$}

\end{document}